\documentclass[aps,pra,groupedaddress,superscriptaddress,twocolumn,notitlepage,bibnotes]{revtex4-2}

\usepackage{cancel}

\usepackage{newpxtext,newpxmath}

\let\coloneqq\relax

\usepackage{amsthm}
\usepackage{amssymb}
\usepackage{amsmath}
\usepackage{bbold}
\usepackage{bbm}
\usepackage{nonfloat}
\usepackage[pdftex]{hyperref}
\usepackage{braket}
\usepackage{dsfont}
\usepackage{mathdots}
\usepackage{mathtools}
\usepackage{enumerate}
\usepackage[shortlabels]{enumitem}
\usepackage{csquotes}
\usepackage{stmaryrd}
\usepackage[cal=boondox]{mathalfa}
\usepackage{graphicx}
\usepackage{stackengine}
\usepackage{scalerel}
\usepackage{xr}
\usepackage{array}
\usepackage{makecell}
\newcolumntype{x}[1]{>{\centering\arraybackslash}p{#1}}
\usepackage{tikz}
\usepackage{pgfplots}
\usetikzlibrary{shapes.geometric, shapes.misc, positioning, arrows, arrows.meta, decorations.pathreplacing, decorations.pathmorphing, patterns, angles, quotes, calc}
\usepackage{booktabs}
\usepackage{xfrac}
\usepackage{siunitx}
\usepackage{centernot}
\usepackage{comment}
\usepackage{chngcntr}
\usepackage{physics}

\newcommand{\RN}[1]{%
  \textup{\uppercase\expandafter{\romannumeral#1}}%
}

\newtheorem{theorem}{Theorem}
\newtheorem{prop}[theorem]{Proposition}
\newtheorem{lemma}[theorem]{Lemma}

\newtheorem{cor}[theorem]{Corollary}
\newtheorem{cj}[theorem]{Conjecture}
\newtheorem{Def}[theorem]{Definition}

\newtheorem{remark}[theorem]{Remark}

\makeatletter
\def\thmhead@plain#1#2#3{%
  \thmname{#1}\thmnumber{\@ifnotempty{#1}{ }\@upn{#2}}%
  \thmnote{ {\the\thm@notefont#3}}}
\let\thmhead\thmhead@plain
\makeatother

\theoremstyle{definition}

\newcommand{\bb}{\begin{equation}\begin{aligned}\hspace{0pt}}
\newcommand{\bbb}{\begin{equation*}\begin{aligned}}
\newcommand{\ee}{\end{aligned}\end{equation}}
\newcommand{\eee}{\end{aligned}\end{equation*}}
\newcommand*{\coloneqq}{\mathrel{\vcenter{\baselineskip0.5ex \lineskiplimit0pt \hbox{\scriptsize.}\hbox{\scriptsize.}}} =}

\newcommand{\texteq}[1]{\stackrel{\mathclap{\scriptsize \mbox{#1}}}{=}}
\renewcommand{\textleq}[1]{\stackrel{\mathclap{\scriptsize \mbox{#1}}}{\leq}}

\renewcommand{\textgeq}[1]{\stackrel{\mathclap{\scriptsize \mbox{#1}}}{\geq}}

\newcommand{\sumno}{\sum\nolimits}

\newcommand{\N}{\mathbb{N}}

\DeclareMathAlphabet{\pazocal}{OMS}{zplm}{m}{n}

\DeclareMathOperator{\Id}{id}

\newcommand{\HH}{\pazocal{H}}

\newcommand{\D}{\pazocal{D}}

\newcommand{\NN}{\pazocal{N}}

\newcommand{\lsmatrix}{\left(\begin{smallmatrix}}
\newcommand{\rsmatrix}{\end{smallmatrix}\right)}

\stackMath

\stackMath

\makeatletter
\newcommand*\rel@kern[1]{\kern#1\dimexpr\macc@kerna}
\newcommand*\widebar[1]{%
  \begingroup
  \def\mathaccent##1##2{%
    \rel@kern{0.8}%
    \overline{\rel@kern{-0.8}\macc@nucleus\rel@kern{0.2}}%
    \rel@kern{-0.2}%
  }%
  \macc@depth\@ne
  \let\math@bgroup\@empty \let\math@egroup\macc@set@skewchar
  \mathsurround\z@ \frozen@everymath{\mathgroup\macc@group\relax}%
  \macc@set@skewchar\relax
  \let\mathaccentV\macc@nested@a
  \macc@nested@a\relax111{#1}%
  \endgroup
}

\counterwithin*{equation}{part}
\counterwithin*{theorem}{part}
\counterwithin*{figure}{part}

\tikzset{meter/.append style={draw, inner sep=10, rectangle, font=\vphantom{A}, minimum width=30, line width=.8, path picture={\draw[black] ([shift={(.1,.3)}]path picture bounding box.south west) to[bend left=50] ([shift={(-.1,.3)}]path picture bounding box.south east);\draw[black,-latex] ([shift={(0,.1)}]path picture bounding box.south) -- ([shift={(.3,-.1)}]path picture bounding box.north);}}}
\tikzset{roundnode/.append style={circle, draw=black, fill=gray!20, thick, minimum size=10mm}}
\tikzset{squarenode/.style={rectangle, draw=black, fill=none, thick, minimum size=10mm}}

\definecolor{Blues5seq1}{RGB}{239,243,255}
\definecolor{Blues5seq2}{RGB}{189,215,231}
\definecolor{Blues5seq3}{RGB}{107,174,214}
\definecolor{Blues5seq4}{RGB}{49,130,189}
\definecolor{Blues5seq5}{RGB}{8,81,156}

\definecolor{Greens5seq1}{RGB}{237,248,233}
\definecolor{Greens5seq2}{RGB}{186,228,179}
\definecolor{Greens5seq3}{RGB}{116,196,118}
\definecolor{Greens5seq4}{RGB}{49,163,84}
\definecolor{Greens5seq5}{RGB}{0,109,44}

\definecolor{Reds5seq1}{RGB}{254,229,217}
\definecolor{Reds5seq2}{RGB}{252,174,145}
\definecolor{Reds5seq3}{RGB}{251,106,74}
\definecolor{Reds5seq4}{RGB}{222,45,38}
\definecolor{Reds5seq5}{RGB}{165,15,21}

\usepackage{pgfplots}

\newcommand{\FS}[1]{\textcolor{red}{FS:\quad #1}}
\usepackage{epigraph}
\usepackage{pgfplots}
\pgfplotsset{width=10cm,compat=1.9}
\usetikzlibrary{decorations.markings}
\newcommand*{\addFileDependency}[1]{
  \typeout{(#1)}
  \@addtofilelist{#1}
  \IfFileExists{#1}{}{\typeout{No file #1.}}
}
\makeatother

\usepackage[toc,page,header]{appendix}
\usepackage{minitoc}

\begin{document}

\author{Francesco Anna Mele}
\email{francesco.mele@sns.it}
\affiliation{NEST, Scuola Normale Superiore and Istituto Nanoscienze, Consiglio Nazionale delle Ricerche, Piazza dei Cavalieri 7, IT-56126 Pisa, Italy}

\author{Farzin Salek}
\email{farzin.salek@gmail.com}
\affiliation{Department of Mathematics, Technical University of Munich, Boltzmannstrasse 3, 85748 Garching, Germany}

\author{Vittorio Giovannetti}
\email{vittorio.giovannetti@sns.it}
\affiliation{NEST, Scuola Normale Superiore and Istituto Nanoscienze, Consiglio Nazionale delle Ricerche, Piazza dei Cavalieri 7, IT-56126 Pisa, Italy}

\author{Ludovico Lami}
\email{ludovico.lami@gmail.com}
\affiliation{QuSoft, Science Park 123, 1098 XG Amsterdam, the Netherlands}
\affiliation{Korteweg-de Vries Institute for Mathematics, University of Amsterdam, Science Park 105-107, 1098 XG Amsterdam, the Netherlands}
\affiliation{Institute for Theoretical Physics, University of Amsterdam, Science Park 904, 1098 XH Amsterdam, the Netherlands}


\title{Quantum communication on the bosonic loss-dephasing channel}

\begin{abstract}
Quantum optical systems are typically affected by two types of noise: photon loss and dephasing. Despite extensive research on each noise process individually, a comprehensive understanding of their combined effect is still lacking. A crucial problem lies in determining the values of loss and dephasing for which the resulting loss-dephasing channel is anti-degradable, implying the absence of codes capable of correcting its effect or, alternatively, capable of enabling quantum communication. A conjecture in [Quantum 6, 821 (2022)] suggested that the bosonic loss-dephasing channel is not anti-degradable if the loss is below $50\%$. In this paper we refute this conjecture, specifically proving that for any value of the loss, if the dephasing is above a critical value, then the bosonic loss-dephasing channel is anti-degradable.  While our result identifies a large parameter region where quantum communication is not possible, we also prove that if two-way classical communication is available, then quantum communication --- and thus quantum key distribution --- is always achievable, even for high values of loss and dephasing. 
\end{abstract}

\maketitle



\section{Introduction}
Quantum optical platforms are key elements of quantum technologies, contributing significantly to both quantum communication and quantum computation~\cite{Hayashi2017,BUCCO,Caves,weedbrook12,hardware-1,hardware-2,hardware-3,Ofek2016-mu,WANG20071}. Since the potential benefits of quantum technologies are hindered by the presence of decoherence~\cite{RevModPhys.76.1267}, the investigation of decoherence sources affecting bosonic systems and the development of bosonic quantum error-correcting codes have been extensively analysed in recent years~\cite{PhysRevA.99.032344,Terhal_2020, PhysRevA.97.032346, Denys2023tqutrittwomode, PhysRevX.10.011058,PhysRevA.101.012316,PhysRevA.106.022431,PhysRevA.107.032423,lower-bound,Die-Hard-2-PRL, Die-Hard-2-PRA, mele2023optical}. The primary noise processes in bosonic systems that act as dominant sources of decoherence are \emph{photon loss} and \emph{bosonic dephasing}~\cite{PhysRevX.6.031006, Brito_2008,1992ElL....28...53W}, which have been both extensively analysed~\cite{PhysRevA.97.032346,8482307,Wilde2014-fm,exact-solution}. Loss affects the system by causing it to dissipate some of its energy, whereas dephasing works to transform coherent superpositions into probabilistic mixtures. 
Although loss and dephasing sources can simultaneously affect bosonic systems~\cite{Campagne-Ibarcq2020-cb,PhysRevX.10.011058}, such as in superconducting systems~\cite{PhysRevB.94.014506,Rosenblum_2018}, the existing literature provides only partial results about their combined effect~\cite{Loss_dephasing0}. On a technical level, understanding the combined effect of loss and dephasing is challenging due to the conflicting behaviours they exhibit: the action of loss takes a simple form when written in the coherent state basis but is complicated to analyse in the Fock basis~\cite{BUCCO}, whereas dephasing demonstrates the opposite pattern, making the analysis of their combined effect quite intricate.

Consider an optical link (e.g.~an optical fibre or a free-space link) or a quantum memory affected by both loss and dephasing, where the link is used for quantum communication and the memory for quantum computation. A crucial challenge is to determine the conditions under which there exist protocols capable of enabling reliable quantum communication across the optical link or capable of mitigating the combined noise affecting the quantum memory. This problem is closely related to the anti-degradability condition in quantum Shannon theory: if a noise channel is \emph{anti-degradable}~\cite{MARK,Sumeet_book}, there are no quantum communication protocols for reliable information transmission or quantum error-correcting codes capable of overcoming it. Consequently, it is crucial to understand whether the combined effect of loss and dephasing results in an anti-degradable channel. This has been a puzzling problem, to the point that in~\cite{Loss_dephasing0} it was conjectured that the combined loss-dephasing noise does not result in an anti-degradable channel if the loss is below $50\%$.

In this paper we refute the above conjecture; specifically, we prove that for any value of the photon loss there exists a \emph{critical value} of the dephasing above which the resulting \emph{bosonic loss-dephasing channel} is anti-degradable. Our discovery thus identifies a large region of the loss-dephasing parameter space where correcting the noise and achieving reliable quantum communication is impossible. On the more positive side, however, we also prove that if the sender and the receiver are assisted by two-way classical communication, then reliable quantum communication --- and thus quantum key distribution --- is always possible, even in scenarios characterised by arbitrarily high levels of loss and dephasing. Ours are the first analytical results to characterise the transmission of quantum information in the presence of both loss and dephasing noise.

The structure of the paper is as follows. In Section~\ref{Sec_preliminaries}, we 
briefly review some preliminary notions necessary for stating our results. In Section~\ref{Sec_antideg}, we present our results on the anti-degradability of the bosonic loss-dephasing channel, refuting a conjecture put forth in~\cite{Loss_dephasing0}. In Section~\ref{Sec_two_way}, we analyse quantum communication across the bosonic loss-dephasing channel assisted by two-way classical communication. In Section~\ref{Sec_deg}, we discuss the degradability of the loss-dephasing channel. In the Appendix, we provide detailed derivations and additional results concerning the bosonic loss-dephasing channel.

\section{Preliminaries}\label{Sec_preliminaries}
The \emph{quantum capacity} $Q(\pazocal{N})$ of a quantum channel $\pazocal{N}$ quantifies the efficiency in transmitting qubits reliably across $\pazocal{N}$~\cite{MARK,Sumeet_book}. The condition $Q(\pazocal{N})=0$ implies that there exist neither reliable quantum communication protocols across $\pazocal{N}$ nor codes capable of correcting the errors induced by $\pazocal{N}$. Accordingly, if $\pazocal{N}$ is \emph{anti-degradable} its quantum capacity vanishes~\cite{MARK}. This underscores the significance of determining whether a channel is anti-degradable, as the noise associated with such a channel cannot be corrected. By definition, a channel $\pazocal{N}$ is anti-degradable if there exists a channel $\pazocal{A}$ --- called the \emph{anti-degrading map} --- such that $\pazocal{A}\circ\pazocal{N}^{\text{c}} = \pazocal{N}$, where $\pazocal{N}^{\text{c}}$ denotes a complementary channel of $\pazocal{N}$~\cite{MARK}. Conversely, a channel $\pazocal{N}$ is degradable if there exists a channel $\pazocal{D}$ such that $\pazocal{D}\circ\pazocal{N}=\pazocal{N}^{\text{c}}$. Degradable channels are theoretically important because their quantum capacity can be calculated as the single-letter coherent information of the channel~\cite{Devetak-Shor,MARK,Sumeet_book}.

The phenomenon of photon loss is mathematically modelled by the \emph{pure-loss channel} $\pazocal{E}_\lambda$~\cite{BUCCO,weedbrook12}, a single-mode continuous-variable channel that acts on the input state $\rho$ by mixing it with an environmental vacuum state in a beam splitter of transmissivity $\lambda\in[0,1]$: 
\bb
\pazocal{E}_\lambda(\rho)\coloneqq \Tr_E\left[U_\lambda\left(\rho_S\otimes\ketbra{0}_E\right)U_\lambda^\dagger\right]\,,
\ee
where $U_\lambda \coloneqq \exp\left[\arccos(\sqrt{\lambda})(\hat{a}^\dagger \hat{e}-\hat{a}\, \hat{e}^\dagger)\right]$ is the beam splitter unitary, $\hat{a}$ and $\hat{e}$ are the annihilation operators of the input system $S$ and of the environment $E$, and $\Tr_{E}$ is the partial trace w.r.t.~$E$. When a single photon is fed into $\pazocal{E}_\lambda$, it is transmitted to the output with probability $\lambda$, while it is lost to the environment with probability $1-\lambda$. More generally, if $n$ photons are fed into the channel, the output is given by the binomial probability mixture $\pazocal{E}_\lambda(\ketbra{n}) = \sum_{\ell=0}^n\binom{n}{l}(1-\lambda)^{\ell}\lambda^{n-\ell}\ketbra{n-\ell}$, where $\ket{n}$ denotes the Fock state with $n$ photons~\cite{BUCCO}. 
When $\lambda=1$ the pure-loss channel is noiseless, while when $\lambda=0$ it is completely noisy --- it maps any state into the vacuum. It is known that the pure-loss channel is anti-degradable for $\lambda\in [0, \frac{1}{2}]$ and degradable for $\lambda\in [\frac{1}{2}$,1]~\cite{Caruso_weak,LossyECEAC1,Caruso2006,Wolf2007}.

The phenomenon of bosonic dephasing is mathematically described by the \emph{bosonic dephasing channel} $\pazocal{D}_\gamma$~\cite{Loss_dephasing0,exact-solution,PhysRevA.102.042413}, which maps the state $\rho=\sum_{m,n=0}^\infty\rho_{mn}\ketbra{m}{n}$, written in the Fock basis, 
to
\bb
\pazocal{D}_\gamma(\rho)\coloneqq \sum_{m,n=0}^\infty\rho_{mn}e^{-\frac{\gamma}{2}(m-n)^2}\ketbra{m}{n}\,,
\ee
resulting in a reduction in magnitude of the off-diagonal elements.  When $\gamma=0$, the bosonic dephasing channel is noiseless. In contrast, when $\gamma\to\infty$, it completely annihilates all off-diagonal components of the input density matrix, reducing it to an incoherent probabilistic mixture of Fock states. Moreover, the bosonic dephasing channel is never anti-degradable and it is always degradable~\cite{PhysRevA.102.042413}.

Consider an optical system undergoing simultaneous loss and dephasing over a finite time interval. At each instant, the system is susceptible to both an infinitesimal pure-loss channel and an infinitesimal bosonic dephasing channel. Hence, the overall channel, which describes the simultaneous effect of loss and dephasing, results in a suitable composition of numerous concatenations between infinitesimal pure-loss and bosonic dephasing channels. However, given that (i) the pure-loss channel and the bosonic dephasing channel commute, $\pazocal{E}_\lambda\circ\pazocal{D}_\gamma=\pazocal{D}_\gamma\circ\pazocal{E}_\lambda$; (ii) the composition of pure-loss channels is a pure-loss channel, $\pazocal{E}_{\lambda_1}\circ \pazocal{E}_{\lambda_2}=\pazocal{E}_{\lambda_1\lambda_2}$; and (iii) the composition of bosonic dephasing channels is a bosonic dephasing channel, $\pazocal{D}_{\gamma_1}\circ \pazocal{D}_{\gamma_2}=\pazocal{D}_{\gamma_1+\gamma_2}$; it follows that the combined effect of loss and dephasing can be modelled by the composition 
\bb
\pazocal{N}_{\lambda,\gamma}\coloneqq \pazocal{E}_\lambda\circ\pazocal{D}_\gamma\,,
\ee
which we will refer to as the \emph{bosonic loss-dephasing channel}.

\section{Anti-degradability}\label{Sec_antideg}
Prior to this work, the only result on the anti-degradability of the bosonic loss-dephasing channel was that it is anti-degradable if the transmissivity is below $\frac{1}{2}$~\cite{Loss_dephasing0}. This result trivially follows from the anti-degradability of the pure-loss channel for transmissivities below $\frac{1}{2}$, and the fact that the composition of an anti-degradable channel with another channel inherits the property of being anti-degradable (see Lemma~\ref{lemma_comp_antideg} in the Appendix). Notably, in the regime $\lambda>\frac{1}{2}$, it was an open question to understand whether or not the bosonic loss-dephasing channel $\pazocal{N}_{\lambda,\gamma}$ is anti-degradable for some values of the dephasing $\gamma$, and in~\cite{Loss_dephasing0} the answer was conjectured to be negative. However, in the forthcoming Theorem~\ref{Main_Theorem}, we show that the latter conjecture is incorrect, specifically, we prove that for all $\lambda\in[0,1)$, if $\gamma$ is sufficiently large, then $\pazocal{N}_{\lambda,\gamma}$ becomes anti-degradable. 
\begin{theorem}\label{Main_Theorem}
The bosonic loss-dephasing channel $\pazocal{N}_{\lambda,\gamma}$ is anti-degradable if the transmissivity $\lambda$ and the dephasing $\gamma$ fall within one of the following regions: (i) $\lambda\in[0,\frac{1}{2}]$ and $\gamma\ge0$; (ii) $\lambda\in(\frac{1}{2},1)$ and $\gamma$ such that $\theta\!\left(e^{-\gamma/2},\sqrt{\frac{\lambda}{1-\lambda}}\right)\le\frac{3}{2}$, where $\theta(x,y)\coloneqq \sum_{n=0}^\infty x^{n^2}y^n$. A weaker but simpler sufficient condition that implies anti-degradability is given by $\lambda\le\max\!\left(\frac{1}{2},\frac{1}{1+9e^{-\gamma}}\right)$.
\end{theorem}
Here, we present a sketch of the proof, with a detailed version of the proof provided in Theorem~\ref{main-result} in the Appendix.
\begin{proof}[Proof sketch]
Any finite dimensional channel $\pazocal{N}$ is anti-degradable if and only if its Choi state is \emph{two-extendible}~\cite{Myhr2009}, meaning that there exists a tripartite state $\rho_{AB_1B_2}$ 
such that the reduced states on $AB_1$ and $AB_2$ both coincide with the Choi state: 
\bb\label{main_extension_def}
\Tr_{B_2}\left[\rho_{AB_1B_2} \right]&= C_{AB_1}\!(\pazocal{N})\,,\\
\Tr_{B_1}\left[\rho_{AB_1B_2} \right]&= C_{AB_2}\!(\pazocal{N}) \,,
\ee
where the Choi state is defined as $C_{AB}(\pazocal{N})\coloneqq \Id_A\otimes\pazocal{N}_{A'\to B}(\Phi_{AA'}) $, with $\Phi_{AA'}$ being the maximally entangled state. Such a characterisation extends to infinite dimension by considering the generalised Choi state $C^{(r)}_{AB}(\pazocal{N})\coloneqq \Id_A\otimes\pazocal{N}_{A'\to B}\big(\Psi^{(r)}_{AA'}\big)$~\cite{extendibility}, obtained by replacing $\Phi_{AA'}$ by the two-mode squeezed vacuum state $\Psi^{(r)}_{AA'}$ with squeezing parameter $r>0$~\cite{BUCCO}. The crux of our proof is to find a two-extension of $C^{(r)}_{AB}(\pazocal{N}_{\lambda,\gamma})$ in the region identified by condition~(ii). We do this in two steps. 

First, after scrutinising the matrix $C^{(r)}_{AB}(\pazocal{N}_{\lambda,\gamma})$ written in the Fock basis, we construct a tripartite state $\tau_{AB_1B_2}$ such that the reduced states on $AB_1$ and $AB_2$ have the same diagonal as $C^{(r)}_{AB}(\pazocal{N}_{\lambda,\gamma})$, and the same pattern of vanishing off-diagonal entries. The construction of this tripartite state involves applying several channels --- namely, beam splitter unitaries, squeezing unitary, partial trace, and a three mode controlled-add-add isometry --- to a 4-mode vacuum state.

The second step consists in transforming $\tau_{AB_1B_2}$ into a two-extension of $C^{(r)}_{AB}(\pazocal{N}_{\lambda,\gamma})$ by tweaking its off-diagonal entries. This is done by using the toolbox of \emph{Hadamard maps}~\cite{Sumeet_book}. For any matrix $A\coloneqq (a_{mn})_{m,n\in\N}$, the associated Hadamard map $H^{(A)}$ is defined by 
\bb\label{def_hadamard_channels_main}
    H^{(A)}(\ketbra{m}{n})=a_{mn}\ketbra{m}{n}
\ee
for all $m, n$~\cite{Sumeet_book}. In practice, $H^{(A)}$ acts on the input density matrix by multiplying each $(m,n)$ entry by the corresponding coefficient $a_{mn}$. 
Importantly, $H^{(A)}$ is a quantum channel if and only if $A$ is Hermitian, positive semi-definite, and has all $1$'s on the main diagonal~\cite{Sumeet_book}. The crucial observation is that it is always possible to find an infinite matrix $A_{\lambda,\gamma}$ (possibly not positive semi-definite), which is real, symmetric, and has all $1$'s on the main diagonal, such that the operator $\Id_A\otimes H^{(A_{\lambda,\gamma})}_{B_1}\otimes H^{(A_{\lambda,\gamma})}_{B_2}\left(\tau_{AB_1B_2}\right)$ coincides with $C^{(r)}_{AB}(\pazocal{N}_{\lambda,\gamma})$ when tracing out either $B_1$ or $B_2$. 

This however does not mean that we have found a two-extension of $C^{(r)}_{AB}$, because the above operator is not necessarily a state --- it may fail to be positive semi-definite. It \emph{is} a state, however, whenever $H^{(A_{\lambda,\gamma})}$ is a quantum channel, i.e.~when the infinite matrix $A_{\lambda,\gamma}$ is positive semi-definite, in formula $A_{\lambda,\gamma}\geq 0$. 
Therefore, a sufficient condition on the anti-degradability of $\pazocal{N}_{\lambda,\gamma}$ is that $A_{\lambda,\gamma} \geq 0$. 

The rest of the proof consists in showing that under condition~(ii) one indeed finds $A_{\lambda,\gamma}\geq 0$. This is not straightforward to check, because $A_{\lambda,\gamma}$ is an \emph{infinite} matrix, and it cannot be diagonalised analytically nor numerically. To by-pass this last hurdle we employ the theory of diagonally dominant matrices, and in particular the statement that if a matrix $A$ is such that $a_{nn} - \sum_{m:\,m\neq n} |a_{mn}|\geq 0$ for all $n$, then necessarily $A\geq 0$~\cite[Chapter 6]{HJ1}. 
We demonstrate that if $\lambda$ and $\gamma$ satisfy condition~(ii), then $A_{\lambda,\gamma}$ satisfies this condition, which establishes that $A_{\lambda,\gamma}\geq 0$ and hence concludes the proof. For a more detailed proof, please refer to Theorem~\ref{main-result} in the Appendix.
\end{proof}

Theorem~\ref{Main_Theorem} identifies a region of the parameter space $(\lambda,\gamma)$, with $\lambda$ identifying the transmissivity and $\gamma$ the dephasing, where the channel is anti-degradable and thus its quantum capacity vanishes, thereby implying the absence of viable error correcting codes for quantum data transfer and storage. 
This region is illustrated in Fig.~\ref{Fig_main}. Interestingly, Theorem~\ref{Main_Theorem} implies that even if $\lambda>\frac{1}{2}$ one can pick $\gamma$ large enough so that there exists an anti-degrading map achieving the transformation $\pazocal{N}_{\lambda,\gamma}^{\text{c}}(\ketbra{n}_F)\longrightarrow\pazocal{N}_{\lambda,\gamma}(\ketbra{n}_F)$, which can be expressed as (see Lemma~\ref{lemma_appendix_condition_antideg} in the Appendix)
\bb\label{remarkable_consequence}
\pazocal{E}_{1-\lambda}(\ketbra{n}_F)\otimes\ketbra{\sqrt{\gamma}n}_C\longrightarrow \pazocal{E}_\lambda(\ketbra{n}_F)\,,
\ee
where $\ket{n}_F$ denotes the $n$th Fock state and $\ket{\sqrt{\gamma}n}_C$ denotes a coherent state~\cite{BUCCO}. In Theorem~\ref{explicit-maps} in the Appendix, we provide an explicit construction of such anti-degrading map. This entails the following remarkable fact: for $\lambda>1/2$ and large enough $\gamma$ there exists an \emph{$n$-independent} strategy to convert the lossy Fock state $\pazocal{E}_{1-\lambda}(\ketbra{n}_F)$ into the less lossy Fock state $\pazocal{E}_{\lambda}(\ketbra{n}_F)$ using the coherent state $\ket{\sqrt{\gamma}n}_C$ as a resource. In other words, one can undo part of the loss on $\ket{n}_F$ if one has a coherent state that contains some information on $n$, sufficiently amplified so that that information is accessible enough. The nontrivial and somewhat surprising nature of this exact conversion strategy arises from the fact that the coherent states $\{\ket{\sqrt{\gamma}n}_C\}_{n\in\N}$ are not orthogonal, meaning that the strategy that consists in measuring the coherent state, guessing $n$, and re-preparing $\pazocal{E}_\lambda(\ketbra{n}_F)$ cannot succeed with probability $1$. 

Theorem~\ref{Main_Theorem} does not identify the entire anti-degradability region of $\pazocal{N}_{\lambda,\gamma}$, but only a subset of it. One way to improve this approximation is to determine numerically the region where the infinite matrix $A_{\lambda,\gamma}$ introduced in the proof sketch of Theorem~\ref{Main_Theorem} is positive semi-definite. 
In Fig.~\ref{Fig_main} we depict a numerical estimate of this region (see the crossed red part of the plot).


So far we have been concerned with inner approximations of the anti-degradability region. To obtain outer approximations, instead, one can start by observing that the action of $\pazocal{N}_{\lambda,\gamma}$ can only subtract and never add any photons. Mathematically, if the input state to $\pazocal{N}_{\lambda,\gamma}$ is supported on the span of the first $d$ Fock states, so is the output state. This restriction defines a qu$d$it-to-qu$d$it channel $\pazocal{N}_{\lambda,\gamma}^{(d)}$, analysing which can yield some insights into $\pazocal{N}_{\lambda,\gamma}$ itself. First, if $\pazocal{N}^{(d)}_{\lambda,\gamma}$ is \emph{not} anti-degradable then the same is true of $\pazocal{N}_{\lambda,\gamma}$ (see Corollary~\ref{cor-antideg} in the Appendix); secondly, as discussed above the anti-degradability of $\pazocal{N}^{(d)}_{\lambda,\gamma}$ is equivalent to the two-extendibility of the corresponding Choi state~\cite{Myhr2009}, and for moderate values of $d$ this latter condition can be efficiently checked numerically via \emph{semi-definite programming}~\cite{WATROUS, Sumeet_book}.  In this way, as discussed in Section~\ref{subsec_finite} of the Appendix, we can numerically determine a parameter region (see green region of Fig.~\ref{Fig_main}) where $\pazocal{N}_{\lambda,\gamma}$ is not anti-degradable.

Interestingly, already the qubit restriction $\pazocal{N}^{(2)}_{\lambda,\gamma}$, which coincides with the composition between the amplitude damping channel and the qubit dephasing channel~\cite{Sumeet_book}, yields the necessary condition $\lambda\le\frac{1}{1+e^{-\gamma}}$ on the anti-degradability of $\pazocal{N}_{\lambda,\gamma}$, as shown in the forthcoming Theorem~\ref{thm_main_not_antideg}. Based on the analysis of the qudit restrictions, we conjecture that, if $\gamma$ is sufficiently large, the latter condition $\lambda\le \frac{1}{1+e^{-\gamma}}$  is not only necessary but also sufficient. 
\begin{theorem}\label{thm_main_not_antideg}
If $\lambda> \frac{1}{1+e^{-\gamma}}$ then $\pazocal{N}_{\lambda,\gamma}$ is not anti-degradable.
\end{theorem}
\begin{proof}
A qubit channel $\pazocal{N}$ is anti-degradable if and only if 
\bb\label{condition_antideg_qubit}
\frac{1}{4}\Tr\!\left[\pazocal{N}\!\left(\mathbb{1}_2\right)^2\right]\ge \Tr\!\left[C(\pazocal{N})^2\right] -4\sqrt{\det[C(\pazocal{N})]},
\ee
where $C(\pazocal{N})$ is the Choi state~\cite{paddock2017characterization,Chen_2014,Myhr2009}. By employing the condition in \eqref{condition_antideg_qubit} together with the expression of the Choi state of the qubit restriction $\pazocal{N}^{(2)}_{\lambda,\gamma}$ provided in \eqref{choistate_bidimensional} in the Appendix, one can show that $\pazocal{N}^{(2)}_{\lambda,\gamma}$ is anti-degradable if and only if $\lambda\le \frac{1}{1+e^{-\gamma}}$. Hence, we conclude that the bosonic loss-dephasing channel $\pazocal{N}_{\lambda,\gamma}$ is not anti-degradable if $\lambda> \frac{1}{1+e^{-\gamma}}$.
\end{proof}
 We are interested in the anti-degradability of the loss-dephasing channel because it implies that the quantum capacity vanishes, entailing the impossibility of quantum communication. We now look at the complementary question: when is the quantum capacity $Q(\pazocal{N}_{\lambda,\gamma})$ strictly positive? A simple sufficient condition can be obtained by optimising the coherent information~\cite{Sumeet_book,MARK} of $\pazocal{N}_{\lambda,\gamma}$ over input states of the form $\rho_p\coloneqq p\ketbra{0}+(1-p)\ketbra{1}$. By doing so we identify a region of the $(\lambda,\gamma)$ parameter space
where $Q(\pazocal{N}_{\lambda,\gamma}) > 0$ (see the crossed green region in Fig.~\ref{Fig_main}). In this region quantum communication and quantum error correction become feasible.
\begin{figure}[!h]
\centering
\includegraphics[width=1\linewidth]{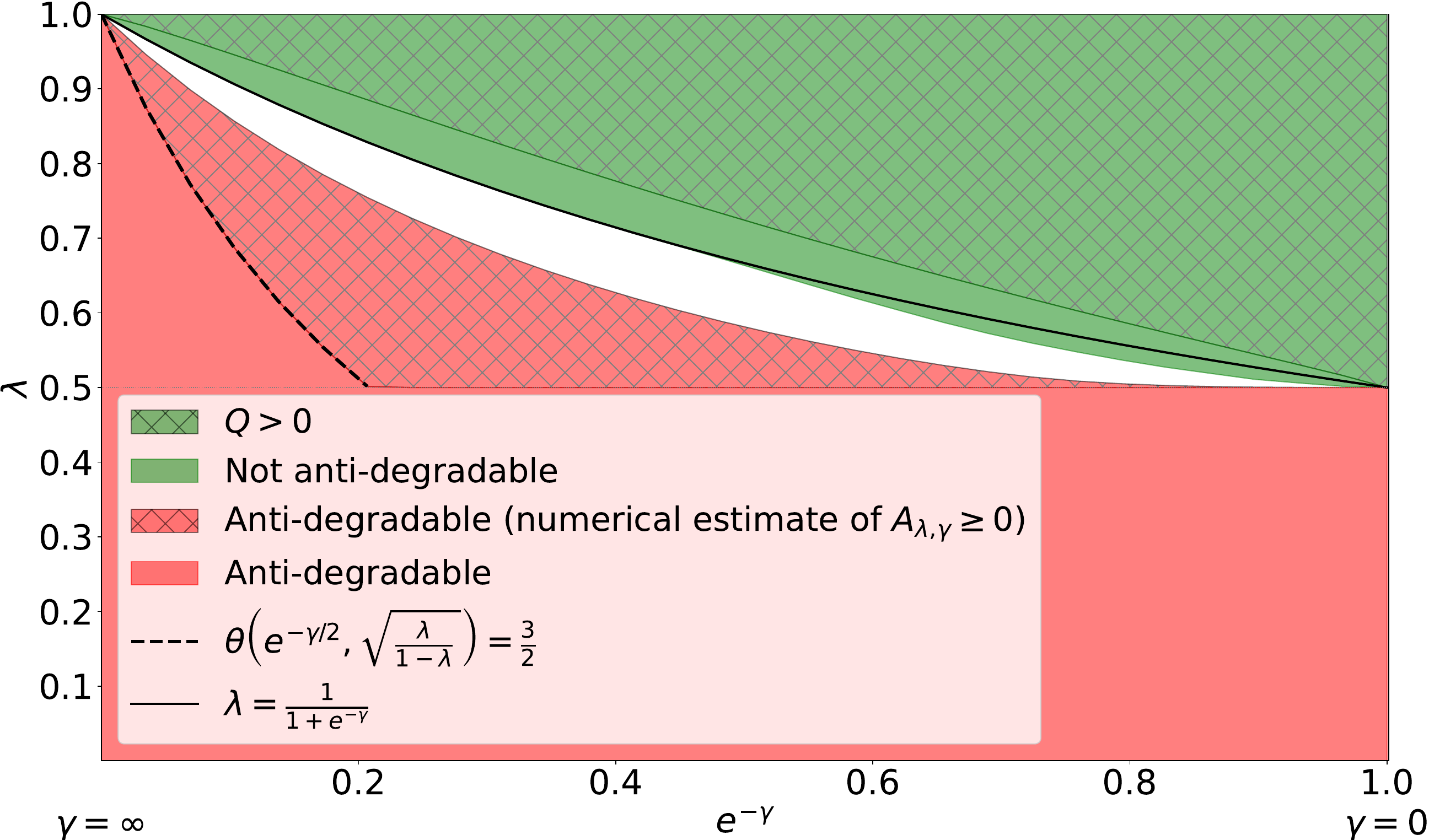}
\caption{Summary of results on the anti-degradability of the bosonic loss-dephasing channel $\pazocal{N}_{\lambda,\gamma}$.  The vertical axis represents the transmissivity $\lambda$, while the horizontal axis corresponds to $e^{-\gamma}$, where $\gamma$ is the dephasing parameter. In the green region $\pazocal{N}_{\lambda,\gamma}$ is not anti-degradable, while in the red region it is anti-degradable. In the crossed green region, the quantum capacity of $\pazocal{N}_{\lambda,\gamma}$ is strictly positive. The crossed red region is a numerical estimate of the region where the infinite matrix $A_{\lambda,\gamma}$ is positive semi-definite, a condition implying that $\pazocal{N}_{\lambda,\gamma}$ is anti-degradable, as explained in the proof sketch of Theorem~\ref{Main_Theorem}. Such an estimate can be obtained by examining the positive semi-definiteness of the $d\times d$ top-left corner of $A_{\lambda,\gamma}$ for large values of $d$ (here we employ $d=30$, but increasing $d$ already beyond $d\ge20$ yields no discernible change in the plot). The restriction $\pazocal{N}_{\lambda,\gamma}^{(6)}$ is anti-degradable if and only if $\lambda$ and $\gamma$ fall within the green region, and this is the reason why $\pazocal{N}_{\lambda,\gamma}$ is not anti-degradable in the green region. Below the curve $\theta\!\left(e^{-\gamma/2},\sqrt{\frac{\lambda}{1-\lambda}}\right)=\frac{3}{2}$, the channel $\pazocal{N}_{\lambda,\gamma}$ is anti-degradable, as stated in Theorem~\ref{Main_Theorem}. Above the curve $\lambda=\frac{1}{1+e^{-\gamma}}$, the channel is not anti-degradable, as guaranteed by Theorem~\ref{thm_main_not_antideg}. }

\label{Fig_main}
\end{figure}

\section{Two-way quantum communication}\label{Sec_two_way}
As we have just seen, (unassisted) quantum communication is not possible when the combined effects of loss and dephasing are too strong. However, in the forthcoming Theorem~\ref{Thm_main_two_way} we show that if Alice (the sender) and Bob (the receiver) have access to a \emph{two-way} classical communication line, then quantum communication, entanglement distribution, and quantum-key distribution~\cite{Ekert91} become again achievable for any value of loss and dephasing, even when Alice's input signals are constrained to have limited energy.

In this two-way communication setting the relevant notion of capacity is the \emph{two-way quantum capacity} $Q_2(\pazocal{N})$~\cite{MARK, Sumeet_book}, defined as the maximum achievable rate of qubits that can be reliably transmitted across $\pazocal{N}$ with the aid of two-way classical communication. Since in practice Alice has only a limited amount of energy to produce her input signals, one usually defines the so-called \emph{energy-constrained} two-way quantum capacity~\cite{Davis2018,8541091}, denoted as $Q_2(\pazocal{N},N_s)$. Here, $N_s$ denotes the mean photon number constraint at the input of the channel. 
\begin{theorem}\label{Thm_main_two_way}
For all $N_s>0$, $\lambda\in (0,1]$, and $\gamma\ge 0$, the energy-constrained two-way quantum capacity of the bosonic loss-dephasing channel is strictly positive, i.e.~$Q_2(\pazocal{N}_{\lambda,\gamma},N_s)>0$. In particular, $\pazocal{N}_{\lambda,\gamma}$ is not entanglement breaking. An explicit 
lower bound is 
\bb\label{lower_bound_q2_explicit}
    Q_2(\pazocal{N}_{\lambda,\gamma},N_s)\ge \frac{\lambda^N}{k}\left[\log_2\! \binom{N+k-1}{N} - S\!\left( \rho_{N,k,\gamma} \right) \right]
\ee
for any $N,k\in\N_+$ satisfying $\frac{N}{k}\le N_s$. Here, $S(\cdot)$ is the von Neumann entropy, $\rho_{N,k,\gamma}$ is a $\binom{N+k-1}{N}$-dimensional state defined by
    \bb
        \rho_{N,k,\gamma} &\coloneqq \binom{N+k-1}{N}^{-1}\sum_{p,q\in \Pi(N,k)}e^{-\frac{\gamma}{2}\|p-q\|_2^2}\ketbra{p}{q}\, ,
    \ee
    where $\Pi(N,k) \coloneqq\big\{p\in\N^k:\, \sumno_{i=1}^k p_i=N\big\}$ represents the set of partitions of a set of $N$ elements into $k$ parts, and the vectors $\{\ket{p}\}_{p\in\Pi(N,k)}$ are orthonormal.
\end{theorem}
Here, we provide a sketch of the proof, while in Theorem~\ref{thm_two_capacities_explicit} in the Appendix we provide a detailed proof.
\begin{proof}[Proof sketch] 
The proof exploits entanglement transmission protected by a particular error correction technique, \emph{rail encoding}. In a $k$-mode bosonic system, consider the subspace $V_{N,k}$ corresponding to a total photon number $N$. We can use this subspace, whose dimension is $d_{N,k} \coloneqq \dim V_{N,k} = \binom{N+k-1}{N}$, as an error correction code that protects against the detrimental action of $\pazocal{N}_{\lambda,\gamma}$. To this end, we prepare a maximally entangled state of dimension $d_{N,k}$ and send one share of it through $k$ copies of the channel $\NN_{\lambda,\gamma}$, one per mode. 
Since under the action of $\NN_{\lambda,\gamma}$ photons can only be lost and never added, and each photon has a probability $\lambda$ of being transmitted, the probability that an $N$-photon state will retain all of its photons at the output of the channel is exactly $\lambda^N$.  If this happens to be the case, which --- crucially --- can be certified by a total photon number measurement at the output, then the input state has been subjected to no loss and only dephasing. The entanglement of the resulting, maximally correlated state can be distilled via an explicit protocol known as the hashing protocol~\cite{devetak2005,Sumeet_book}, resulting in $\log_2 d_{N,k} - S(\rho_{N,k,\gamma}) > 0$ singlet (a.k.a.\ ebit, i.e.\ unit of entanglement) yield. 
The strict positivity of this yield follows by observing that $\rho_{N,k,\gamma}$ is a $d_{N,k}$-dimensional mixed state that is not maximally mixed.  
\end{proof}

\section{Degradability}\label{Sec_deg}
The bosonic loss-dephasing channel $\pazocal{N}_{\lambda,\gamma}$ is never degradable, except in the simple cases when either $\lambda=1$ or $\gamma=0$ and $\lambda\geq 1/2$~\cite{Loss_dephasing0}. This in turn implies that no single-letter formula for its quantum capacity is known outside of the anti-degradability region studied here, where the capacity vanishes. The failure of degradability has been demonstrated in~\cite{Loss_dephasing0} through a lengthy proof; we are now in position to provide an alternative, much simpler argument. The key ideas 
are as follows: (i)~If the qubit restriction $\pazocal{N}_{\lambda,\gamma}^{(2)}$ is not degradable, then $\pazocal{N}_{\lambda,\gamma}$ is not degradable either (see Corollary~\ref{cor-antideg} in the Appendix); and (ii)~if the rank of the Choi state of a qubit channel is greater or equal to $3$ than such channel is not degradable~\cite[Theorem 4]{Cubitt2008}. The result then follows by observing that the Choi state of the qubit channel $\pazocal{N}_{\lambda,\gamma}^{(2)}$, as provided in \eqref{choistate_bidimensional} in the Appendix, has a rank exactly equal to $3$. For a detailed proof see Theorem~\ref{thm_degradability} in the Appendix.

\section{Conclusion}
In this paper we have provided the first analytical investigation of the quantum communication capabilities of the bosonic loss-dephasing channel, a much more realistic model of noise than dephasing and loss treated separately. Refuting a conjecture put forth in~\cite{Loss_dephasing0}, we showed that the bosonic loss-dephasing channel is anti-degradable in a large region of the loss-dephasing parameter space, entailing that neither quantum communication nor quantum error correcting codes are possible in this region. On the positive side, we also showed that if two-way classical communication is suitably exploited, then quantum communication is always achievable, even in scenarios characterised by high levels of loss and dephasing, and even in the presence of stringent energy constraints.

Two fundamental technical innovations are key to our approach. First, a new method to analyse anti-degradability of bosonic channels, based on a two-stage construction of a symmetric extension of the Choi state. The introduction of this technique is crucial here also on the conceptual level, as \emph{all} other known tools to analyse quantum capacities (e.g.~degradability~\cite{Sumeet_book}, PPT-ness~\cite{incapacity}, or teleportation simulability~\cite{PLOB,Bennett-error-correction}) fail completely for the loss-dephasing channel~\cite{Loss_dephasing0}. In Section~\ref{envision_gen} in the Appendix, we envision that our technique could also be applied to other cases, e.g.\ to analyse the anti-degradability of the composition between the pure-loss channel and a \emph{general} bosonic dephasing channel. The second innovation we introduce is based on the use of rail encoding to investigate two-way assisted entanglement generation on the loss-dephasing channel. This technique, which we anticipate may be used to study general processes where photon loss is involved, has the additional benefit of yielding an explicit entanglement generation protocol.

Although the capacities of the dephasing channel and the pure-loss channel (separately) have been determined~\cite{exact-solution,PLOB,holwer, Caruso2006, Wolf2007, Mark2012, Mark-energy-constrained, Noh2019}, the capacities of the channel resulting from their combined action remain unknown. An intriguing open problem is to calculate or approximate these capacities.

\medskip
\textbf{\em Acknowledgements.}--- FAM and VG acknowledge financial support by MUR (Ministero dell'Istruzione, dell'Universit\`a e della Ricerca) through the following projects: PNRR MUR project PE0000023-NQSTI, PRIN 2017 Taming complexity via Quantum Strategies: a Hybrid Integrated Photonic approach (QUSHIP) Id.\ 2017SRN-BRK, and project PRO3 Quantum Pathfinder. FS is supported by Walter Benjamin Fellowship, DFG project No.\ 524058134. FAM and LL thank the Freie Universit\"{a}t Berlin for hospitality. FAM thanks the University of Amsterdam for hospitality.

\bibliographystyle{unsrt}
\bibliography{biblio}
\onecolumngrid
\appendix

\section{Preliminaries and notation}

\subsection{Quantum states and Channels}
In this subsection, we present a summary of the notation and fundamental properties used in the paper, drawing from the conventions established in standard quantum information theory textbooks~\cite{NC,MARK,WATROUS,Sumeet_book}.
Every quantum system is associated with a separable complex Hilbert space $\HH$ whose dimension is denoted by $|\HH|$. We use subscripts to denote the system associated to a Hilbert space and also systems on which the operators act. The composite quantum systems $A$ and $B$ exist within the tensor product of their individual Hilbert spaces $\pazocal{H}_A\otimes\pazocal{H}_B$ which is also denoted by $\HH_{AB}$.

We shall use $\mathbb{1}$ to denote the identity operator on $\HH$. The operator norm  of a linear operator $\Theta:\HH\to\HH$ is defined as 
\bb
    \|\Theta\|_{\infty}\coloneqq \sup_{\ket{\psi}\in\HH:\,\langle{\psi}|\psi\rangle =1}\sqrt{\langle\psi| \Theta^\dagger\Theta |\psi\rangle}\,.
\ee
An alternative (but equivalent) definition of the operator norm is as follows: 
\bb\label{def2_operator_norm}
\| \Theta \|_\infty \coloneqq \sup_{\ket{v},\ket{w}\in\HH,\, \braket{v}{v}=\braket{w}{w}=1} \left|\bra{v}\Theta\ket{w}\right|\,.
\ee
An operator is called bounded if its operator norm is bounded, i.e.~$\| \Theta \|_\infty<\infty$. A bounded operator $\Theta$ is positive semi-definite if 
$\bra{\psi}\Theta\ket{\psi}\geq 0,\forall\, \ket{\psi}\in \HH\,$,
while it is positive definite if 
$\bra{\psi}\Theta\ket{\psi}> 0,\,\forall\, \ket{\psi}\in \HH\,$. The trace norm of a linear operator $\Theta:\HH\to\HH$ is defined as
$\|\Theta\|_1\coloneqq \Tr\sqrt{\Theta^\dagger\Theta}$.
The set of trace class operators, denoted as $\pazocal{T}(\HH)$, is the set of all the linear operators on $\HH$ such that their trace norm is finite, i.e.~$\|\Theta\|_1<\infty$.
The operator and trace norm satisfy $\|\Theta\|_\infty\le \|\Theta\|_1$.
The set of quantum states (density operators), denoted as $\pazocal{P}(\HH)$, is the set of positive semi-definite trace class operators on $\HH$ with unit trace. The \emph{fidelity} between two quantum states $\rho,\sigma\in\pazocal{P}(\HH)$ is defined as $F(\rho,\sigma)\coloneqq \Tr[\sqrt{\sqrt{\rho}\sigma\sqrt{\rho}}]$.

A superoperator is a linear map between spaces of linear operators. The identity superoperator will be denoted as $\Id$. Quantum channels are completely-positive trace-preserving (cptp) superoperators. In this paper, we will use two different representations of a quantum channel that
are known as Stinespring and Choi--Jamio\l{}kowski representation. A quantum channel $\pazocal{N}_{A'\to B}$ can be represented in Stinespring representation as
\begin{align*}
      \pazocal{N}_{A'\to B}(\cdot )=\Tr_{E}\!\left[ U_{A'E\to BE}    (\cdot\otimes\ketbra{0}_E) U_{A'E\to BE} ^\dagger \right].
\end{align*}
Here, $E$ is an environment system, $\ket{0}_E$ is a pure state of the environment, and $U_{A'E\to BE} $ is an isometry that takes as input the systems $A'$ and $E$ and outputs the systems $B$, $E$. The associated complementary channel  $\pazocal{N}_{A'\to B}^\text{c}$ is given by
\begin{align*}
       \pazocal{N}^{\text{c}}_{A'\to E}(\cdot )=\Tr_{B}\!\left[ U_{A'E\to BE}   (\cdot\otimes\ketbra{0}_E) U_{A'E\to BE} ^\dagger \right].
\end{align*}
A channel $\pazocal{N}$ is called degradable if there exist a quantum channel $\pazocal{J}$, such that when is used after $\pazocal{N}$, we get back to the complementary channel $\pazocal{N}^{\text{c}}$, i.e.~$\pazocal{J}\circ\pazocal{N}=\pazocal{N}^{\text{c}}$. On the other hand, a channel is called anti-degradable if there is another quantum channel $\pazocal{A}$, such that using it  after the complementary channel, gives back the original channel,
i.e.~$\pazocal{A}\circ\pazocal{N}^\text{c}=\pazocal{N}$. The channels $\pazocal{J}$ and $\pazocal{A}$ are usually called the degrading map and the anti-degrading map of $\pazocal{N}$, respectively.

The Choi--Jamio\l{}kowski representation of the channel $\pazocal{N}_{A'\to B}$ is the operator $C(\pazocal{N})\in\pazocal{T}(\HH_{A}\otimes \HH_{B})$ that is defined as 
\bb\label{def_Choi}
    C(\pazocal{N})\coloneqq\Id_A\otimes \pazocal{N}_{A'\to B}(\ketbra{\Phi}_{AA'})\,,
\ee
where $\ket{\Phi}=\frac{1}{\sqrt{|\HH_A|}}\sum_{i=0}^{|\HH_A|-1}\ket{i}_{A}\otimes\ket{i}_{A'}$ is a maximally entangled state of schmidt rank $|\HH_A|$, the set of states $\{\ket{i}\}_{i=0,\ldots,|\HH_A|-1}$ forms a basis for $\HH_A$, and $\HH_{A}=\HH_{A'}$. It is a well-established fact that the superoperator $\pazocal{N}$ is a quantum channel if and only if $C(\pazocal{N}) \geq 0$ and $\Tr_BC(\pazocal{N}) = \mathbb{1}_A/|\HH_A|$~\cite{MARK}.

\begin{Def}\label{extendibility}
    A bipartite state $\rho_{AB}$ is symmetric two-extendible on $B$ if there exists a tripartite state $\tau_{AB_1B_2}$ such that
    \begin{itemize}
        \item $F_{B_1B_2}\tau_{AB_1B_2}F_{B_1B_2}^\dagger=\tau_{AB_1B_2}$
        \item $\Tr_{B_1}\tau_{AB_1B_2}=\rho_{AB}$,
    \end{itemize}
    where $B_1$ and $B_2$ are two copies of the system $B$, the operator $F_{B_1B_2}\coloneqq\sum_{i,j}\ketbra{i}{j}_{B_1}\otimes\ketbra{j}{i}_{B_2}$ denotes the swap unitary, and $\{\ket{i}_{B_1}\}_i$ and $\{\ket{i}_{B_2}\}_i$ form an orthonormal basis. The state $\tau_{AB_1B_2} $ is called a symmetric two-extension of $\rho_{AB}$ on $B$.
\end{Def}
\begin{Def}\label{two_extendibility_sm}
    A bipartite state $\rho_{AB}$ is called two-extendible on $B$ if there exists a tripartite state $\tau_{AB_1B_2}$ such that
    \bb
    \Tr_{B_1}\tau_{AB_1B_2}=\Tr_{B_2}\tau_{AB_1B_2}=\rho_{AB}\,,
    \ee
   where $B_1$ and $B_2$ are two copies of the system $B$.
\end{Def}

\begin{lemma}[\cite{Nowakowski_2016}]
A bipartite state $\rho_{AB}$ is two-extendible on $B$ if and only if it is symmetric two-extendible on $B$ .
 \end{lemma}
 \begin{proof}
      First, assume $\rho_{AB}$ is symmetric two-extendible on $B$. Since $F_{B_1B_2}\tau_{AB_1B_2}F_{B_1B_2}^\dagger=\tau_{AB_1B_2}$, it holds that $\Tr_{B_2}\tau_{AB_1B_2}=\Tr_{B_1}\tau_{AB_1B_2}$. This implies that $\rho_{AB}$ is two-extendible on $B$. Second, let $\rho_{AB}$ be two-extendible on $B$. One can easily check that the state $1/2(\tau_{AB_1B_2}+F_{B_1B_2}\tau_{AB_1B_2}F_{B_1B_2}^\dagger)$ is a symmetric two-extension of $\rho_{AB}$.
 \end{proof}

It has been demonstrated that a quantum channel is anti-degradable if and only if its Choi state is two-extendible on the output system~\cite{Myhr2009}. This equivalence leads to a simple necessary and sufficient condition for the anti-degradability of qubit channels:
\begin{lemma}\cite[Corollary 4]{paddock2017characterization}(See also~\cite{Chen_2014,Myhr2009})\label{lemma_qubit_charact_antideg}
    Any qubit quantum channel $\pazocal{N}$ is anti-degradable if and only if it satisfies
    \begin{align*}   \Tr\!\left[\left(\pazocal{N}\left(\frac{\mathbb{1}_2}{2}\right)\right)^2\right]\ge \Tr\!\left[(C(\pazocal{N}))^2\right] -4\sqrt{\det(C(\pazocal{N}))},
    \end{align*}
    where $\mathbb{1}_2$ denotes the identity operator on the qubit Hilbert space.
\end{lemma}

\medskip
\subsection{Bosonic systems}
In this subsection, we will provide an overview of relevant definitions and properties concerning quantum information with continuous variable systems; refer to~\cite{BUCCO} for detailed explanations. A single-mode of electromagnetic radiation with definite frequency and polarisation is represented by the Hilbert space $L^2(\mathbb{R})$, which comprises all square-integrable complex-valued functions over $\mathbb{R}$. Let $\mathbb{N}_+$ be the set of strictly positive integers and let $\mathbb{N}\coloneqq \{0\}\cup\mathbb{N}_+$. For any $n\in\mathbb{N}$, the construction of the \emph{Fock state} $\ket{n}$ (the quantum state with $n$ photons) involves the application of the bosonic creation operator $\hat{a}^\dagger$ to the \emph{vacuum state} $\ket{0}$: 
\bb
    \ket{n}\coloneqq \frac{1}{\sqrt{n!}}(\hat{a}^\dagger)^n\ket{0}\,.
\ee
The Fock states $\{\ket{n}\}_{n\in\N}$ form an orthonormal basis of $L^2(\mathbb{R})$. The bosonic annihilation operator $\hat{a}$ and creation operator $\hat{a}^\dagger$ satisfy the well-known canonical commutation relation $[\hat{a},\hat{a}^\dagger]=\mathbb{1}$.

Let $\mathbb{C}$ be the set of complex numbers.
For any $\alpha\in\mathbb{C}$, let $D(\alpha)\coloneqq e^{\alpha \hat{a}^\dagger-\alpha^\ast \hat{a}}$ be the displacement operator. A coherent state of parameter $\alpha$, denoted by $\ket{\alpha}$, is defined by applying the displacement operator $D(\alpha)$ to the vacuum state, i.e.~$\ket{\alpha}\coloneqq D(\alpha)\ket{0}$. The overlap between coherent states is given by 
\bb\label{overlap_coh}
\braket{\alpha}{\beta}=e^{-\frac{1}{2}(|\alpha|^2+|\beta|^2-2\alpha^{*}\beta)}\,.
\ee

Quantum channels acting on bosonic systems are sometimes called bosonic channels. Similar to finite-dimensional channels, bosonic channels admit a Choi--Jamio\l{}kowski representation, usually referred to as \emph{generalised} Choi--Jamio\l{}kowski representation~\cite{Holevo-CJ,Holevo-CJ-arXiv}. 
Consider isomorphic Hilbert spaces $\HH_A,\HH_{A'}$ which are possibly infinite dimensional. Let $\ket{\psi}_{AA'}$ be a pure state satisfying $\Tr_{A'}\!\left[\ketbra{\psi}_{AA'}\right]>0$ (See~\cite[Lemma 26]{9039682}). The generalised Choi state is constructed by applying the channel to the subsystem $A'$ of $\ket{\psi}_{A'A}$ (see Lemma~\ref{gen_choi_thm} in the Appendix): $\Id_{A}\otimes\pazocal{N}_{A'\to B}(\ketbra{\psi}_{AA'})$.
The construction of the generalised Choi state usually utilises the two-mode squeezed vacuum state with squeezing parameter $r>0$, defined as follows~\cite{BUCCO}:
\begin{align}\label{def_squeezed0}
    \ket{\psi(r)}_{AA'} \coloneqq \frac{1}{\cosh (r)}\sum_{n=0}^\infty \tanh^n (r) \ket{n}_{A}\ket{n}_{A'}\,.
\end{align}
The equivalence between anti-degradability of a channel and two-extendibility of its Choi state extends to the infinite-dimensional channels~\cite{extendibility}.
We provide a detailed proof of this equivalence in Lemma~\ref{lemma_infinite_extendibility} in the Appendix as it helps us in developing our intuition in inventing an explicit example of an anti-degrading map of the bosonic loss-dephasing channel.

\subsection{Hadamard maps}
\label{Hadamard-intro}
Let $A = (a_{mn})_{m,n\in\mathbb{N}}, a_{mn}\in\mathbb{C}$, be an infinite matrix of complex numbers. Consider the superoperator $H$ on $\pazocal{T}\left(L^2(\mathbb{R})\right)$, recognised as the \emph{Hadamard map}, whose action on rank one operator $\ketbra{m}{n}$ is defined as follows:
\begin{align*}
        H(\ketbra{m}{n})=a_{mn}\ketbra{m}{n},\quad\forall\,m,n\in\mathbb{N}\,.
\end{align*}
In Section~\ref{sec_hadamard_chann} in the Appendix, we provide an overview of relevant properties of Hadamard maps. In particular, by combining known results about Hadamard maps and matrix analysis, in Lemma~\ref{diag_dom_implies_hadamard} in the Appendix, we show that given an infinite matrix $A = (a_{mn})_{m,n\in\mathbb{N}}, a_{mn}\in\mathbb{C}$, the associated Hadamard map is a quantum channel if
\begin{itemize}
    \item $A$ is Hermitian
    \item $a_{nn}=1$, $\forall n\in\N$
    \item $A$ is diagonally dominant, i.e.~$\sum_{\substack{m=0\\  m\ne n}}^\infty|a_{mn}|\le 1,\,\forall\,n\in\N\,.$ 
\end{itemize}

\subsection{Beam splitter}
A beam splitter serves as a linear optical tool employed for creating quantum entanglement between two modes, referred to as the \emph{system} mode (denoted as $S$) and the \emph{environment} mode (denoted as $E$). A depiction of a beam splitter is reported in Fig~\ref{beam-splitter}.

\begin{Def}
Let $\HH_S,\HH_E\coloneqq L^2(\mathbb{R})$. Let $\hat{a}$ and $\hat{e}$ denote the annihilation operator of $\HH_S$ and $\HH_E$, respectively. 
For all $\lambda\in[0,1]$, the beam splitter unitary of transmissivity $\lambda$ is given by
\bb\label{def_beam_splitter}
    U_{\lambda}^{SE}&\coloneqq\exp\left[\arccos\sqrt{\lambda}\left(\hat{a}^\dagger \hat{e}-\hat{a}\, \hat{e}^\dagger\right)\right]\,.
\ee
\end{Def}
\begin{lemma}\label{trasf_annihilation}
    For all $\lambda\in[0,1]$, it holds that
    \bb
        (U_\lambda^{SE})^\dagger \hat{a}\, U_\lambda^{SE}&=\sqrt{\lambda}\,\hat{a}+\sqrt{1-\lambda}\,\hat{e}\,,\\
U_\lambda^{SE} \hat{a}\, (U_{\lambda}^{SE})^\dagger&=\sqrt{\lambda}\,\hat{a}-\sqrt{1-\lambda}\,\hat{e}\,,\\
(U_\lambda^{SE})^\dagger \hat{e}\, U_{\lambda}^{SE}&=-\sqrt{1-\lambda}\,\hat{a}+\sqrt{\lambda}\,\hat{e}\,,\\
U_\lambda^{SE} \hat{e}\, (U_{\lambda}^{SE})^\dagger&=\sqrt{1-\lambda}\,\hat{a}+\sqrt{\lambda}\,\hat{e}\,.
    \ee
\end{lemma}
\begin{proof}
These identities can be readily proved by applying the Baker-Campbell-Hausdorff formula. For an alternative proof see~\cite[Lemma A.2]{Die-Hard-2-PRA}.
\end{proof}
\begin{lemma}\label{lemma_beam_n0}
For all $\lambda\in[0,1]$ and all $n\in\mathbb{N}$, it holds that
\begin{align}
U^{SE}_\lambda\ket{n}_S\otimes\ket{0}_E&=
\sum_{l=0}^n(-1)^l\sqrt{\binom{n}{l}}\lambda^{\frac{n-l}{2}}(1-\lambda)^{\frac{l}{2}}\ket{n-l}_{S}\otimes\ket{l}_E\,,\label{formula_beam_spl1}\\
U^{SE}_\lambda\ket{0}_S\otimes\ket{n}_E&=
\sum_{l=0}^n \sqrt{\binom{n}{l}}(1-\lambda)^{\frac{l}{2}}\lambda^{\frac{n-l}{2}}\ket{l}_{S}\otimes\ket{n-l}_E\,.\label{formula_beam_spl2}
\end{align}
\end{lemma}
\begin{proof}
Thanks to Lemma~\eqref{trasf_annihilation}, we have that $U_{\lambda}^{SE} \hat{a} \left(U_{\lambda}^{SE}\right)^\dagger=\sqrt{\lambda}\hat{a} -\sqrt{1-\lambda}\hat{e}$. Consequently, it holds that
\bb
U_{\lambda}^{SE}\ket{n}_S\otimes\ket{0}_E&=\frac{1}{\sqrt{n!}}U_{\lambda}^{SE} (a^\dagger)^n\ket{0}_S\otimes\ket{0}_E
\\&=\frac{1}{\sqrt{n!}}\left(U_{\lambda}^{SE} a^\dagger \left(U_{\lambda}^{SE}\right)^\dagger\right)^n\ket{0}_S\otimes\ket{0}_E
\\&=\frac{1}{\sqrt{n!}}\left(\sqrt{\lambda}a^\dagger -\sqrt{1-\lambda}b^\dagger\right)^n\ket{0}_S\otimes \ket{0}_E  
\\&=\frac{1}{\sqrt{n!}}\sum_{l=0}^{n}(-1)^l \binom{n}{l}\lambda^{\frac{n-l}{2}}(1-\lambda)^{\frac{l}{2}}(a^\dagger)^{n-l}\ket{0}_S\otimes (b^\dagger)^{l}\ket{0}_E\\&=\sum_{l=0}^n(-1)^l\sqrt{\binom{n}{l}}\lambda^{\frac{n-l}{2}}(1-\lambda)^{\frac{l}{2}}\ket{n-l}_{S}\otimes\ket{l}_E\,.
\ee    
Analogously, one can show the validity of~\eqref{formula_beam_spl2} by exploiting $U_\lambda^{SE} \hat{e}\, (U_{\lambda}^{SE})^\dagger=\sqrt{1-\lambda}\,\hat{a}+\sqrt{\lambda}\,\hat{e}$.
\end{proof}
\begin{remark}
    It is easily seen that Eq.~\eqref{formula_beam_spl1} is equivalent to
    \begin{align*}
        U^{SE}_\lambda\ket{n}_S\otimes\ket{0}_E&=
(-1)^n\sum_{l=0}^n(-1)^l\sqrt{\binom{n}{l}}\lambda^{\frac{l}{2}}(1-\lambda)^{\frac{n-l}{2}}\ket{l}_{S}\otimes\ket{n-l}_E\,.
    \end{align*}
\end{remark}

\subsection{Pure-loss channel}
In optical platforms, the most common source of noise is photon loss, which is modelled by the \emph{pure-loss channel}~\cite{BUCCO}. For any $\lambda\in[0,1]$, the pure-loss channel $\pazocal{E}_\lambda$ of transmissivity $\lambda$ is a single-mode bosonic channel which acts on the input system by mixing it in a beam splitter of transmissivity $\lambda$ with an environmental vacuum state; see Fig~\ref{beam-splitter}. In this model, the input signal is partially transmitted and partially reflected by the beam splitter, representing the loss of photons (or energy) in the channel. When $\lambda=1$, the pure-loss channel is noiseless (it equals the identity superoperator). Conversely, when $\lambda=0$, the pure-loss channel is completely noisy (specifically, it is a constant channel that maps any state in $\ketbra{0}$).

\begin{figure}[!h]
\centering
\includegraphics[width=0.3\linewidth]{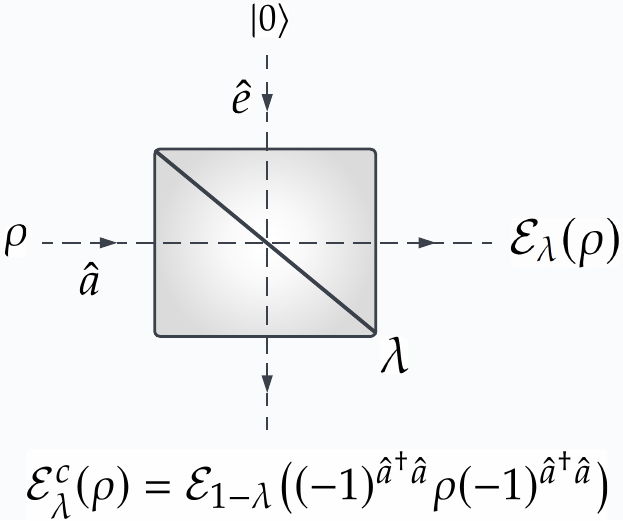}
\caption{ Stinespring representation of the pure-loss channel $\pazocal{E}_\lambda$. The pure-loss channel $\pazocal{E}_\lambda$ acts on the input state $\rho$ by mixing it in a beam splitter of transmissivity $\lambda$ (represented by the grey box) with an environmental vacuum state $\ket{0}$. Moreover, $\hat{a}$ and $\hat{e}$ are the annihilation operators of the input mode and the environmental mode, respectively. The complementary channel of the pure-loss channel is given by $\pazocal{E}^\text{c}_\lambda(\rho)=\pazocal{E}_{1-\lambda}\big( (-1)^{\hat{a}^\dagger \hat{a}}\rho  (-1)^{\hat{a}^\dagger \hat{a}}  \big)$.
 }
\label{beam-splitter}
\end{figure}

\begin{Def}
    Let $\HH_S,\HH_E\coloneqq L^2(\mathbb{R})$. For all $\lambda\in[0,1]$, the pure-loss channel of transmissivity $\lambda$ is a quantum channel $\pazocal{E}_{\lambda}:\pazocal{T}(\HH_S)\to \pazocal{T}(\HH_S)$ defined as follows: 
\begin{align*}
        \pazocal{E}_{\lambda}(\rho)&\coloneqq\Tr_E\left[U_\lambda^{SE} \big(\rho_S \otimes\ketbra{0}_E\big) (U_\lambda^{SE})^\dagger\right], \quad\forall\,\rho\in\pazocal{T}(\HH_S)\,,
\end{align*}
where $\ketbra{0}_E$ denotes the vacuum state of $\HH_E$ and $U_\lambda^{SE}$ denotes the beam splitter unitary defined in~\eqref{def_beam_splitter}.
\end{Def}

\begin{lemma}\label{fock_pureloss}
For all $\lambda\in[0,1]$ and all $n,m\in\mathbb{N}$, it holds that
\begin{align*}
    \pazocal{E}_\lambda(\ketbra{m}{n}) &= \sum_{\ell=0}^{\min(n,m)} \sqrt{\binom{m}{\ell}\binom{n}{\ell}}\, (1-\lambda)^\ell \lambda^{\frac{n+m}{2} - \ell} \ketbra{m-\ell}{n-\ell}\,.
\end{align*}

\end{lemma}
\begin{proof}
This is a direct consequence of~\eqref{formula_beam_spl1} and of the definition of pure-loss channel.
\end{proof}

\begin{lemma}
    For all $\lambda\in[0,1]$, a complementary channel $\pazocal{E}^\text{c}_\lambda:\pazocal{T}(\HH_S)\to \pazocal{T}(\HH_E)$ of the pure-loss channel $\pazocal{E}_\lambda: \pazocal{T}(\HH_S)\to \pazocal{T}(\HH_S)$ is given by 
\bb\label{compl_pure_loss_channel}
    \pazocal{E}^\text{c}_\lambda(\rho)\coloneqq\Tr_S\left[U_\lambda^{SE} \big(\rho_S \otimes\ketbra{0}_E\big) \left(U_\lambda^{SE}\right)^\dagger\right] = \pazocal{E}_{1-\lambda}\circ \pazocal{R}(\rho),\qquad\forall\,\rho\in\pazocal{T}(\HH_S)\,,
\ee
where $\pazocal{R}(\cdot)\coloneqq V\cdot V^\dagger$ with $V\coloneqq(-1)^{\hat{a}^\dagger \hat{a}}$. 
\end{lemma} 
\begin{proof}
    By linearity, it suffices to show the identity in~\eqref{compl_pure_loss_channel} for rank-one Fock operators of the form $\ketbra{m}{n}$ for any $n,m\in\N$, i.e.
    \bb
        \Tr_S\left[U_\lambda^{SE} \big(\ketbra{m}{n}_S \otimes\ketbra{0}_E\big) \left(U_\lambda^{SE}\right)^\dagger\right] = \pazocal{E}_{1-\lambda}\circ \pazocal{R}(\ketbra{m}{n})\,.
    \ee
    This follows directly from~\eqref{formula_beam_spl1}.
\end{proof}

\begin{prop}\cite{Wolf2007,Caruso2006}
    The pure-loss channel $\pazocal{E}_\lambda$ is degradable if and only if $\lambda\in[\frac{1}{2},1]$, and it is anti-degradable if and only if $\lambda\in[0,\frac{1}{2}]$.
\end{prop}

\begin{lemma}\cite[Lemma A.7]{Die-Hard-2-PRA}\label{lemma_comp_pure}
    For all $\lambda_1,\lambda_2\in[0,1]$ it holds that $\pazocal{E}_{\lambda_1}\circ\pazocal{E}_{\lambda_2}=\pazocal{E}_{\lambda_1\lambda_2}$.
\end{lemma}

\subsection{Bosonic dephasing channel}
Another main source of noise in optical platforms is bosonic dephasing, which serves as a prominent example of a non-Gaussian channel~\cite{PhysRevA.102.042413,exact-solution}.  

\begin{Def}
Let $\HH_S\coloneqq L^2(\mathbb{R})$ and let $\hat{a}$ be the corresponding annihilation operator. For all $\gamma\ge0$, the bosonic dephasing channel $\pazocal{D}_\gamma:\pazocal{T}(\HH_S)\to \pazocal{T}(\HH_S)$ is a quantum channel defined as follows: 
\begin{align*}
      \pazocal{D}_{\gamma}(\rho)\coloneqq \frac{1}{\sqrt{2\pi\gamma}}\int_{-\infty}^{\infty}\mathrm{d}\phi\, e^{-\frac{\phi^2}{2\gamma}}\,e^{i\phi \hat{a}^\dagger \hat{a}}\,\rho\, e^{-i\phi \hat{a}^\dagger \hat{a}},\quad\forall\,\rho\in\pazocal{T}(\HH_S)\,.
\end{align*}
In other words, $\pazocal{D}_\gamma$ is a convex combination of phase space rotations $\rho\longmapsto e^{i\phi \hat{a}^\dagger \hat{a}}\,\rho\, e^{-i\phi \hat{a}^\dagger \hat{a}}$, where the random variable $\phi$ follows a centered Gaussian distribution with a variance of $\gamma$.
\end{Def}

\begin{lemma}\label{lemma_app_deph}
For all $\gamma\ge0$ and all $n,m\in\mathbb{N}$, it holds that
\begin{align*}
        \pazocal{D}_{\gamma}(\ketbra{m}{n})= e^{-\frac{1}{2}\gamma(n-m)^2 }\ketbra{m}{n}\,.
\end{align*}
\end{lemma} 
\begin{proof}
This result follows from the Fourier transform of the Gaussian function:
\bb
    \frac{1}{\sqrt{2\pi\gamma}}\int_{-\infty}^{\infty}\mathrm{d}\phi\, e^{-\frac{\phi^2}{2\gamma}} e^{i\phi k} = e^{-\frac{1}{2}\gamma k^2},\quad \forall\,k\in\mathbb{R}\,.
\ee
\end{proof}
When $\gamma=0$, the bosonic dephasing channel is noiseless. In contrast, when $\gamma\to\infty$ it annihilates all off-diagonal components of the input density matrix, reducing it to an incoherent probabilistic mixture of Fock states. We now present a Stinespring dilation of the bosonic dephasing channel. 

\begin{lemma}\label{dephasing-stinespring}
Let $\HH_S=\HH_E\coloneqq L^2(\mathbb{R})$ and let $\hat{a}$ and $\hat{e}$ be annihilation operators on $\HH_S$ and $\HH_E$, respectively. For all $\gamma>0$, the bosonic dephasing channel $\pazocal{D}_\gamma:\pazocal{T}(\HH_S)\to \pazocal{T}(\HH_S)$ can be expressed in Stinespring representation as 
\bb\label{stinespring_dephasing}
    \pazocal{D}_{\gamma}(\rho)=\Tr_E\left[V_\gamma^{SE}\big(\rho_S\otimes \ketbra{0}_E\big)(V_\gamma^{SE})^\dagger \right],\qquad\forall\,\rho\in\pazocal{T}(\HH_S)\, ,
\ee
where $V_\gamma^{SE}$ denotes the conditional displacement unitary defined by
\bb\label{dephasing-unitary}
    V_\gamma^{SE}\coloneqq \exp\left[ \sqrt{\gamma}\, \hat{a}^\dagger \hat{a} \otimes\,(\hat{e}^\dagger - \hat{e})\right] =\sum_{n=0}^{\infty}\ketbra{n}_{S}\otimes D(\sqrt{\gamma}n).
\ee
 \end{lemma}
\begin{proof}
For any $n,m\in\N$ it holds that 
\bb\label{dephasing_stin_proof}
    &\Tr_E\left[V_\gamma^{SE}\big(\ketbra{m}{n}_S\otimes \ketbra{0}_E\big)(V_\gamma^{SE})^\dagger \right]\\&\texteq{(i)} \Tr_E\left[\ketbra{m}{n}_S\otimes \ketbra{\sqrt{\gamma}n}{\sqrt{\gamma}m}_E\right]\\&\texteq{(ii)}e^{-\frac{\gamma}{2}(n-m)^2}\ketbra{m}{n}\\&\texteq{(iii)}\pazocal{D}_\gamma(\ketbra{m}{n})\,.
\ee
Here, in (i) we used that $ V^{SE}_{\gamma}\ket{n}_S\otimes\ket{0}_E=\ket{n}_S\otimes \ket{\sqrt{\gamma}n}_E$,
 where $\ket{\sqrt{\gamma}n}_{E}$ denotes the coherent state with parameter $\sqrt{\gamma}n$.
In (ii), we exploited the formula for the overlap between coherent states provided in~\eqref{overlap_coh}, and
in (iii) we used Lemma~\ref{lemma_app_deph}. The proof is completed by linearity.
\end{proof}
\begin{remark}
   A different approach to dilating the bosonic dephasing channel, as outlined in the existing literature~\cite{PhysRevA.102.042413,Rexiti_2022,dehdashti2022quantum}, is as follows:
    \begin{align*}
        \tilde{V}_\gamma^{\text{SE}}=\exp\left[ -i\sqrt{\gamma}\, \hat{a}^\dagger \hat{a} \,(\hat{e}^\dagger + \hat{e})\right].
    \end{align*}
    This unitary transformation is achieved by rotating the environmental mode of the unitary operator $V_\gamma^{\text{SE}}$ in~\eqref{dephasing-unitary} by $\frac{\pi}{2}$, that is $ \tilde{V}_\gamma^{\text{SE}}= 
        e^{i\frac{\pi}{2}\hat{e}^\dagger\hat{e}}
        V_\gamma^{\text{SE}}
        e^{-i\frac{\pi}{2}\hat{e}^\dagger\hat{e}}.$ 
    These dilations yield the same dephasing channel, as all dilations are equivalent up to unitary transformations.
\end{remark}

\begin{prop}[\cite{Loss_dephasing0}]
    The bosonic dephasing channel $\pazocal{D}_\gamma$ is degradable for all $\gamma\ge0$.
\end{prop}
\begin{prop}[\cite{Loss_dephasing0}]
    The bosonic dephasing channel $\pazocal{D}_\gamma$ is never anti-degradable.
\end{prop}

\begin{lemma}\label{lemma_comp_deph}
For all $\gamma_1,\gamma_2\ge0$, the composition of bosonic dephasing channels is given by $\pazocal{D}_{\gamma_1} \circ \pazocal{D}_{\gamma_2} = \pazocal{D}_{\gamma_1 + \gamma_2}$.
\end{lemma}
    \begin{proof}
        This can be shown by leveraging Lemma~\ref{lemma_app_deph}. 
    \end{proof}

 \subsection{Bosonic loss-dephasing channel}\label{loss-dephasing-def}
 Consider an optical system undergoing simultaneous loss and dephasing over a finite time interval. At each instant, the system is susceptible to both an infinitesimal pure-loss channel and an infinitesimal bosonic dephasing channel. Hence, the overall channel describing the simultaneous effect of loss and dephasing results in a suitable composition of numerous concatenations between infinitesimal pure-loss and bosonic dephasing channels. However, given that:
 \begin{itemize}
     \item The pure-loss channel and the bosonic dephasing channel commute, i.e.~$\pazocal{E}_\lambda\circ\pazocal{D}_\gamma=\pazocal{D}_\gamma\circ\pazocal{E}_\lambda$,
as implied by Lemma~\ref{fock_pureloss} and Lemma~\ref{lemma_app_deph};
    \item The composition of pure-loss channels is a pure-loss channel, $\pazocal{E}_{\lambda_1}\circ \pazocal{E}_{\lambda_2}=\pazocal{E}_{\lambda_1\lambda_2}$ (Lemma~\ref{lemma_comp_pure});
    \item The composition of bosonic dephasing channels is a bosonic dephasing channel, $\pazocal{D}_{\gamma_1}\circ \pazocal{D}_{\gamma_2}=\pazocal{D}_{\gamma_1+\gamma_2}$ (Lemma~\ref{lemma_comp_deph});
 \end{itemize}  
 it follows that the combined effect of loss and dephasing can be modelled by the composition between pure-loss channel and bosonic dephasing channel,
\bb
    \pazocal{N}_{\lambda,\gamma}\coloneqq \pazocal{E}_\lambda\circ\pazocal{D}_\gamma\,,
\ee
dubbed the \emph{bosonic loss-dephasing channel}.
\begin{Def}
For all $\gamma\ge0$ and $\lambda\in[0,1]$, the {bosonic loss-dephasing channel} $\pazocal{N}_{\lambda,\gamma}$ is a quantum channel defined by the composition between pure-loss channel and bosonic dephasing channel: $\pazocal{N}_{\lambda,\gamma}\coloneqq\pazocal{D}_\gamma\circ \pazocal{E}_\lambda\,.$
\end{Def}
\begin{lemma}\label{App_remark_commutation}
For all $\lambda\in[0,1]$ and $\gamma\ge0$, it holds that $ \pazocal{N}_{\lambda,\gamma}\coloneqq\pazocal{D}_\gamma\circ \pazocal{E}_\lambda=\pazocal{E}_\lambda\circ\pazocal{D}_\gamma$. Moreover, for all $n,m\in\N$ it holds that
\begin{align*}
     \pazocal{N}_{\lambda,\gamma}(\ketbra{m}{n})&=e^{-\frac{1}{2}\gamma (n-m)^2}\pazocal{E}_{\lambda}(\ketbra{m}{n})\\&=e^{-\frac{1}{2}\gamma (n-m)^2}\sum_{\ell=0}^{\min(n,m)} \sqrt{\binom{n}{\ell}\binom{m}{\ell}}\, (1-\lambda)^\ell \lambda^{\frac{n+m}{2} - \ell} \ketbra{m-\ell}{n-\ell}\,.
\end{align*}
\end{lemma}
\begin{proof}
    It follows from Lemma~\ref{fock_pureloss} and Lemma~\ref{lemma_app_deph}.
\end{proof}

\begin{lemma}\label{lemma_compl_channel}
    Let $\HH_S,\HH_{E_1},\HH_{E_2}\coloneqq L^2(\mathbb{R})$. For all $\lambda\in[0,1]$ and all $\gamma\ge0$, the bosonic loss-dephasing channel $\pazocal{N}_{\lambda,\gamma}:\pazocal{T}(\HH_S)\to\pazocal{T}(\HH_S)$ admits the following Stinespring representation:
    \begin{align*}
                \pazocal{N}_{\lambda,\gamma}(\rho)=\Tr_{E_1E_2}\left[U_\lambda^{SE_1}V_\gamma^{SE_2} \big(\rho_S \otimes\ketbra{0}_{E_1}\otimes\ketbra{0}_{E_2}\big) \left(U_\lambda^{SE_1}V_\gamma^{SE_2}\right)^\dagger\right],
        \qquad\forall\,\rho\in\pazocal{T}(\HH_S),
    \end{align*}
    The associated complementary channel $\pazocal{N}_{\lambda,\gamma}^\text{c}:\pazocal{T}(\HH_S)\to \pazocal{T}(\HH_{E_1}\otimes \HH_{E_2})$ is defined as follows:
    \begin{align*}
                \pazocal{N}_{\lambda,\gamma}^\text{c}(\rho)\coloneqq\Tr_{S}\!\left[U_\lambda^{SE_1}V_\gamma^{SE_2} \big(\rho_S \otimes\ketbra{0}_{E_1}\otimes\ketbra{0}_{E_2}\big) \left(U_\lambda^{SE_1}V_\gamma^{SE_2}\right)^\dagger\right],
        \qquad\forall\,\rho\in\pazocal{T}(\HH_S)\,.
    \end{align*}
    In particular, 
    \begin{align}
        \label{action_compl_on_fock}
        \pazocal{N}^\text{c}_{\lambda,\gamma}(\ketbra{m}{n})=(-1)^{m-n}\pazocal{E}_{1-\lambda}\left(\ketbra{m}{n}_{E_1}\right)\otimes \ketbra{\sqrt{\gamma}m}{\sqrt{\gamma}n}_{E_2},
        \qquad\forall\,m,n\in\mathbb{N}\,,
    \end{align}
    where $\ket{\sqrt{\gamma}n}_{E_2}$ denotes the coherent state with parameter $\sqrt{\gamma}n$.  
\end{lemma}
\begin{proof}
    The Eq.~\eqref{action_compl_on_fock} 
    is derived by first applying~\eqref{compl_pure_loss_channel} and subsequently utilising the dilation of the bosonic dephasing channel. Finally, the non-environment mode of the bosonic dephasing channel is traced out.
\end{proof}

    \begin{lemma}\label{lemma_appendix_condition_antideg}
 Consider the Hilbert spaces $\HH_S$, $\HH_{E_1}$, and $\HH_{E_2}$, all isomorphic to $L^2(\mathbb{R})$. Let $\lambda\in[0,1]$ and $\gamma\ge0$. The bosonic loss-dephasing channel $\pazocal{N}_{\lambda,\gamma}:\pazocal{T}(\HH_S)\rightarrow\pazocal{T}(\HH_S)$ exhibits anti-degradability if and only if there exists a quantum channel $\pazocal{A}_{\lambda,\gamma}: \pazocal{T}(\HH_{E_{\text{out}}})\rightarrow \pazocal{T}(\HH_S)$ satisfying the following condition:
    \bb\label{antideg0_cond}
        \pazocal{A}_{\lambda,\gamma}\circ \pazocal{N}^{\text{c}}_{\lambda,\gamma}(\ketbra{m}{n})=\pazocal{N}_{\lambda,\gamma}(\ketbra{m}{n}),\quad\forall\,m,n\in\N\,,
    \ee
    where $\pazocal{N}^{\text{c}}_{\lambda,\gamma}:\pazocal{T}(\HH_S)\to \pazocal{T}(\HH_{E_{\text{out}}})$ denotes the complementary channel reported in~\eqref{action_compl_on_fock}, and $\HH_{E_{\text{out}}}\subset \HH_{E_1}\otimes\HH_{E_2}$ is defined in~\eqref{statement_def_h_out}. Moreover, the condition in~\eqref{antideg0_cond} is equivalent to
    \bb\label{cond_A_antideg_proof_st}
        \pazocal{A}_{\lambda,\gamma}\left(   \pazocal{E}_{1-\lambda}\left(\ketbra{m}{n}\right)\otimes \ketbra{\sqrt{\gamma}m}{\sqrt{\gamma}n}    \right) = (-1)^{m-n} e^{-\frac{\gamma}{2}(m-n)^2}\pazocal{E}_\lambda(\ketbra{m}{n}),
        \quad\forall\,m,n\in\mathbb{N}\,,
    \ee
    where $\ket{\sqrt{\gamma}n}$ denotes the coherent state with parameter $\sqrt{\gamma}n$.
\end{lemma}
\begin{proof}
$\pazocal{N}_{\lambda,\gamma}$ is anti-degradable if and only if there exists a quantum channel $\pazocal{A}_{\lambda,\gamma}: \pazocal{T}(\HH_{E_{\text{out}}})\to \pazocal{T}(\HH_S)$ such that 
    \bb\label{antideg_condition_general_rho}
        \pazocal{A}_{\lambda,\gamma}\circ \pazocal{N}^{\text{c}}_{\lambda,\gamma}(\rho)=\pazocal{N}_{\lambda,\gamma}(\rho),
        \qquad\forall\,\rho\in\pazocal{T}(\HH_S)\,.
    \ee    
        By linearity, it suffices to show the condition in~\eqref{antideg0_cond}, i.e.~which corresponds to the condition in~\eqref{antideg_condition_general_rho} restricted to rank-one Fock operators of the form $\rho=\ketbra{m}{n}$ with $m,n\in\N$. Moreover, by exploiting Lemma~\ref{App_remark_commutation} and~\ref{action_compl_on_fock}, the condition in~\eqref{antideg0_cond} is equivalent to~\eqref{cond_A_antideg_proof_st}.
\end{proof}

\begin{lemma}\label{comp_rule_bosonic_loss_deph}
    For all $\gamma_1,\gamma_2\ge0$ and all $\lambda_1,\lambda_2$ it holds that $\pazocal{N}_{\lambda_1,\gamma_1}\circ\pazocal{N}_{\lambda_2,\gamma_2}=\pazocal{N}_{\lambda_1\lambda_2\,,\,\gamma_1+\gamma_2}$.
\end{lemma}
    \begin{proof}
        This can be shown by exploiting Lemma~\ref{lemma_comp_pure} and Lemma~\ref{lemma_comp_deph}. 
    \end{proof}

\subsection{Preliminaries on capacities of quantum channels}\label{subsec_eb}
Quantum channels can be suitably exploited in order to transfer information from their input port to a possibly distant output port, a crucial task in quantum information theory~\cite{NC,WATROUS}. In particular, quantum Shannon theory~\cite{MARK,Sumeet_book} primarily helps us understand the fundamental limits of quantum communication using a quantum channel $\pazocal{N}$. These limits are called \emph{capacities} and tell us the ultimate amount of information we can send through the channel when we use it many times~\cite{MARK,Sumeet_book}. Different capacities are defined based on the type of information that is being sent down the channel. For example, classical and quantum capacities of a quantum channel correspond to its ultimate capability of transmission of classical and quantum information, respectively. A channel can also be used to generate secret bits and the relevant capacity in this context is the so-called secret-key capacity. Each of the above-mentioned capacities might be endowed with other resources such as initial shared entanglement between the sender and the receiver or (possibly interactive) classical communication over a noiseless but public channel. This latter scenario gives rise to the notion of two-way capacities.

Specifically, the \emph{quantum capacity} $Q(\pazocal{N})$ of a quantum channel $\pazocal{N}$ is the maximum rate at which qubits can be reliably transmitted through $\pazocal{N}$~\cite{Sumeet_book}. We can further assume that both the sender, Alice, and the receiver, Bob, have free access to a public, noiseless two-way classical channel. In this two-way communication setting, the relevant notion of capacities are the \emph{two-way quantum capacity} $Q_2(\pazocal{N})$ and the \emph{secret-key capacity} $K(\pazocal{N})$~\cite{MARK, Sumeet_book}, defined as the maximum achievable rate of qubits and secret-key bits, respectively, that can be reliably transmitted across $\pazocal{N}$ with the aid of two-way classical communication. Since an ebit (i.e.~a maximally entangled state of Schmidt rank $2$) can always be used to generate one bit of secret key~\cite{Ekert91}, a trivial bound relates these capacities: $Q_2(\pazocal{N})\leq K(\pazocal{N})$.

In practical scenarios, it is important to consider that the input state prepared by Alice can not have unlimited energy and it adheres to specific energy constraints. In bosonic systems, it is common to limit the average photon number of any input state $\rho$ as $\Tr[\hat{a}^\dagger\hat{a}\rho]\leq N_s$, where $N_s>0$ is a given energy constraint. For any $N_s>0$, the energy-constrained (EC) two-way capacities for transmitting qubits and secret-key bits, denoted as $Q_2(\pazocal{N},N_s)$ and $K(\pazocal{N},N_s)$ respectively, are defined similarly to the unconstrained capacities with the difference that now the optimisation is performed over those strategies that exploit input states that adhere to the specified energy constraint. As in the unconstrained scenario, the relation between EC two-way capacities continue to hold, i.e.~$Q_2(\pazocal{N},N_s)\leq K(\pazocal{N},N_s)$. Moreover, the unconstrained capacities are upper bounds for the corresponding energy constrained capacities and they become equal in the limit $N_s\to\infty$.

Two-way capacities of a quantum channel are closely related to another important information-processing task, namely, entanglement distillation over a quantum channel. Suppose Alice generates $n$ copies of a state $\rho_{A'A}$ and sends the $A'$ subsystems to Bob using the channel $\pazocal{N}$ for $n$ times. Now,  Alice and Bob share $n$ copies of the state $\rho'_{AB}\coloneqq \Id_A\otimes \pazocal{N}_{A'\to B}(\rho_{AA'})$. The task of an entanglement distillation protocol concerns identifying the largest number $m$ of ebits that can be extracted using $ n$ copies of $\rho'_{AB}$ via LOCC (Local Operations and Classical Communication) operations.  The rate of an entanglement distillation protocol is defined by the ratio $m/n$. The distillable entanglement $E_d(\rho_{AB}')\,$ of $\rho_{AB}'$ is defined as the maximum rate over all the possible entanglement distillation protocols~\cite{reviewEDP_dur}~\cite[Chapter 8]{Sumeet_book}. Note that without extra classical communication, entanglement distillation is not possible~\cite{salek2023new}. The following lemma establishes a link between the two-way quantum capacity, secret-key capacity, and distillable entanglement.
\begin{lemma}\label{lemma_Ed_Q2_link}
Let $\HH_{A},\HH_{A'},\HH_{B}\coloneqq L^2(\mathbb{R})$. Let $\pazocal{N}:\HH_{A'}\rightarrow\HH_{B}$ be a quantum channel. Let $N_s>0$ be the energy constaint, and let $\rho_{A'A}\in\pazocal{P}(\HH_A\otimes\HH_{A'})$ be a two-mode state satisfying the energy constraint $\Tr[(\hat{a}^\dagger \hat{a}\otimes \Id_A) \rho_{A'A}]\le N_s$, where $\hat{a}$ denotes the annihilation operator on $\HH_{A'}$. Then, it holds that
\bb\label{link_D2_D_main}
     K(\pazocal{N},N_s)\ge Q_2(\pazocal{N},N_s)\ge E_d\!\left(\Id_{A}\otimes\pazocal{N}(\rho_{AA'}) \right)\,,
\ee
where $K(\pazocal{N},N_s)$ denotes the energy-constrained secret-key capacity of $\pazocal{N}$, $Q_2(\pazocal{N},N_s)$ denotes the energy-constrained two-way quantum capacity of $\pazocal{N}$, and $E_d\!\left(\Id_{A}\otimes\pazocal{N}(\rho_{AA'}) \right)$ denotes the distillable entanglement of the state $\Id_{A}\otimes\pazocal{N}(\rho_{AA'})$.
\end{lemma}
The proof idea of the above lemma is the following. Suppose that Alice produces $n$ copies of a state $\rho_{AA'}$ such that the energy constraint is satisfied. Then, she can use the channel $n$ times to send all subsystems $A'$ to Bob. Then, Alice and Bob share $n$ copies of $\Id_{A}\otimes\pazocal{N}(\rho_{AA'})$, which can now be used to generate $\approx n\, E_d\!\left(\Id_{A}\otimes\pazocal{N}(\rho_{AA'}) \right) $ ebits by means of a suitable entanglement distillation protocol. The ebit rate of this protocol is thus $E_d\!\left(\Id_{A}\otimes\pazocal{N}(\rho_{AA'}) \right)$, which provides a lower bound on $Q_2(\pazocal{N},N_s)$ thanks to quantum teleportation~\cite{teleportation}. In addition, it holds that $K(\pazocal{N},N_s)\ge Q_2(\pazocal{N},N_s)$ because an ebit can generate a secret-key bit~\cite{Ekert91}. Consequently,~\eqref{link_D2_D_main} holds.

\section{Anti-degradability and Degradability of bosonic loss-dephasing channel}
This section is split into two parts based on the observation that if the input state to the bosonic loss-dephasing channel is chosen from a finite-dimensional subspace, the bosonic loss-dephasing channel effectively becomes a finite-dimensional channel, a fact we show in Lemma~\ref{qudit restriction-1}. This property allows us to apply established insights about the finite-dimensional channels to the bosonic loss-dephasing channel.
%
%
In subsection~\ref{subsec_infinite}, we present our study of the bosonic loss-dephasing channel when the input resides in the entire infinite-dimensional space, while subsection~\ref{subsec_finite} is dedicated to the findings resulting from analysis of finite-dimensional restrictions of the bosonic loss-dephasing channel.
\subsection{Sufficient condition on anti-degradability}~\label{subsec_infinite}
It is known that the bosonic dephasing channel $\pazocal{D}_\gamma$ is degradable across all dephasing parameter range $\gamma\ge 0$ and also it is never anti-degradable~\cite{PhysRevA.102.042413}. The pure-loss channel $\pazocal{E}_\lambda$ displays the peculiar characteristic of being anti-degradable for transmissivity values within the range $\lambda\in [0, \frac{1}{2}]$ and degradable for 
$\lambda\in [\frac{1}{2},1]$~\cite{Wolf2007,Caruso2006}. It turns out that when an anti-degradable channel is concatenated with another channel, the resulting channel inherits the property of being anti-degradable (see Lemma~\ref{lemma_comp_antideg}). This implies the following: if $\lambda\in[0,\frac{1}{2}]$ and $\gamma\ge0$, then the bosonic loss-dephasing channel $\pazocal{N}_{\lambda,\gamma}$ is anti-degradable~\cite{Loss_dephasing0}.
The authors of~\cite{Loss_dephasing0} left as an open question to understand whether or not the bosonic loss-dephasing channel $\pazocal{N}_{\lambda,\gamma}$ is anti-degradable in the region $\lambda\in(\frac{1}{2},1]$. In particular, they conjecture that $\pazocal{N}_{\lambda,\gamma}$ is not anti-degradable for all transmissivity values $\lambda\in(\frac{1}{2},1]$ and for all $\gamma\ge0$. In the following theorem, we refute this conjecture by explicitly finding values of $\lambda\in(\frac{1}{2},1]$ and $\gamma\ge0$ where the channel is anti-degradable. Our approach also yields an explicit expression for an anti-degrading map of the bosonic loss-dephasing channel.

\begin{theorem}\label{main-result}
Each of the following is a sufficient condition for the bosonic loss-dephasing channel $\pazocal{N}_{\lambda,\gamma}$ to exhibit anti-degradability: 
    \begin{enumerate}[label=(\roman*)]
        \item  $\lambda\in[0,\frac{1}{2}]$ and $\gamma\ge0$.
        \item  $\lambda\in(\frac{1}{2},1)$ and $\theta\!\left(e^{-\gamma/2},\sqrt{\frac{\lambda}{1-\lambda}}\right)\le\frac{3}{2}$, where $\theta$ is defined as
        $ \theta(x,y)\coloneqq \sum_{n=0}^\infty x^{n^2}y^n,\,\forall\,x,y\in[0,1).$
     \end{enumerate}
    In particular, $\pazocal{N}_{\lambda,\gamma}$ is anti-degradable if $\lambda\le\max\!\left(\frac{1}{2},\frac{1}{1+9e^{-\gamma}}\right)$. 
  \end{theorem}  
\begin{proof}\label{proof_main_thm1}
The proof of the sufficient condition~(i) follows directly from the observation that the composition of a pure-loss channel with transmissivity $\lambda\in[0,\frac{1}{2}]$ with any other channel inherits the anti-degradability from the pure-loss channel (see Lemma~\ref{lemma_comp_antideg}). 
    
The proof of the sufficient condition~(ii) is more involved, and it is the main technical contribution of our work. We rely on the equivalence between anti-degradability of a quantum channel and the two-extendibility of its Choi state~\cite{Myhr2009,extendibility}. To provide a comprehensive and intuitive understanding of this idea and to aid in the construction of an anti-degrading map of the bosonic loss-dephasing channel, we present this equivalence in Lemma~\ref{lemma_infinite_extendibility} within the Appendix.

Assume $\lambda\in(\frac{1}{2},1)$ and $\theta\!\left(e^{-\gamma/2},\sqrt{\frac{\lambda}{1-\lambda}}\right)\le\frac{3}{2}$.
Let $\HH_{A},\HH_B,\HH_{B_1},\HH_{B_2} = L^2(\mathbb{R})$ and suppose that the bosonic loss-dephasing channel $\pazocal{N}_{\lambda,\gamma}$ is a quantum channel from the system $A'$ to $B$. We want to show that the generalised Choi state of $\pazocal{N}_{\lambda,\gamma}$ is two-extendible on $B$. In other words, we want to show that there exists a tripartite state $\rho_{AB_1B_2}$ such that
\bb\label{extend_cond}
    \Tr_{B_2}\left[\rho_{AB_1B_2} \right]&= \tau_{AB_1}\,,\\
    \Tr_{B_1}\left[\rho_{AB_1B_2} \right]&= \tau_{AB_2}\,,
\ee
where $ \tau_{AB}\coloneqq \Id_{A}\otimes \pazocal{N}_{\lambda,\gamma}(\ketbra{\psi(r)}_{AA'})\,$ is the generalised Choi state of $\pazocal{N}_{\lambda,\gamma}$, and $\ket{\psi(r)}_{AA'}$ is the two-mode squeezed vacuum state with squeezing parameter $r>0$ defined in~\eqref{def_squeezed0}. By Lemma~\eqref{App_remark_commutation}, the generalised Choi state of the bosonic loss-dephasing channel can be expressed as follows:
\begin{align}\label{elements_choi_gen}
    &\tau_{AB}= \frac{1}{\cosh^2(r)}\sum_{m,n=0}^{\infty}\sum_{\ell=0}^{\min(m,n)}(\tanh(r))^{m+n}
    e^{-\frac{\gamma}{2}(m-n)^2}\sqrt{\pazocal{B}_\ell(m,1-\lambda)\,\pazocal{B}_\ell(n,1-\lambda)}\ketbra{m}{n}_{A}\otimes\ketbra{m-\ell}{n-\ell}_{B}\,,
\end{align}
where $\pazocal{B}_{\ell}\!\left(n,\lambda\right)\coloneqq \binom{n}{\ell}\lambda^{\ell}(1-\lambda)^{n-\ell}\,.$
We observe that $\tau_{AB}$ is a linear combination of $\ketbra{m}{n}\otimes\ketbra{j_1}{j_2}$, where $m,n,j_1,j_2\in\mathbb{N}$, $j_1\le n_1$, $j_2\le n$, and $m-n=j_1-j_2$. This insight leads to the following educated guess about the structure of a potential two-extension:
\begin{align}
    \label{educated_guess}
   \tilde{\rho}_{AB_1B_2}=\sum_{m,n=0}^\infty \sum_{\ell_1=0}^{m}\sum_{\ell_2=0}^{n}\sum_{k=0}^{\min(m-\ell_1,n-\ell_2)}c(m,n,\ell_1,\ell_2,k)\ketbra{m}{n}_A\otimes\ketbra{m-\ell_1}{n-\ell_2}_{B_1}\otimes\ketbra{\ell_1+k}{\ell_2+k}_{B_2}\,,
\end{align}
where $\{c(m,n,\ell_1,\ell_2,k)\}_{m,n,\ell_1,\ell_2,k}$ are some suitable coefficients. The soundness of this guess is confirmed by the fact that both $\Tr_{B_2}\left[\tilde{\rho}_{AB_1B_2} \right]$ and $\Tr_{B_1}\left[\tilde{\rho}_{AB_1B_2} \right]$ are linear combinations of $\ketbra{m}{n}\otimes\ketbra{j_1}{j_2}$ with $m,n,j_1,j_2\in\mathbb{N}$, $j_1\le m$, $j_2\le n$, and $m-n=j_1-j_2$, similar to the generalised Choi state $\tau_{AB}$ in~\eqref{elements_choi_gen}.

At this point, one could try to define the coefficients $\{c(m,n,\ell_1,\ell_2,k)\}_{m,n,\ell_1,\ell_2,k}$ so as to satisfy the required conditions of the extendibility. However, the resulting  tripartite operator may not qualify as a quantum state. In order to ensure that we obtain a quantum state, our approach consists in producing the operator $\tilde{\rho}_{AB_1B_2}$ via a physical process that consists in applying a sequence of quantum channels to a quantum state.
%
%
%
    
We begin by 
constructing a quantum state that has the same operator structure as the operator $\tilde{\rho}_{AB_1B_2}$ in~\eqref{educated_guess}. This means that at this initial stage, we only aim to build a tripartite state consisting of a linear combination of operators $\ketbra{m}{n}_A\otimes\ketbra{m-\ell_1}{n-\ell_2}_{B_1}\otimes\ketbra{\ell_1+k}{\ell_2+k}_{B_2}$ with the summation limits identical to those in~\eqref{educated_guess}. This ensures that the coefficients equal to zero coincide in the two operators. 

We now illustrate on each step of this construction. For a fixed $n\in\N$, consider the state $\ket{n}_{A}\otimes\ket{n}_{B_1}$. We introduce two auxiliary single-mode systems $B_2$ and $C$ initially in vacuum states: $\ket{n}_A\otimes\ket{n}_{B_1}\otimes\ket{0}_{B_2}\otimes\ket{0}_{C}$. We next send the systems $B_2$ and $B_1$ through the ports of a beam splitter, resulting in a superposition of $\ket{n}_A\otimes\ket{n-\ell}_{B_1}\otimes\ket{\ell}_{B_2}\otimes\ket{0}_{C}$ for $\ell=0,1,\ldots, n$, as implied by Lemma~\eqref{lemma_beam_n0}. We repeat this for systems $B_1$ and $C$, thus obtaining a superposition of $\ket{n}_A\otimes\ket{n-\ell-k}_{B_1}\otimes\ket{\ell}_{B_2}\otimes\ket{k}_{C}$, with $k=0,1,\ldots,n-\ell$ and $\ell=0,1,\ldots, n$. 
Consider now the isometry $W^{CB_1B_2}$ defined by 
\bb    W^{CB_1B_2}\ket{n}_{B_1}\otimes\ket{m}_{B_2}\otimes\ket{k}_C=\ket{n+k}_{B_1}\otimes\ket{m+k}_{B_2}\otimes\ket{k}_C,\quad\forall\,n,m,k\in\N\,,
\ee
dubbed \emph{controlled-add-add isometry} (mode $C$ is the control mode). By applying this isometry to the superposition we created by using beam splitters, we obtain a superposition of
$\ket{n}_A\otimes\ket{n-\ell}_{B_1}\otimes\ket{\ell+k}_{B_2}\otimes\ket{k}_{C}$ with $k=0,1,\ldots,n-\ell$ and $\ell=0,1,\ldots, n$.
Finally, by tracing out system $C$, we obtain the same operator structure of the operator $\tilde{\rho}_{AB_1B_2}$ in~\eqref{educated_guess}. Having focused on the operator structure, we have not considered the transmissivities of the two beam splitters so far. We will see that these transmissivities can be chosen carefully such that the diagonal elements of the operator at hand becomes equal to those of the Choi state.
 
We now apply the outlined construction to the two-mode squeezed vacuum state with two vacuum states appended to it, i.e.~$\ket{\psi(r)}_{AB_1}\otimes\ket{0}_{B_2}\otimes \ket{0}_{C}$. For reasons that will become clear in a moment, we choose the two beam splitter transmissivities to be $\lambda$ (for the beam splitter acting on $B_2B_1$) and $\frac{1-\lambda}{\lambda}$ (for the one acting on $CB_1$). By doing so we obtain the state 
\bb
      \ket{\phi}_{AB_1B_2C}&\coloneqq W^{C,B_1B_2}U_{\frac{1-\lambda}{\lambda}}^{CB_1}U_{\lambda}^{B_2B_1}\ket{\psi(r)}_{AB_1}\otimes\ket{0}_{B_2}\otimes \ket{0}_{C}\\
   &=\frac{1}{\cosh(r)}\sum_{n=0}^\infty\sum_{\ell=0}^n\sum_{k=0}^{n-\ell} \tanh^n(r)\sqrt{\pazocal{B}_\ell(n,1-\lambda)\,\pazocal{B}_k\!\left(n-\ell,\frac{2\lambda-1}{\lambda}\right)}\ket{n}_{A}\otimes\ket{n-\ell}_{B_1}\otimes\ket{\ell+k}_{B_2}\otimes\ket{k}_{C}\,,
\ee
where we used~\eqref{def_squeezed0} and Lemma~\ref{lemma_beam_n0}. Note that the transmissivities of the beam splitters $U_{\frac{1-\lambda}{\lambda}}^{CB_1}$ and $U_{\lambda}^{B_2B_1}$ are chosen such that the diagonal elements of $\Tr_{B_2C}\!\big[\ketbra{\phi}_{AB_1B_2C}\big]$ and $\Tr_{B_1C}\!\big[\ketbra{\phi}_{AB_1B_2C}\big]$ both coincide with those of the Choi state $\tau_{AB}$ in~\eqref{elements_choi_gen}. To verify this, let us calculate the state $\Tr_C\big[\ketbra{\phi}_{AB_1B_2C}\big]$:
\begin{equation*}
    \begin{aligned}
\Tr_C\big[\ketbra{\phi}_{AB_1B_2C}\big]=\frac{1}{\cosh^2(r)}&\sum_{m,n=0}^\infty\sum_{\ell_1=0}^{m}\sum_{\ell_2=0}^{n}\sum_{k=0}^{\min(m-\ell_1,n-\ell_2)}(\tanh(r))^{m+n}\sqrt{\pazocal{B}_{\ell_1}(m,1-\lambda)\,\pazocal{B}_{\ell_2}(n,1-\lambda)\,}\\
    &\hspace{-.9cm}\left[ \sqrt{\pazocal{B}_k\!\left(m-\ell_1,\frac{2\lambda-1}{\lambda}\right)\,\pazocal{B}_k\!\left(n-
    \label{tr_c_phi}\ell_2,\frac{2\lambda-1}{\lambda}\right)}\right]
    \ketbra{m}{n}_{A}\otimes\ketbra{m-\ell_1}{n-\ell_2}_{B_1}\otimes\ketbra{\ell_1+k}{\ell_2+k}_{B_2}. 
    \end{aligned}
\end{equation*}
Notably, the structure of the state $\Tr_C\big[\ketbra{\phi}_{AB_1B_2C}\big]$ mirrors that of~\eqref{educated_guess} with specific coefficients $c(m,n,\ell_1,\ell_2,k)$. Moreover, it holds that 
\begin{equation}
\begin{aligned} \label{eq1_tau0}
\Tr_{B_2C}\!\big[\ketbra{\phi}_{AB_1B_2C}\big]=\frac{1}{\cosh^2(r)}&\sum_{m,n=0}^\infty\sum_{\ell=0}^{\min(m,n)}(\tanh(r))^{m+n}\sqrt{\pazocal{B}_\ell(m,1-\lambda)\,\pazocal{B}_\ell(n,1-\lambda)} \\
&\qquad\left[\sum_{k=0}^{\min(m-\ell,n-\ell)}\sqrt{\,\pazocal{B}_k\!\left(m-\ell,\frac{2\lambda-1}{\lambda}\right)\,\pazocal{B}_k\!\left(n-\ell,\frac{2\lambda-1}{\lambda}\right)}\right]\ketbra{m}{n}_{A}\otimes\ketbra{m-\ell}{n-\ell}_{B_1}
\end{aligned}
\end{equation}
and that
\bb\label{equality_reduced_states}
\Tr_{B_1C}\!\big[\ketbra{\phi}_{AB_1B_2C}\big]=\Tr_{B_2C}\!\big[\ketbra{\phi}_{AB_1B_2C}\big]\,.
\ee
In order to prove \eqref{equality_reduced_states}, observe that
\bb\label{tedious_calculation}
    \Tr_{B_1C}\!\big[\ketbra{\phi}_{AB_1B_2C}\big]&=\frac{1}{\cosh^2(r)}\sum_{m,n=0}^\infty\sum_{\ell_1=\max(m-n,0)}^{m}\sum_{k=0}^{m-\ell_1}(\tanh(r))^{m+n}\sqrt{\pazocal{B}_{\ell_1}(m,1-\lambda)\,\pazocal{B}_{n-m+\ell_1}(n,1-\lambda)} \\
    &\qquad\qquad\pazocal{B}_k\!\left(m-\ell_1,\frac{2\lambda-1}{\lambda}\right)\ketbra{m}{n}_{A}\otimes\ketbra{\ell_1+k}{n-m+\ell_1+k}_{B_2}\\
    &=\frac{1}{\cosh^2(r)}\sum_{m,n=0}^\infty\sum_{\ell_1=\max(m-n,0)}^{m}\sum_{\ell=0}^{m-\ell_1}(\tanh(r))^{m+n}\sqrt{\pazocal{B}_{\ell_1}(m,1-\lambda)\,\pazocal{B}_{n-m+\ell_1}(n,1-\lambda)} \\
    &\qquad\qquad\pazocal{B}_{m-\ell_1-\ell}\!\left(m-\ell_1,\frac{2\lambda-1}{\lambda}\right)\ketbra{m}{n}_{A}\otimes\ketbra{m-\ell}{n-\ell}_{B_2}\\
    &=\frac{1}{\cosh^2(r)}\sum_{m,n=0}^\infty\sum_{\ell=0}^{\min(m,n)}(\tanh(r))^{m+n}\sum_{\ell_1=\max(m-n,0)}^{m-\ell}\sqrt{\pazocal{B}_{\ell_1}(m,1-\lambda)\,\pazocal{B}_{n-m+\ell_1}(n,1-\lambda)} \\
    &\qquad\qquad\pazocal{B}_{m-\ell_1-\ell}\!\left(m-\ell_1,\frac{2\lambda-1}{\lambda}\right)\ketbra{m}{n}_{A}\otimes\ketbra{m-\ell}{n-\ell}\\
    &=\frac{1}{\cosh^2(r)}\sum_{m,n=0}^\infty\sum_{\ell=0}^{\min(m,n)}(\tanh(r))^{m+n}\sum_{k=0}^{\min(n-\ell,m-\ell)}\sqrt{\pazocal{B}_{m-\ell-k}(m,1-\lambda)\,\pazocal{B}_{n-\ell-k}(n,1-\lambda)} \\
    &\qquad\qquad\pazocal{B}_{k}\!\left(k+\ell,\frac{2\lambda-1}{\lambda}\right)\ketbra{m}{n}_{A}\otimes\ketbra{m-\ell}{n-\ell}\\
    &\texteq{(i)}\frac{1}{\cosh^2(r)}\sum_{m,n=0}^\infty\sum_{\ell=0}^{\min(m,n)}(\tanh(r))^{m+n}\sqrt{\pazocal{B}_\ell(m,1-\lambda)\,\pazocal{B}_\ell(n,1-\lambda)} \\
    &\qquad\left[\sum_{k=0}^{\min(m-\ell,n-\ell)}\sqrt{\,\pazocal{B}_k\!\left(m-\ell,\frac{2\lambda-1}{\lambda}\right)\,\pazocal{B}_k\!\left(n-\ell,\frac{2\lambda-1}{\lambda}\right)}\right]\ketbra{m}{n}_{A}\otimes\ketbra{m-\ell}{n-\ell}_{B_2}\\
    &\texteq{(ii)}\Tr_{B_2C}\!\big[\ketbra{\phi}_{AB_1B_2C}\big]
\ee
Here, in (i), we used the identity
\bb
    &\sqrt{\pazocal{B}_{m-\ell-k}(m,1-\lambda)\,\pazocal{B}_{n-\ell-k}(n,1-\lambda)}\pazocal{B}_{k}\!\left(k+\ell,\frac{2\lambda-1}{\lambda}\right)\\&\qquad=\sqrt{\pazocal{B}_\ell(m,1-\lambda)\,\pazocal{B}_\ell(n,1-\lambda)}\sqrt{\,\pazocal{B}_k\!\left(m-\ell,\frac{2\lambda-1}{\lambda}\right)\,\pazocal{B}_k\!\left(n-\ell,\frac{2\lambda-1}{\lambda}\right)}\,,
\ee
which can be easily proved by substituting the definition $\pazocal{B}_{\ell}\!\left(n,\lambda\right)\coloneqq \binom{n}{\ell}\lambda^{\ell}(1-\lambda)^{n-\ell}$ and by leveraging the binomial identity
\bb
    \binom{n}{l+k}\binom{l+k}{k}=\binom{n}{l}\binom{n-l}{k}\,.
\ee Moreover, in (ii), we exploited \eqref{eq1_tau0}.

Note that the off-diagonal terms of the state in~\eqref{eq1_tau0} are not equal to those of the Choi state $\tau_{AB}$ in~\eqref{elements_choi_gen}. Specifically, the points of difference with the Choi state $\tau_{AB}$ are the presence of the term inside the square brackets and the absence of the dephasing exponent. To address these additional terms, let us use the toolbox of \emph{Hadamard maps}. Let $H$ be the Hadamard map, introduced in Sec.~\ref{Hadamard-intro}, associated with the infinite matrix $A\coloneqq (a_{mn})_{m,n\in\N}$ defined as follows: 
\begin{align}
    \label{def_a_element_matrix}
    a_{mn}\coloneqq\frac{e^{-\frac{\gamma}{2}(n-m)^2}}{\sum_{j=0}^{\min(n,m)} \sqrt{\,\pazocal{B}_j\!\left(n,\frac{2\lambda-1}{\lambda}\right)\,\pazocal{B}_j\!\left(m,\frac{2\lambda-1}{\lambda}\right)}},\quad\forall\,n,m\in\N\,.
\end{align}
By construction, we have that
\bb
&\Id_A\otimes H_{B_1}\left(\Tr_{B_2C}\!\big[\ketbra{\phi}_{AB_1B_2C}\big]\right) \\
&\qquad = \frac{1}{\cosh^2(r)} \sum_{m,n=0}^\infty\sum_{\ell=0}^{\min(m,n)}(\tanh(r))^{m+n}\sqrt{\pazocal{B}_\ell(m,1-\lambda)\,\pazocal{B}_\ell(n,1-\lambda)} \\
&\qquad\qquad \left[\sum_{k=0}^{\min(m-\ell,n-\ell)}\sqrt{\,\pazocal{B}_k\!\left(m-\ell,\frac{2\lambda-1}{\lambda}\right)\,\pazocal{B}_k\!\left(n-\ell,\frac{2\lambda-1}{\lambda}\right)}\right] a_{m-\ell,\,n-\ell}\,\ketbra{m}{n}_{A}\otimes\ketbra{m-\ell}{n-\ell}_{B_1} \\
&\qquad = \frac{1}{\cosh^2(r)} \sum_{m,n=0}^\infty\sum_{\ell=0}^{\min(m,n)}(\tanh(r))^{m+n} e^{-\frac{\gamma}{2}(n-m)^2}\, \sqrt{\pazocal{B}_\ell(m,1-\lambda)\,\pazocal{B}_\ell(n,1-\lambda)} \ketbra{m}{n}_{A}\otimes\ketbra{m-\ell}{n-\ell}_{B_1} \\
&\qquad = \tau_{AB_1}\,.
\ee
This means that the operator
\begin{align}
   \label{state_sym_ext}
    \rho_{AB_1B_2}\coloneqq \Id_A\otimes H_{B_1}\otimes H_{B_2}\left(  \Tr_{C}\left[\ketbra{\phi}_{AB_1B_2C}\right] \right) 
\end{align}
satisfies the extendibility conditions in~\eqref{extend_cond}. All that remains to prove is that $\rho_{AB_1B_2}$ is in fact a quantum state. 
We will do this by showing that the superoperator $H$ is a quantum channel. In Sec.~\ref{Hadamard-intro}, we establish that a Hadamard map is a quantum channel if its defining infinite matrix is Hermitian, has diagonal elements equal to one, and is diagonally dominant.
The first two properties are trivially satisfied by the infinite matrix $A$ defined in~\eqref{def_a_element_matrix}. We only need to demonstrate that for the parameter region $\lambda > \frac{1}{2}$ and $\theta\!\left(e^{-\gamma/2},\sqrt{\frac{\lambda}{1-\lambda}}\right)\le\frac{3}{2}$, the infinite matrix $A$ is diagonally dominant. We recall that, by definition, $A$ is diagonally dominant if it holds that $\sum_{\substack{m=0\\  m\ne n}}^\infty|a_{mn}|\le 1,\,\forall\,n\in\N$. Note that for any $n,m\in\N$ we have that 
\bb
    |a_{nm}|&=\frac{e^{ -\frac{\gamma}{2}(n-m)^2 }}{\sum_{j=0}^{\min(n,m)}\sqrt{\binom{n}{j}\binom{m}{j}\left(\frac{1-\lambda}{\lambda}\right)^{n+m-2j}\left(\frac{2\lambda-1}{\lambda}\right)^{2j}}}\\
    &\le \frac{e^{ -\frac{\gamma}{2}(n-m)^2 }}{\sum_{j=0}^{\min(n,m)}\sqrt{\binom{\min(n,m)}{j}^2\left(\frac{1-\lambda}{\lambda}\right)^{n+m-2j}\left(\frac{2\lambda-1}{\lambda}\right)^{2j}}} \\
    &=  e^{ -\frac{\gamma}{2}(n-m)^2 } \left(\frac{\lambda}{1-\lambda}\right)^{\frac{|n-m|}{2}} \,.
\ee
Consequently, if $\lambda$ and $\gamma$ are such that $\lambda>\frac{1}{2}$ and $\theta\!\left(e^{-\gamma/2},\sqrt{\frac{\lambda}{1-\lambda}}\right)\le\frac{3}{2}$, 
for any $n\in\N$ we have that
\begin{equation}
    \begin{aligned}
         \sum_{\substack{m=0\\  m\ne n}}^\infty|a_{mn}|&\le \sum_{\substack{m=0\\  m\ne n}}^\infty e^{ -\frac{\gamma}{2}(m-n)^2 } \left(\frac{\lambda}{1-\lambda}\right)^{\frac{|m-n|}{2}}\\&= \sum_{k=1}^{\infty}e^{ -\frac{\gamma}{2}k^2 } \left(\frac{\lambda}{1-\lambda}\right)^{\frac{k}{2}}+\sum_{k=1}^{n}e^{ -\frac{\gamma}{2}k^2 } \left(\frac{\lambda}{1-\lambda}\right)^{\frac{k}{2}}\\&\le 2\sum_{k=1}^{\infty}e^{ -\frac{\gamma}{2}k^2 } \left(\frac{\lambda}{1-\lambda}\right)^{\frac{k}{2}} \\&= 2\,\theta\!\left(e^{-\gamma/2},\sqrt{\frac{\lambda}{1-\lambda}}\right) - 2\\&\le 1\,.
    \end{aligned}
\end{equation}
Therefore, the infinite matrix $A$ is diagonally dominant if $\lambda>\frac{1}{2}$ and $\theta\!\left(e^{-\gamma/2},\sqrt{\frac{\lambda}{1-\lambda}}\right) \le\frac{3}{2}$. This establishes that $H_{B_1}$ and $H_{B_2}$ are valid quantum channels in this parameter range, implying that $\rho_{AB_1B_2}$ is a valid two-extension of the Choi state of $\pazocal{N}_{\lambda,\gamma}$, and in turn entailing that $\pazocal{N}_{\lambda,\gamma}$ is anti-degradable. Finally, note that if $\lambda>\frac{1}{2}$ the condition
\bb
\theta\!\left(e^{-\gamma/2},\sqrt{\frac{\lambda}{1-\lambda}}\right) \le\frac{3}{2}
\ee
is implied by $\lambda\le  \frac{1}{1+9e^{-\gamma}}$.
Indeed,
\begin{align*}
       \theta\!\left(e^{-\gamma/2},\sqrt{\frac{\lambda}{1-\lambda}}\right)\coloneqq\sum_{k=0}^{n}e^{ -\frac{\gamma}{2}k^2 } \left(\sqrt{\frac{\lambda}{1-\lambda}}\right)^{k}\le \sum_{k=0}^\infty\left(\frac{e^{-\gamma}\lambda}{1-\lambda}\right)^{k/2}=\frac{1}{1-\sqrt{\frac{e^{-\gamma}\lambda}{1-\lambda}}}\le \frac{3}{2}\,,
\end{align*}
where the last inequality follows from $\frac{e^{-\gamma}\lambda}{1-\lambda}\le \frac{1}{9}$, which is implied by $\lambda\le\frac{1}{1+9e^{-\gamma}}$.
\end{proof}

\subsubsection{Expanding the anti-degradability region numerically}

Note that Theorem~\ref{main-result} does not identify the entire anti-degradability region of the bosonic loss-dephasing channel $\pazocal{N}_{\lambda,\gamma}$. 
In fact, from the above argument it becomes clear that a way to obtain a better inner approximation of this region is to check for which values of the parameters the infinite matrix $A$ defined by~\eqref{def_a_element_matrix} is positive semi-definite. This is established in the following theorem.

\begin{theorem}\label{thm_cond_infinite_A-infinite}
Let $\ell^2(\N)$ be the space of square-summable complex-valued sequences (defined by~\eqref{def_l2N} below). For any $\lambda\in(\frac{1}{2},1)$ and $\gamma>0$, let $A = (a_{mn})_{m,n \in \mathbb{N}}$ be the infinite matrix whose components are defined by~\eqref{def_a_element_matrix}. If $A\geq 0$ is positive semi-definite as an operator on $\ell^2(\N)$, then the bosonic loss-dephasing channel $\pazocal{N}_{\lambda,\gamma}$ is anti-degradable. 
\end{theorem}
\begin{proof}
In the proof of Theorem~\ref{main-result} we have seen that the bosonic loss-dephasing channel is anti-degradable if the Hadamard map associated with the infinite matrix $A$ (given in~\eqref{def_a_element_matrix}) is a quantum channel. Since the diagonal elements of $A$ are equal to one, from Lemma~\ref{lemma_hadamard_channel_sm} we deduce that the Hadamard map associated to $A$ is a quantum channel if and only if $A$ is positive semi-definite. This concludes the proof.
\end{proof}

In Theorem~\ref{main-result}, we showed that the above-mentioned infinite matrix $A$ is positive semi-definite if $\theta\!\left(e^{-\gamma/2},\sqrt{\lambda/(1-\lambda)}\right)\le\frac{3}{2}$, where $\theta(x,y)\coloneqq \sum_{n=0}^\infty x^{n^2}y^n $. This identifies just a portion of the full region of parameters of $\lambda$ and $\gamma$ where the infinite matrix $A$ is positive semi-definite.

To analyse the positive semi-definiteness of the infinite matrix $A$ further, let $A^{(d)}$ denote its $d\times d$ top-left corner. Note that it is well-known that an infinite matrix is positive semi-definite if and only if its $d\times d$ top-left corner is positive semi-definite for all $d\in\N$. For modest values of $d$, we can numerically determine the parameter region where $A^{(d)}$ is positive semi-definite. To achieve this, we plot in Fig.~\ref{fig:eta_d_gamma} the quantity
 \begin{align}\label{eta_d}
      \eta_d(\gamma)\coloneqq \max\left(\lambda\in\left(\frac{1}{2},1\right]:\,\, A^{(d)}\text{ is positive semi-definite}\right)\,,
\end{align}
with respect to $e^{-\gamma}$ for various values of $d$. 
This quantity is relevant because $A^{(d)}$ is positive semi-definite if and only if $\lambda\le \eta_d(\gamma)$. Moreover, the quantity $\eta_d(\gamma)$ monotonically decreases in $d$ and converges to some $\bar{\eta}(\gamma)$ as $d\to\infty$. Notably, the condition $\lambda\le \bar{\eta}(\gamma)$ is necessary and sufficient for positive semi-definiteness of the infinite matrix $A$, and also a sufficient condition for the anti-degradability of the the bosonic loss-dephasing channel $\pazocal{N}_{\lambda,\gamma}$. Our numerical investigation seems to suggest that when $d$ is approximately $20$, the quantity $\eta_d(\gamma)$ has already reached its limiting value $\bar{\eta}(\gamma)$, which can be approximated, for instance, by considering, the curve $\eta_{30}(\gamma)$.

\begin{figure}[!h]
	\centering
	\includegraphics[width=0.8\linewidth]{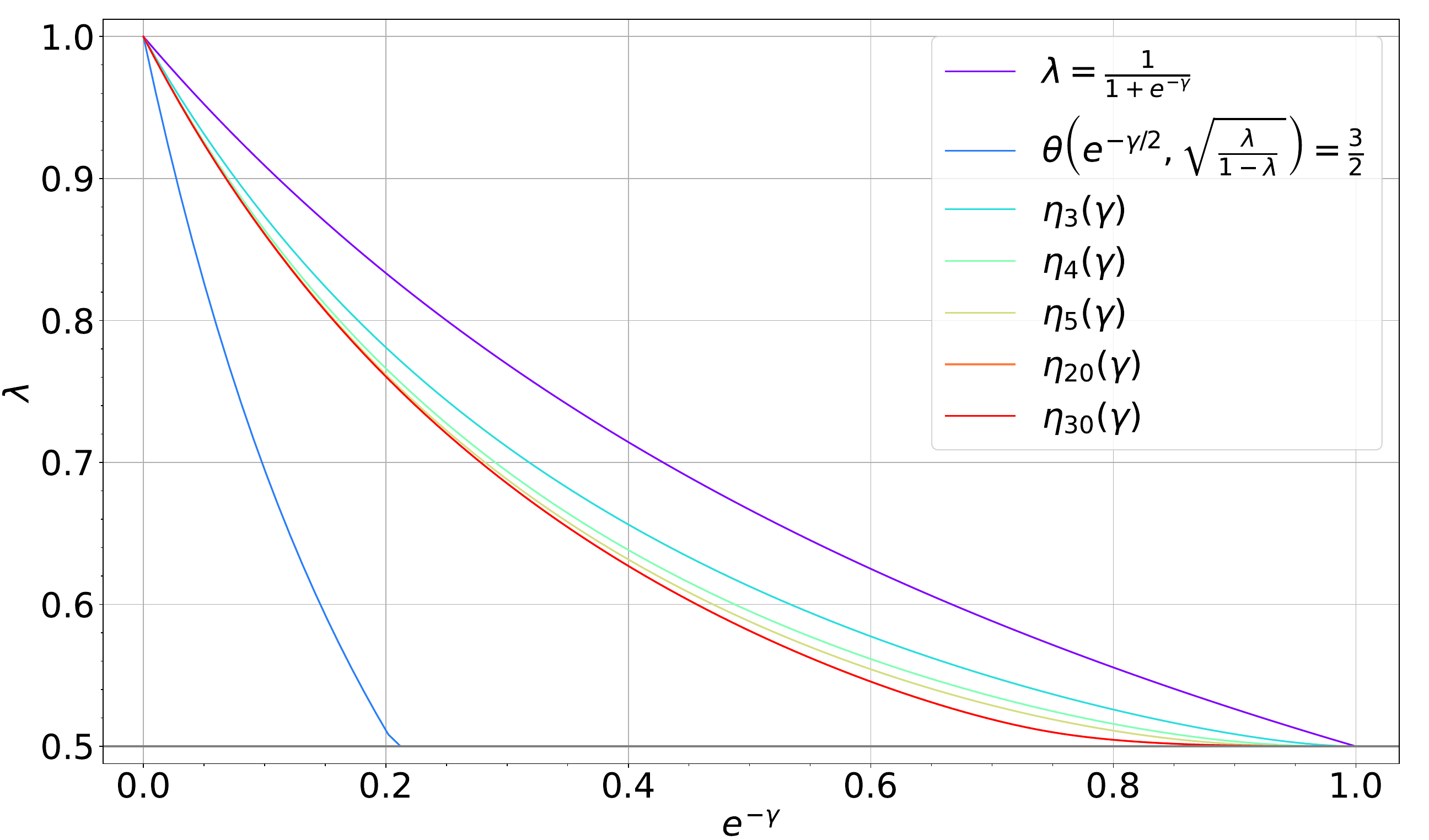}
	\caption{Numerical estimation of the anti-degradability of the loss-dephasing channel. The horizontal axis shows the quantity $e^{-\gamma}$, varying from $0$ to $1$ as dephasing parameter $\gamma$ decreases from $\infty$ to $0$, while the vertical axis corresponds to the transmissivity $\lambda$. Theorem~\ref{main-result} establishes that below the region defined by the blue curve, corresponding to the condition $\theta(e^{-\gamma/2},\sqrt{\lambda/(1-\lambda)})=3/2$, the bosonic loss-dephasing channel is anti-degradable. Moreover, Theorem~\ref{theorem_not_antideg} establishes that above the region defined by the purple curve, corresponding to $\lambda=\frac{1}{1+e^{-\gamma}}$, the bosonic loss-dephasing channel is not anti-degradable. The other curves depict the quantity $\eta_d(\gamma)$, which is defined in~\eqref{eta_d} as the maximum value of the transmissivity where $A^{(d)}$ is positive semi-definite, for various values of $d$. Our numerical analysis seems to indicate that in the region below the red curve, corresponding to $\lambda\le \eta_{30}(\gamma)$, the bosonic loss-dephasing channel is anti-degradable. Here, we employ $d=30$, as increasing $d$ beyond $d\ge20$ yields no discernible change in the plot.
 } 
	\label{fig:eta_d_gamma}
\end{figure}
%
%

\subsubsection{Anti-degrading maps}

Theorem~\ref{main-result} discovers parameter regions of transmissivity $\lambda$ and dephasing $\gamma$ in which the bosonic loss-dephasing channel $\pazocal{N}_{\lambda,\gamma}$ is anti-degradable. Although the proof of this theorem ensures the existence of anti-degrading maps for $\pazocal{N}_{\lambda,\gamma}$ within these parameter regions, it does not offer explicit constructions of such anti-degrading maps. In the forthcoming Theorem~\ref{explicit-maps}, we present such explicit constructions. Note that, thanks to Lemma~\ref{lemma_compl_channel}, the output operators of the complementary channel $\pazocal{N}^{\text{c}}_{\lambda,\gamma}$ reside within the space $\pazocal{T}(\HH_{E_{\text{out}}})$, where $\HH_{E_{\text{out}}}$ is the following subspace of the two-mode Hilbert space $\HH_{E_1}\otimes\HH_{E_2}=L^2(\mathbb{R})\otimes L^2(\mathbb{R})$: 
\begin{equation}
    \begin{aligned}
        \label{statement_def_h_out}
   \HH_{E_{\text{out}}}\coloneqq  \text{Span}\left\{  \ket{\ell}_{E_1}\otimes\ket{\sqrt{\gamma}n}_{E_2}:\,\, \ell\le n  \text{ with }\ell,n\in\N\right\},
    \end{aligned}
\end{equation}
where $\ket{n}$ represents the $n$th Fock state, and $\ket{\sqrt{\gamma}n}$ denotes the coherent state with a parameter of $\sqrt{\gamma}n$. These states correspond to the environmental modes of the pure-loss and dephasing channels, respectively (we shall maintain this notation throughout).
 \begin{theorem}\label{explicit-maps}
   Anti-degrading maps corresponding to each parameter region in Theorem~\ref{main-result} can be defined as follows. In the region (i), i.e.~$\lambda\in\left[0,\frac{1}{2}\right]$ and $\gamma\ge0$, an anti-degrading map is given by 
\bb\label{expr_antideg_lambda_min_05}
     \pazocal{A}_{\lambda,\gamma}= \left(\pazocal{E}_{ \frac{\lambda}{1-\lambda} }\circ\pazocal{R}_{E_1}\right)\otimes\Tr_{E_2},
\ee
where $\pazocal{R}_{E_1}(\cdot)\coloneqq (-1)^{\hat{e}_1^\dagger \hat{e}_1}\cdot (-1)^{\hat{e}_1^\dagger \hat{e}_1}$, with $\hat{e}_1$ as the annihilation operator of the output mode of the pure-loss channel $E_1$.

In the region (ii), i.e.~$\lambda\in(\frac{1}{2},1)$ and $\gamma$ such that $\theta\!\left(e^{-\gamma/2},\sqrt{\frac{\lambda}{1-\lambda}}\right)\le\frac{3}{2}$, an anti-degrading map $\pazocal{A}_{\lambda,\gamma}:\pazocal{T}\left(\HH_{E_\text{out}}\right)\to\pazocal{T}\left(\HH_\text{B}\right)$, with $\HH_{E_\text{out}}$ given by~\eqref{statement_def_h_out}, is defined as follows. For all $\ell_1,\ell_2,n_1,n_2\in\mathbb{N}$ with $\ell_1\le n_1$ and $\ell_2\le n_2$, it holds that
    \bb\label{def_A_tilde_state}
        {\pazocal{A}}_{\lambda,\gamma}\!\left(\ketbra{\ell_1}{\ell_2}\otimes\ketbra{\sqrt{\gamma}n_1}{\sqrt{\gamma}n_2}\right)\coloneqq \sum_{k=0}^{\min(n_1-\ell_1,n_2-\ell_2)}c_{k}^{(\ell_1,\ell_2,n_1,n_2)}\ketbra{k+\ell_1}{k+\ell_2}\,,
    \ee
    where for all $k\in\{0,1,\ldots,\min(n_1-\ell_1,n_2-\ell_2)\}$ the coefficients $c_{k}^{(\ell_1,\ell_2,n_1,n_2)}$ are defined as
    \bb\label{def_coeff_c_state}
        c_{k}^{(\ell_1,\ell_2,n_1,n_2)}&\coloneqq (-1)^{\ell_1-\ell_2}\sqrt{\pazocal{B}_k\!\left(n_1-\ell_1,\frac{2\lambda-1}{\lambda}\right)\, \pazocal{B}_k\!\left(n_2-\ell_2,\frac{2\lambda-1}{\lambda}\right)}\,a_{n_1-\ell_1n_2-\ell_2}\,a_{k+\ell_1k+\ell_2}\,,
    \ee
    where $\pazocal{B}_{l}\!\left(n,\lambda\right)\coloneqq \binom{n}{l}\lambda^{l}(1-\lambda)^{n-l}$, and $a_{mn}$ is defined in~\eqref{def_a_element_matrix}. 
\end{theorem}

\begin{proof}[Proof of Theorem~\ref{explicit-maps}] 
Let us suppose that $\lambda$ and $\gamma$ fall within the parameter region (i). A complementary channel of the pure-loss channel $\pazocal{E}_\lambda$ is given by $\pazocal{E}_\lambda^\text{c}=\pazocal{E}_{1-\lambda}\circ\pazocal{R}$.
Consequently, Lemma~\ref{lemma_comp_pure} implies that $\left(\pazocal{E}_{ \frac{\lambda}{1-\lambda} }\circ\pazocal{R}\right)\circ \pazocal{E}_\lambda^\text{c}= \pazocal{E}_\lambda$, i.e.~the channel $\pazocal{E}_{ \frac{\lambda}{1-\lambda} }\circ\pazocal{R}$ is an anti-degrading map of the pure-loss channel. The general construction detailed in the proof of Lemma~\ref{lemma_comp_antideg} for the anti-degrading map of the composition between an anti-degradable channel and another channel demonstrates that the map given in~\eqref{expr_antideg_lambda_min_05} is an anti-degrading map of $\pazocal{N}_{\lambda,\gamma}$.

Let us now suppose that $\lambda$ and $\gamma$ fall within the parameter region (ii). To come up with the anti-degrading map defined in~\eqref{def_A_tilde_state}, we drew intuition from the proof of Lemma~\ref{lemma_infinite_extendibility}, which demonstrates the equivalence between two-extendibility of the Choi state and the existence of an anti-degrading map, while also considering the two-extension of the Choi state of $\pazocal{N}_{\lambda,\gamma}$ explicitly found in~\eqref{state_sym_ext}.
In order to show that the map ${\pazocal{A}}_{\lambda,\gamma}$ in~\eqref{def_A_tilde_state} is an anti-degrading map of $\pazocal{N}_{\lambda,\gamma}$, we need to show that it is a quantum channel satisfying ${\pazocal{A}}_{\lambda,\gamma}\circ \pazocal{N}^{\text{c}}_{\lambda,\gamma}=\pazocal{N}_{\lambda,\gamma}$. We begin by proving that ${\pazocal{A}}_{\lambda,\gamma}$ is trace preserving.  By linearity, it suffices to show that for any $\ell_1,\ell_2,n_1,n_2\in\mathbb{N}$ with $\ell_1\le n_1$ and $\ell_2\le n_2$ it holds that
   \bb
        \Tr\!\left[{\pazocal{A}}_{\lambda,\gamma}\!\left(\ketbra{\ell_1}{\ell_2}_{E_1}\otimes\ketbra{\sqrt{\gamma}n_1}{\sqrt{\gamma}n_2}_{E_2}\right)\right]=\Tr\!\left[ \ketbra{\ell_1}{\ell_2}_{E_1}\otimes\ketbra{\sqrt{\gamma}n_1}{\sqrt{\gamma}n_2}_{E_2} \right] \,.
   \ee
    Indeed, we obtain
    \begin{align*}
        \Tr\!\left[{\pazocal{A}}_{\lambda,\gamma}\!\left(\ketbra{\ell_1}{\ell_2}_{E_1}\otimes\ketbra{\sqrt{\gamma}m}{\sqrt{\gamma}n}_{E_2}\right)\right]&=\,\delta_{\ell_1,\ell_2}\sum_{k=0}^{\min(m-\ell_1,n-\ell_1)}c_k^{(\ell_1,\ell_1,m,n)}\\
        &=\, \delta_{\ell_1,\ell_2}\sum_{k=0}^{\min(m-\ell_1,n-\ell_1)} \sqrt{ \pazocal{B}_k\!\left(m-\ell_1,\frac{2\lambda-1}{\lambda}\right) \pazocal{B}_k\!\left(n-\ell_1,\frac{2\lambda-1}{\lambda}\right)} a_{m-\ell_1,n-\ell_1}\,a_{k+\ell_1,k+\ell_1}\\
        &=\, \delta_{\ell_1,\ell_2} e^{-\frac{\gamma}{2}(m-n)^2}\\
        &=\Tr\!\left[ \ketbra{\ell_1}{\ell_2}_{E_1}\otimes\ketbra{\sqrt{\gamma}n_1}{\sqrt{\gamma}n_2}_{E_2} \right] \,,
    \end{align*}
    where $\delta_{\ell_1,\ell_2}$ denotes the kronecker delta and where we have exploited the formula for the overlap between coherent states provided in~\eqref{overlap_coh}. Now, let us show that ${\pazocal{A}}_{\lambda,\gamma}$ is completely positive. To achieve this, we need to find a pure state $\Psi$ on $\HH_{\text{anc}}\otimes \HH_{E_{\text{out}}}$, where $\HH_{\text{anc}}$ in an auxiliary reference, such that $\Tr_{\text{anc}}[\ketbra{\Psi}]>0$ and $(\Id_{\text{anc}}\otimes {\pazocal{A}}_{\lambda,\gamma})(\ketbra{\Psi})\geq0$~\cite{HOLEVO-CHANNELS-2,Holevo-CJ,Holevo-CJ-arXiv}. Let $\HH_{\text{anc}}\coloneqq \HH_{A}\otimes\HH_{B_1}=L^2(\mathbb{R})\otimes L^2(\mathbb{R})$ and let us construct the pure state $\ket{\Psi}_{AB_1E_1E_2}\in \HH_{\text{anc}}\otimes \HH_{E_{\text{out}}}$ as follows: 
    \begin{equation}
        \begin{aligned}
           \label{expr_psi_stinespring0}
        \ket{\Psi}_{AB_1E_1E_2}&\coloneqq   U_\lambda^{B_1E_1}V_\gamma^{B_1E_2} \ket{\psi(r)}_{AB_1}\ket{0}_{E_1}\ket{0}_{E_2}\\
        &=\frac{1}{\cosh (r)}\sum_{n=0}^{\infty}\sum_{\ell=0}^n (-1)^\ell\tanh^n (r) \sqrt{\pazocal{B}_l(n,1-\lambda)} \ket{n}_A\ket{n-\ell}_{B_1}\ket{\ell}_{E_1}\ket{\sqrt{\gamma}n}_{E_2}\,, 
        \end{aligned}
    \end{equation}
    where $U_\lambda^{B_1E_1}$ is the beam splitter unitary, $V_\gamma^{B_1E_2}$ is the conditional displacement unitary,
    and $\ket{\psi(r)}_{AB_1}$ is the two-mode squeezed vacuum state with squeezing $r>0$. 
 Let us now show that $\Tr_{AB_1}\left[\ketbra{\Psi}_{AB_1E_1E_2}\right]$ is positive definite on $\HH_{E_{\text{out}}}$. Let $\ket{\phi}_{E_1E_2}\in\HH_{E_{\text{out}}}$. Since there exists $\bar{\ell},\bar{n}\in\N$ with $\bar{l}\le\bar{n}$ such that $\bra{\phi}_{E_1E_2}\ket{\bar{\ell}}_{E_1}\otimes\ket{\sqrt{\gamma}\bar{n}}_{E_2}\ne 0$,~\eqref{expr_psi_stinespring0} implies that
\begin{align*}
\bra{\phi}_{E_1E_2}\Tr_{AB_1}\left[\ketbra{\Psi}_{AB_1E_1E_2}\right]\ket{\phi}_{E_1E_2}&=\frac{1}{\cosh^2(r)}\sum_{n=0}^\infty\sum_{l=0}^n \tanh^{2n}(r)\, \pazocal{B}_{\ell}\left(n,1-\lambda\right) \left|\bra{\phi}_{E_1E_2}\ket{\ell}_{E_1}\otimes\ket{\sqrt{\gamma}n}_{E_2}\right|^2\\&\ge\frac{1}{\cosh^2(r)}\tanh^{2\bar{n}}(r) \,\pazocal{B}_{\bar{l}}\left(\bar{n},1-\lambda\right)\left|\bra{\phi}_{E_1E_2}\ket{\bar{\ell}}_{E_1}\otimes\ket{\sqrt{\gamma}\bar{n}}_{E_2}\right|^2\\&>0\,.
\end{align*}
We next show that $\Id_{AB_{1}}\otimes {\pazocal{A}}_{\lambda,\gamma}(\ketbra{\Psi})$ is positive semi-definite. Let $B_2$ denote the output system of ${\pazocal{A}}_{\lambda,\gamma}$. Note that
\begin{equation}
    \begin{aligned}
        \label{eq_expression_step_proof0}
    \Id_{AB_{1}}\otimes {\pazocal{A}}_{\lambda,\gamma}(\ketbra{\Psi}_{AB_1E_1E_2})& \texteq{(i)}  \frac{1}{\cosh^2 (r)}\sum_{m,n=0}^{\infty}\sum_{\ell_1=0}^{m}\sum_{\ell_2=0}^{n} (-1)^{\ell_1+\ell_2}[\tanh (r)]^{m+n} \sqrt{ \pazocal{B}_{\ell_1}(m,1-\lambda)  \pazocal{B}_{\ell_2}(n,1-\lambda)  }\\&\hspace{3.6cm} \ketbra{m}{n}_A\otimes\ketbra{m-\ell_1}{n-\ell_2}_{B_1}\otimes {\pazocal{A}}_{\lambda,\gamma}\!\left(\ketbra{\ell_1}{\ell_2}_{E_1}\otimes\ketbra{\sqrt{\gamma}m}{\sqrt{\gamma}n}_{E_2}\right)
    \\&\texteq{(ii)} \frac{1}{\cosh^2 (r)}\sum_{m,n=0}^{\infty}\sum_{\ell_1=0}^{m}\sum_{\ell_2=0}^{n}\sum_{k=0}^{\min(m-\ell_1,n-\ell_2)} [\tanh (r)]^{m+n}  \\&\hspace{2cm}\sqrt{ \pazocal{B}_{\ell_1}(m,1-\lambda)\,  \pazocal{B}_{\ell_2}(n,1-\lambda)\,\pazocal{B}_{k}\left(m-\ell_1,\frac{2\lambda-1}{\lambda}\right)  \pazocal{B}_{k}\left(n-\ell_2,\frac{2\lambda-1}{\lambda}\right) }\\&\hspace{3.4cm} a_{m-\ell_1,n-\ell_2}\,a_{k+\ell_1,k+\ell_2}\ketbra{n_1}{n}_A\otimes\ketbra{m-\ell_1}{n-\ell_2}_{B_1}\otimes \ketbra{k+\ell_1}{k+\ell_2}_{B_2}
    \\&\texteq{(iii)} \rho_{AB_1B_2}\,.
    \end{aligned}
\end{equation}
Here, in (i) we used the definition of $\ket{\Psi}_{AB_1E_1E_2}$ given in~\eqref{expr_psi_stinespring0}; in (ii) we utilised the definition of the map ${\pazocal{A}}_{\lambda,\gamma}$ from~\eqref{def_A_tilde_state}; and in (iii) we recognised the tripartite operator $\rho_{AB_1B_2}$ defined in~\eqref{state_sym_ext}, which is a quantum state provided that $\lambda$ and $\gamma$ satisfy $\theta\!\left(e^{-\gamma/2},\sqrt{\frac{\lambda}{1-\lambda}}\right)\le\frac{3}{2}$. Therefore, in such a parameter region, $\Id_{AB_{1}}\otimes {\pazocal{A}}_{\lambda,\gamma}(\ketbra{\Psi}_{AB_1E_1E_2})$ is positive semi-definite, and thus ${\pazocal{A}}_{\lambda,\gamma}$ is a quantum channel.
Let us now verify that
${\pazocal{A}}_{\lambda,\gamma}\circ \pazocal{N}^{\text{c}}_{\lambda,\gamma}=\pazocal{N}_{\lambda,\gamma}$. To show this, note that
\begin{equation}
    \begin{aligned}
          \Id_{A}\otimes\left({\pazocal{A}}_{\lambda,\gamma}\circ \pazocal{N}^{\text{c}}_{\lambda,\gamma}\right)(\ketbra{\psi(r)}_{AB_1})&\texteq{(iv)}\Tr_{B_1}\left[ \Id_{AB_{1}}\otimes {\pazocal{A}}_{\lambda,\gamma}(\ketbra{\Psi}_{AB_1E_1E_2})  \right]
    \\&\texteq{(v)}\Tr_{B_1}\left[ \rho_{AB_1B_2}  \right]
    \\&\texteq{(vi)}\Id_{A}\otimes\pazocal{N}_{\lambda,\gamma}\left(\ketbra{\psi(r)}_{AB_2}\right)\,,
    \end{aligned}
\end{equation}
Here, in (iv) we employed~\eqref{expr_psi_stinespring0}; in (v) we exploited~\eqref{eq_expression_step_proof0}; and in (vi) we used that $\rho_{AB_1B_2}$ is a two-extension of $\Id_{A}\otimes\pazocal{N}_{\lambda,\gamma}\left(\ketbra{\psi(r)}\right)$, as established in the proof of Theorem~\ref{main-result}. Finally, since the two-mode squeezed vacuum state $\ket{\psi(r)}_{AB}$ satisfies $\Tr_B[\ketbra{\psi(r)}_{AB}]>0$, we conclude ${\pazocal{A}}_{\lambda,\gamma}\circ \pazocal{N}^{\text{c}}_{\lambda,\gamma}=\pazocal{N}_{\lambda,\gamma}$.
\end{proof}

\subsection{Analysis of the bosonic loss-dephasing channel via its finite-dimensional restrictions}\label{subsec_finite}
\begin{Def}
   Let $d\in\N$ and let $\HH_d\coloneqq \operatorname{Span}(\{\ket{n}\}_{n=0,\ldots,d-1})$ be the subspace spanned by the first $d$ Fock states. The qudit restriction of the bosonic loss-dephasing channel $\pazocal{N}_{\lambda,\gamma}^{(d)}$ is a quantum channel defined by
    \bb
        \pazocal{N}_{\lambda,\gamma}^{(d)}(\Theta)\coloneqq \pazocal{N}_{\lambda,\gamma}(\Theta)\qquad\forall\,\Theta\in\pazocal{T}(\HH_d)
    \ee
\end{Def}
\begin{lemma}
    Let $\HH_A,\HH_B\coloneqq L^2(\mathbb{R})$.  Let $\pazocal{N}:\pazocal{T}(\HH_A)\to \pazocal{T}(\HH_B)$ be a quantum channel and let $\pazocal{N}^{(d)}$ be its qudit restriction, defined by 
    \bb
        \pazocal{N}^{(d)}(\Theta)\coloneqq \pazocal{N}(\Theta)\qquad\forall\,\Theta\in\pazocal{T}(\HH_d)\,.
    \ee
    If $\pazocal{N}$ is (anti-)degradable, then its qudit restriction $\pazocal{N}^{(d)}$ is (anti-)degradable.
\end{lemma}
\begin{proof}
    Let $\pazocal{N}(\cdot)=\Tr_E[V_{A\to BE}(\cdot)V_{A\to BE}^\dagger]$ be a Stinespring representation of $\pazocal{N}$, and let $\pazocal{N}^\text{c}(\cdot)=\Tr_B[V_{A\to BE}(\cdot)V_{A\to BE}^\dagger]$ be the associated complementary channel. Note that the isometry $V_{A\to BE}$ provides a Stinespring representation also for the qudit restriction $\pazocal{N}^{(d)}$. Hence, the qudit restriction of the complementary channel $\pazocal{N}^\text{c}$ is a complementary channel of the qudit restriction $\pazocal{N}^{(d)}$, i.e.
    \bb
        (\pazocal{N}^{(d)})^\text{c}(\Theta)= \pazocal{N}^\text{c}(\Theta)\qquad\forall\,\Theta\in\pazocal{T}(\HH_d)\,.
    \ee
    Moreover, note that any degrading or anti-degrading map of $\NN$ is effective for all input states, including those restricted to $\HH_d$. Consequently, an (anti-)degrading map of $\NN$ is also an (anti-)degrading map of its qudit restriction $\NN^{(d)}$.
\end{proof}

\begin{cor}\label{cor-antideg}
If the qudit restriction $\pazocal{N}_{\lambda,\gamma}^{(d)}$ is not (anti)-degradable, then the bosonic loss-dephasing channel $\pazocal{N}_{\lambda,\gamma}$ is also not (anti)-degradable.
\end{cor}
The following lemma shows that qudit restriction $\pazocal{N}_{\lambda,\gamma}^{(d)}$ is a qu$d$it-to-qu$d$it channel, mapping the space spanned by $\{\ket{n}\}_{n=0,\ldots,d-1}$ into itself. 

\begin{lemma}\label{qudit restriction-1}
   If the input state to the bosonic loss-dephasing channel is confined into the finite-dimensional subspace spanned by $\{\ket{n}\}_{n=0,\ldots,d-1}$, the resulting output state will similarly be confined to this subspace.
\end{lemma}
\begin{proof}
    Examining Lemma~\ref{App_remark_commutation} reveals that the operator $\ketbra{m}{n}$, when acted on by the bosonic loss-dephasing channel, is transformed into linear combinations of operators $\{\ketbra{\ell}{k}\}_{\ell\leq m,k\leq n}$. This means that
if the input state to the bosonic loss-dephasing channel is restricted to the $d$-dimensional subspace $\{\ket{n}\}_{n=0,\ldots,d-1}$, the output of the channel will reside within the same subspace.
\end{proof}
The qubit restriction $\pazocal{N}_{\lambda,\gamma}^{(2)}$ of the bosonic loss-dephasing channel coincides with the composition between the amplitude damping channel and the qubit dephasing channel~\cite{Sumeet_book,MARK}, which we dub \emph{amplitude-phase damping channel}. Theorems~\ref{thm_degradability} and~\ref{theorem_not_antideg} utilise the amplitude-phase damping channel $\pazocal{N}_{\lambda,\gamma}^{(2)}$ to find that the bosonic loss-dephasing channel $\pazocal{N}_{\lambda,\gamma}$ is never degradable and, additionally, it is not anti-degradabile for $\lambda>\frac{1}{1+e^{-\gamma}}$, respectively.

\subsubsection{The bosonic loss-dephasing channel is never degradable}
The bosonic loss-dephasing channel $\pazocal{N}_{\lambda,\gamma}$ is never degradable, except when it coincides with either the bosonic dephasing channel ($\gamma>0$ and $\lambda=1$) or the degradable pure-loss channel ($\gamma=0$ and $\lambda\ge\frac{1}{2}$), thereby complicating the derivation of its quantum capacity~\cite{Loss_dephasing0}.  This result has been previously demonstrated in~\cite{Loss_dephasing0} through pages-long proof; however, here we provide a significantly simpler proof of this result.
\begin{theorem}\label{thm_degradability}
Let $\lambda\in[0,1]$ and $\gamma\ge0$. The bosonic loss-dephasing channel $\pazocal{N}_{\lambda,\gamma}$ is degradable if and only if one of the following conditions is satisfied:
\begin{itemize}
    \item $\gamma=0$ and $\lambda\in[\frac{1}{2},1]$
    \item $\gamma\ge0$ and $\lambda=1$
\end{itemize}
\end{theorem}
\begin{proof}
Thanks to Corollary~\ref{cor-antideg}, a necessary condition for $\pazocal{N}_{\lambda,\gamma}$ to be degradable is the degradability of the amplitude-phase damping channel $\pazocal{N}_{\lambda,\gamma}^{(2)}$. We now apply~\cite[Theorem 4]{Cubitt2008}, which establishes a necessary condition on the degradability of any qubit channel. Specifically, the rank of the Choi state of a degradable qubit channel is necessarily less or equal to $2$.
By using the notation used in~\eqref{def_Choi}, the matrix associated with the Choi state of the amplitude-phase damping channel $C\!\left(\pazocal{N}_{\lambda,\gamma}^{(2)}\right)$ in the computational basis $\{\ket{00}, \ket{01}, \ket{10}, \ket{11}\}$ is given by:
\begin{align}
    \label{choistate_bidimensional}
&\frac{ 1}{ 2}\left(\begin{matrix} 1\quad & 0& 0& \sqrt{e^{-\gamma}\lambda}\\ 0 & 0\quad & 0\quad& 0\\
0 & 0\quad & 1-\lambda\quad & 0\\ \sqrt{e^{-\gamma}\lambda} & 0 & 0 &\lambda\quad \end{matrix}\right) \,.
\end{align}
One can easily see that its rank is equal to $3$ for $\gamma>0$ and $\lambda\in(0,1)$. In addition, for $\lambda=1$ the bosonic loss-dephasing channel coincides with the bosonic dephasing channel, $\pazocal{N}_{1,\gamma}=\pazocal{D}_{\gamma}$, which is degradable for any value of $\gamma\ge0$~\cite{PhysRevA.102.042413}. Finally, for $\gamma=0$ the bosonic loss-dephasing channel coincides with the pure-loss channel, $\pazocal{N}_{\lambda,0}=\pazocal{E}_\lambda$, which is degradable if and only if $\lambda\in[\frac{1}{2},1]$~\cite{Wolf2007,Caruso2006}.
\end{proof}

\subsubsection{Necessary condition on anti-degradability via qubit restriction}
The next theorem establishes the parameter range where the bosonic loss-dephasing channel is not anti-degradable. We provide three different proofs for this theorem.
\begin{theorem}\label{theorem_not_antideg}
    Let $\gamma\ge0$. If $\lambda> \frac{1}{1+e^{-\gamma}}$, then the bosonic loss-dephasing channel $\pazocal{N}_{\lambda,\gamma}$ is not anti-degradable. 
\end{theorem}
\begin{proof}[Proof 1]
    Assume that $\pazocal{N}_{\lambda,\gamma}$ is anti-degradable; then substituting $m=0$ and $n=1$ in~\eqref{cond_A_antideg_proof_st} yields
    \bb
    \pazocal{A}_{\lambda,\gamma}\left( \pazocal{E}_{1-\lambda}(\ketbra{0}{1})\otimes\ketbra{0}{\sqrt{\gamma}}   \right)=-e^{-\frac{1}{2}\gamma }\pazocal{E}_{\lambda}(\ketbra{0}{1})\,.
    \ee
    By exploiting $\pazocal{E}_\lambda(\ketbra{0}{1})=\sqrt{\lambda}\ketbra{0}{1}$, we have
    \begin{align*}
         \pazocal{A}_{\lambda,\gamma}\left( \ketbra{0}{1}\otimes\ketbra{0}{\sqrt{\gamma}}\right)=-\sqrt{\frac{e^{-\gamma}\lambda}{1-\lambda}}\ketbra{0}{1}\,.
    \end{align*}
    Using Lemma~\ref{lemmino} in the Appendix, we find $\sqrt{\frac{e^{-\gamma}\lambda}{1-\lambda}}\le 1$, or $\lambda\le \frac{1}{1+e^{-\gamma}}$.
\end{proof}
\begin{proof}[Proof 2]
    Assume that $\pazocal{N}_{\lambda,\gamma}$ is anti-degradable. As a consequence of~\eqref{cond_A_antideg_proof_st} and of the data-processing inequality for the fidelity~\cite{NC}, we find
    \begin{equation}
        \begin{aligned}
            \label{ineq_proof2}
        F\big(\,\pazocal{E}_{\lambda}(\ketbra{0})\,,\, \pazocal{E}_{\lambda}(\ketbra{1})\,\big)\ge F\big( \,   \pazocal{E}_{1-\lambda}(\ketbra{0})\otimes \ketbra{0}   \,, \,   \pazocal{E}_{1-\lambda}(\ketbra{1})\otimes\ketbra{\sqrt{\gamma}}    \, \big)\,.
        \end{aligned}
    \end{equation}
    Furthermore, we obtain
    \begin{equation}
        \begin{aligned}
             \sqrt{1-\lambda}&=F\big(\ketbra{0},\lambda\ketbra{1}+(1-\lambda)\ketbra{0}\big)\\&\texteq{(i)}F\big(\pazocal{E}_\lambda(\ketbra{0}),\pazocal{E}_\lambda(\ketbra{1})\big)\\&\textgeq{(ii)} F\big(\ketbra{0},\ketbra{\sqrt{\gamma}}\big)\,F\big(\pazocal{E}_{1-\lambda}(\ketbra{0}),\pazocal{E}_{1-\lambda}(\ketbra{1})\big)\\&\textgeq{(iii)} F\big(\ketbra{0},\ketbra{\sqrt{\gamma}}\big)\,F\big(\ketbra{0},\lambda\ketbra{0}+(1-\lambda)\ketbra{1}\big)\\&\texteq{(iv)}\sqrt{e^{-\gamma}\lambda}\,.
        \end{aligned}
    \end{equation}
     Here, (i) follows from Lemma~\ref{fock_pureloss}, (ii) follows from~\eqref{ineq_proof2} and from the fact that the fidelity is multiplicative under tensor product~\cite{NC}, (iii) uses Lemma~\ref{fock_pureloss} again, and in (iv) we exploited that $|\braket{0}{\sqrt{\gamma}}|=\sqrt{e^{-\gamma}}$. This yields $\sqrt{1-\lambda}\geq\sqrt{e^{-\gamma}\lambda}$, or $\lambda\le \frac{1}{1+e^{-\gamma}}$.
\end{proof}
\begin{proof}[Proof 3]
By exploiting Lemma~\ref{lemma_qubit_charact_antideg} and the Choi matrix of the qubit channel $\pazocal{N}^{(2)}_{\lambda,\gamma}$ reported in~\eqref{choistate_bidimensional}, one can easily obtain that $\pazocal{N}^{(2)}_{\lambda,\gamma}$ is anti-degradable if and only if $\lambda\le \frac{1}{1+e^{-\gamma}}$. Consequently, thanks to Corollary~\ref{cor-antideg}, the bosonic loss-dephasing channel $\pazocal{N}_{\lambda,\gamma}$ is not anti-degradable for $\lambda> \frac{1}{1+e^{-\gamma}}$.
\end{proof}

\subsubsection{Necessary condition on anti-degradability via qudit restrictions}
Let us introduce the following quantity for any $d\in\mathbb{N}$ and $\gamma\ge0$:
\begin{align}
    \label{lambda_d_gamma}
\lambda_d(\gamma)\coloneqq \max\left(\lambda\in[0,1]:\,\, \pazocal{N}^{(d)}_{\lambda,\gamma}\text{ is anti-degradable}\right)\,.
\end{align} 
This quantity is relevant since it allows us to find parameter region where the bosonic loss-dephasing channel is not anti-degradable. Specifically, for $\lambda>\lambda_d(\gamma)$ the bosonic loss-dephasing channel $\pazocal{N}_{\lambda,\gamma}$ is not anti degradable, as established by Corollary~\ref{cor-antideg}. Thanks to the Proof 3 of Theorem~\ref{theorem_not_antideg}, it follows that for $d=2$ we have that $\lambda_2(\gamma)=\frac{1}{1+e^{-\gamma}}$, thereby establishing that $\pazocal{N}_{\lambda,\gamma}$ is not anti degradable for $\lambda>\frac{1}{1+e^{-\gamma}}$. Through an examination of larger values of $d$, we aim to identify an extended parameter region where the channel is not anti-degradable (see Fig.~\ref{fig:sdp_solution}). We begin by proving some useful properties of the quantity $\lambda_d(\gamma)$.

\begin{lemma}\label{lemma_six_fact}
 	For any $\gamma\ge0$ and $d\in\mathbb{N},d\ge 2$, the following facts hold:
 	\begin{itemize}
 		\item Fact 1: The qu$d$it restriction of the bosonic loss-dephasing channel $\pazocal{N}^{(d)}_{\lambda,\gamma}$ is anti-degradable if and only if $\lambda\le\lambda_d(\gamma)$
 		 \item Fact 2: The quantity $\lambda_d(\gamma)$ is monotonically increasing in $\gamma$
 		 \item Fact 3: For $d=2$, it precisely holds that $\lambda_2(\gamma)= \frac{1}{1+e^{-\gamma}}$
	 	\item Fact 4: The quantity $\lambda_d(\gamma)$ is monotonically non-increasing in $d$ 
 		\item Fact 5:   
    It holds that $\frac{1}{2}\le \lambda_d(\gamma)\le \frac{1}{1+e^{-\gamma}}$
 		\item Fact 6: When $d=3$ and $e^{-\gamma}\le \sqrt{2}-1$ (or $\gamma\ge0.881$), it exactly holds that $\lambda_3(\gamma)= \frac{1}{1+e^{-\gamma}}$
 	\end{itemize}
\end{lemma}

\begin{proof}

\emph{Fact 1.}
It suffices to show that for any $\gamma\ge0$ and $\lambda,\lambda'\in[0,1]$ with $\lambda'<\lambda$, if $\pazocal{N}^{(d)}_{\lambda,\gamma}$ is anti-degradable, then so is $\pazocal{N}^{(d)}_{\lambda',\gamma}$. To show this, we exploit the composition rule 
\bb
\pazocal{E}_{\lambda_1}\circ\pazocal{N}_{\lambda_2,\gamma}=\pazocal{N}_{\lambda_1\lambda_2,\gamma}\qquad\forall\,\lambda_1,\lambda_2\in[0,1]\,,
\ee
as established by Lemma~\ref{comp_rule_bosonic_loss_deph}, implying that the channel $\pazocal{N}^{(d)}_{\lambda',\gamma}$ can be written as the composition between $\pazocal{N}^{(d)}_{\lambda,\gamma}$ and another channel. Consequently, Lemma~\ref{lemma_comp_antideg} concludes the proof.

 \medskip
\emph{Fact 2.}
	Analogously to Fact 1, it suffices to show that for any $\lambda\in[0,1]$ and for any $\gamma'\geq\gamma\ge 0$, if $\pazocal{N}^{(d)}_{\lambda,\gamma}$ is anti-degradable, then so is $\pazocal{N}^{(d)}_{\lambda,\gamma'}$. This follows from the composition rule 
\bb
\D_{\gamma_1}\circ\pazocal{N}_{\lambda,\gamma_2}=\pazocal{N}_{\lambda,\gamma_1+\gamma_2}\qquad\forall\,\gamma_1,\gamma_2\ge0\,,
\ee
as proved in Lemma~\ref{comp_rule_bosonic_loss_deph}. Furthermore, Lemma~\ref{lemma_comp_antideg} concludes the proof.

\medskip
\emph{Fact 3.} This has already been proved in the Proof 3 of Theorem~\ref{theorem_not_antideg}.

\medskip
\emph{Fact 4.}
	This follows from the observation that for all $d'\ge d$, if $\pazocal{N}_{\lambda,\gamma}^{(d')}$ is anti-degradable, then so is $\pazocal{N}_{\lambda,\gamma}^{(d)}$. 

\medskip
\emph{Fact 5.}
	The upper bound $\lambda_d(\gamma)\le\frac{1}{1+e^{-\gamma}}$ follows from Fact 3 and Fact 4. Moreover, since the pure-loss channel $\pazocal{E}_{\lambda}$ is anti-degradable if and only if $\lambda\le\frac{1}{2}$, Fact 4 implies that $\lambda_d(0)\ge \frac{1}{2}$ (more specifically, one can also show that $\lambda_d(0)=\frac{1}{2}$). Consequently, Fact 2 concludes the proof. 

\medskip
\emph{Fact 6.}
	This proof relies on the equivalence between anti-degradability of a channel and two-extendibility of its Choi state, as established in Lemma~\ref{lemma_infinite_extendibility}. Let $\lambda$ and $\gamma$ be such that $\lambda=\frac{1}{1+e^{-\gamma}}$ and $e^{-\gamma}\le \sqrt{2}-1$, implying that ~$\lambda\ge \frac{1}{\sqrt{2}}$. By using Lemma~\ref{App_remark_commutation}, we obtain the Choi state of $\pazocal{N}^{(3)}_{\lambda,\gamma}$ as follows:
 \begin{align*}
		\Id_A\otimes \pazocal{N}_{\lambda,\gamma}(\ketbra{\Phi_3})=\frac{1}{3}\sum_{m=0}^{2}\sum_{n=0}^{2}\sum_{\ell=0}^{\min(m,n)}e^{-\frac{\gamma}{2}(m-n)^2}\sqrt{\binom{m}{\ell}\binom{n}{\ell}}\lambda^{\frac{m+n}{2}-\ell}(1-\lambda)^l\ketbra{m}{n}_{A}\otimes\ketbra{m-\ell}{n-\ell}_{B},
 \end{align*}
 where $\Phi_3$ is the maximally entangled state of Schmidt rank $3$. We define a two-extension $\tilde{\rho}_{AB_1B_2}$ of the Choi state by the following conditions. First, $\tilde{\rho}_{AB_1B_2}$ satisfies the following $B_1\leftrightarrow B_2$ symmetry for all $i_1,i_2,i_3,j_1,j_2,j_3\in\{0,1,2\}$:
 \begin{equation}
     \begin{aligned}
         \bra{j_1}_{A}\bra{j_2}_{B_1}\bra{j_3}_{B_2}  \tilde{\rho}_{AB_1B_2}\ket{i_1}_{A}\ket{i_2}_{B_1}\ket{i_3}_{B_2}&=\bra{j_1}_{A}\bra{j_3}_{B_1}\bra{j_2}_{B_2}  \tilde{\rho}_{AB_1B_2}\ket{i_1}_{A}\ket{i_2}_{B_1}\ket{i_3}_{B_2}\\
	\bra{j_1}_{A}\bra{j_2}_{B_1}\bra{j_3}_{B_2}  \tilde{\rho}_{AB_1B_2}\ket{i_1}_{A}\ket{i_2}_{B_1}\ket{i_3}_{B_2}&=\bra{j_1}_{A}\bra{j_2}_{B_1}\bra{j_3}_{B_2}  \tilde{\rho}_{AB_1B_2}\ket{i_1}_{A}\ket{i_3}_{B_1}\ket{i_2}_{B_2}\,.
     \end{aligned}
 \end{equation}
	Furthermore, if $i_1<\max(i_2,i_3)$ or $j_1<\max(j_2,j_3)$, or if $i_1>i_2+i_3$ or $j_1>j_2+j_3$, then $\bra{j_1}_{A}\bra{j_2}_{B_1}\bra{j_3}_{B_2}  \tilde{\rho}_{AB_1B_2}\ket{i_1}_{A}\ket{i_2}_{B_1}\ket{i_3}_{B_2}=0\,.$
	We can thus define $\tilde{\rho}_{AB_1B_2}$ by writing only the matrix elements with respect the set $\{\ket{i_1}_{A}\ket{i_2}_{B_1}\ket{i_3}_{B_2}\}$ with $2\ge i_1\ge i_2\ge i_3\ge0$ such that $i_2+i_3\ge i_1$.
	Hence, in order to fully define $\tilde{\rho}_{AB_1B_2}$, it suffices to write the matrix elements of $\tilde{\rho}_{AB_1B_2}$ with respect to the set 
	$\{$ $\ket{0}_{A}\ket{0}_{B_1}\ket{0}_{B_2}$, $\ket{1}_{A}\ket{0}_{B_1}\ket{0}_{B_2}$, $\ket{1}_{A}\ket{1}_{B_1}\ket{1}_{B_2}$, $\ket{2}_{A}\ket{1}_{B_1}\ket{1}_{B_2}$ , $\ket{2}_{A}\ket{2}_{B_1}\ket{0}_{B_2}$, $\ket{2}_{A}\ket{2}_{B_1}\ket{1}_{B_2}$ , $\ket{2}_{A}\ket{2}_{B_1}\ket{2}_{B_2}$ $\}$. This gives rise to the following $7\times 7$ matrix:

	\[
	\begin{array}{c|ccccccc}
		& 0_A0_{B_1}0_{B_2} & 1_A1_{B_1}0_{B_2} & 1_A1_{B_1}1_{B_2} & 2_A1_{B_1}1_{B_2} & 2_A2_{B_1}0_{B_2} & 2_A2_{B_1}1_{B_2} & 2_A2_{B_1}2_{B_2} \\
		\hline\\
		0_A0_{B_1}0_{B_2} & 1 & \sqrt{1-\lambda} & 0 & \sqrt{2}(1-\lambda) & \frac{(1-\lambda)^2}{\lambda} & 0 & 0  \\\\
		1_A1_{B_1}0_{B_2} & \sqrt{1-\lambda} & 1-\lambda & 0 & \sqrt{2(1-\lambda)^3} & \sqrt{\frac{(1-\lambda)^5}{\lambda^2}} & 0 & 0  \\\\
		1_A1_{B_1}1_{B_2} & 0 & 0 & 2\lambda-1 & 0 & 0 & \frac{2\lambda-1}{\lambda}\sqrt{1-\lambda} & 0  \\\\
		2_A1_{B_1}1_{B_2} & \sqrt{2}(1-\lambda) & \sqrt{2(1-\lambda)^3} & 0 & 2(1-\lambda)^2 & \sqrt{2}\frac{(1-\lambda)^3}{\lambda}  & 0 & 0  \\\\
		2_A2_{B_1}0_{B_2} &  \frac{(1-\lambda)^2}{\lambda} &\sqrt{\frac{(1-\lambda)^5}{\lambda^2}}  & 0 & \sqrt{2}\frac{(1-\lambda)^3}{\lambda} & (1-\lambda)^2 & 0 & 0  \\\\
		2_A2_{B_1}1_{B_2} & 0 & 0 & \frac{2\lambda-1}{\lambda}\sqrt{1-\lambda} & 0 & 0 & 2(1-\lambda)(2\lambda-1) & 0  \\\\
		2_A2_{B_1}2_{B_2} & 0 & 0 & 0 & 0 & 0 & 0 & (2\lambda-1)^2 
	\end{array}
	\]
	One can show by direct calculation that this matrix is positive semi-definite if and only $\lambda\ge \frac{1}{\sqrt{2}}$. Note that the matrix is positive semi-definite if and only if $\tilde{\rho}_{AB_1B_2}$ is positive semi-definite. The latter follows from the following two simple facts: (i) A $n\times n$ symmetric matrix with a duplicate column is positive semi-definite if and only if the $(n-1)\times (n-1)$ matrix obtained by deleting one of the two equal column and the corresponding row is positive semi-definite, and (ii) A $n\times n$ symmetric matrix with a zero column is positive semi-definite if and only if the $(n-1)\times (n-1)$ matrix obtained by deleting such a column and the corresponding row is positive semi-definite.
	One can also verify $\Tr_{B_1}\tilde{\rho}_{AB_1B_2}=\Tr_{B_2}\tilde{\rho}_{AB_1B_2}$, and they are equal to the Choi state of the qutrit restriction with $\lambda=\frac{1}{1+e^{-\gamma}}$. We therefore conclude that the curve $\lambda=\frac{1}{1+e^{-\gamma}}$ provides a necessary and sufficient condition for the anti-degradability of the qutrit channel $\pazocal{N}_{\lambda,\gamma}^{(3)}$ when $\lambda\ge \frac{1}{\sqrt{2}}$, or equivalently when $e^{-\gamma}\le\sqrt{2}-1$.
\end{proof}

\medskip
To numerically compute the quantity $\lambda_d(\gamma)$ given in~\eqref{lambda_d_gamma}, we utilise the equivalence between anti-degradability of a channel and two-extendibility of its Choi state~\cite{Myhr2009}. Specifically,
for small values of $d$, we can determine necessary and sufficient conditions for the anti-degradability of $\pazocal{N}^{(d)}_{\lambda,\gamma}$ by numerically solving the following \emph{semi-definite program}:
\begin{equation}
    \begin{aligned}
        \label{SDP}
&\hspace{0.4cm}\min_{\rho_{AB_1B_2}}1\\&\text{s.t. } \rho_{AB_1B_2}\ge0\,,\\&\quad \Tr[\rho_{AB_1B_2}]=1,\\
&\quad\Tr_{B_2}\left[\rho_{AB_1B_2} \right]= \Id_{A}\otimes\pazocal{N}_{\lambda,\gamma}^{(d)}(\ketbra{\Phi_d}_{AA'})\,,\\
&\quad\Tr_{B_1}\left[\rho_{AB_1B_2} \right]= \Id_{A}\otimes\pazocal{N}_{\lambda,\gamma}^{(d)}(\ketbra{\Phi_d}_{AA'})\,.
    \end{aligned}
\end{equation}
where $\ket{\Phi_d}$ is the maximally entangled state of schmidt rank $d$. 
The channel $\pazocal{N}_{\lambda,\gamma}^{(d)}$ is anti-degradable if and only if the semi-definite program admits a feasible solution. We compute the quantity $\lambda_d(\gamma)$ defined in~\eqref{lambda_d_gamma} by numerically solving the semi-definite program.
\begin{figure}[t!]
	\centering
	\includegraphics[width=0.8\linewidth]{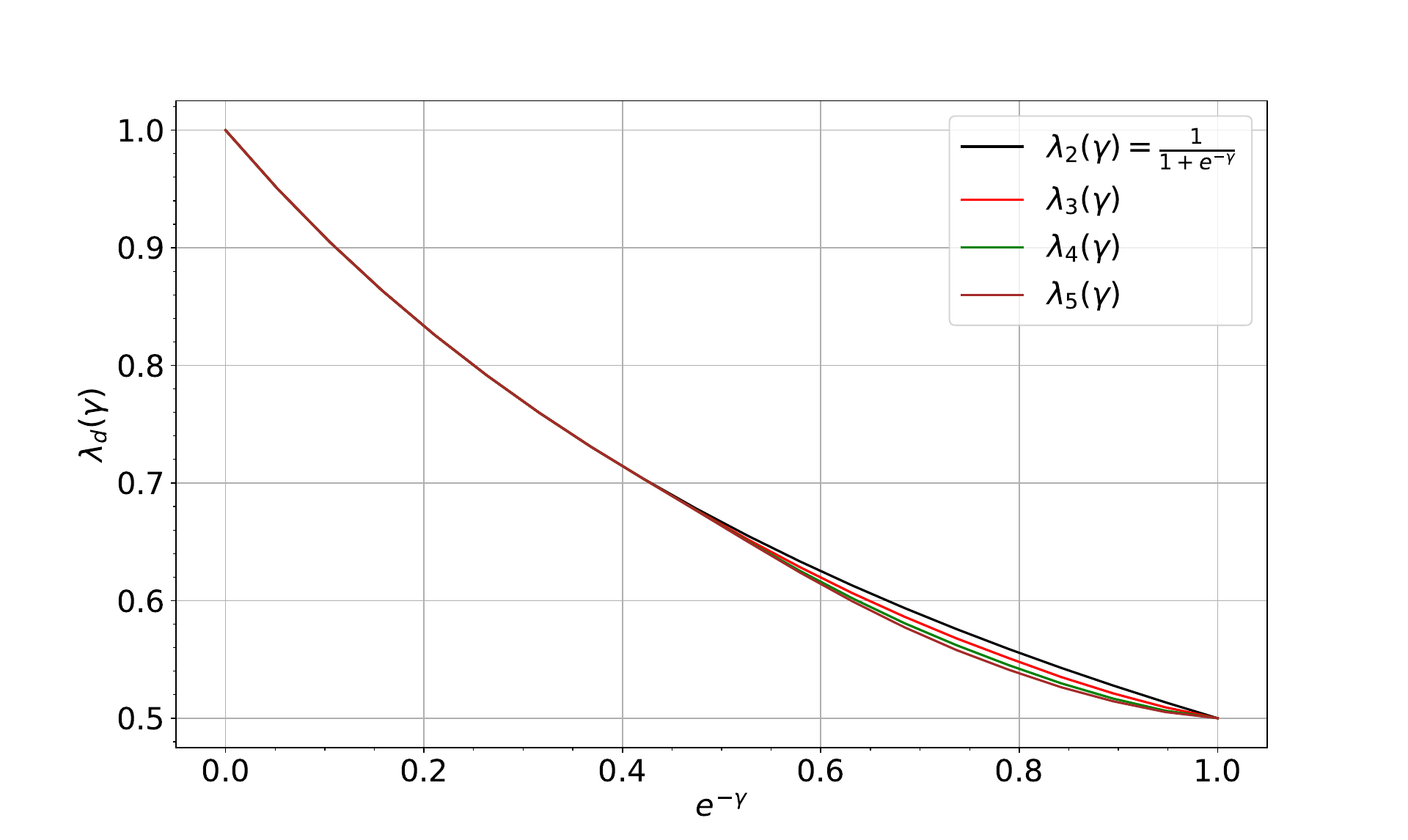}
	\caption{Each curve indicates necessary and sufficient conditions where the qu$d$it restriction $\pazocal{N}_{\lambda,\gamma}^{(d)}$ of the bosonic loss-dephasing channel $\pazocal{N}_{\lambda,\gamma}$ is anti-degradable. In the region above the curve $\lambda_5(\gamma)$, the bosonic loss-dephasing channel is never anti-degradable.}
 \label{fig:sdp_solution}
\end{figure}
The results are plotted with respect to $e^{-\gamma}$ for various values of $d$ in Fig.~\ref{fig:sdp_solution}, showcasing the dependence of $\lambda_d(\gamma)$ on $\gamma$ for small values of $d$. Our numerical analysis reveals that when $\gamma$ is sufficiently large, i.e.~$e^{-\gamma}\lesssim 0.41$ or $\gamma\gtrsim 0.89$, the value of $\lambda_d(\gamma)$ consistently equals $\frac{1}{1+e^{-\gamma}}$ for all examined values of $d$. In particular, this seems to suggest that within this range of the dephasing parameter, $\pazocal{N}_{\lambda,\gamma}$ is anti-degradable if and only if $\lambda\le \frac{1}{1+e^{-\gamma}}$. Based on this numerical exploration, we propose the following conjecture:
\begin{cj}
If $\gamma$ is sufficiently large ($\gamma\gtrsim 0.89$), then the bosonic loss-dephasing channel $\pazocal{N}_{\lambda,\gamma}$ is anti-degradable if and only if $\lambda\le \frac{1}{1+e^{-\gamma}}$.
\end{cj}
Notably, from Fig.~\ref{fig:sdp_solution} we observe that if $\lambda$ and $\gamma$ satisfy $\lambda=\frac{1}{1+e^{-\gamma}}$ with $e^{-\gamma}\gtrsim 0.41$ (or $\gamma\lesssim 0.89$), then $\pazocal{N}_{\lambda,\gamma}$ is not anti-degradable.

\section{Generalisation of our methods to general bosonic dephasing channels}\label{envision_gen}
In Theorem~\ref{main-result} we introduced a method to analyse anti-degradability of the bosonic loss-dephasing channel. In this section, we show that this method can be applied also to analyse the anti-degradability of the composition between a \emph{general} bosonic dephasing channel and the pure-loss channel channel.

Given a probability distribution $p(\cdot)$ over $\mathbb{R}$, the associated \emph{general bosonic dephasing channel} is given by
\bb
    \pazocal{D}^{(p)}(X)\coloneqq  \int_{-\infty}^{\infty}\mathrm{d}\phi\, p(\phi)\,e^{i\phi \hat{a}^\dagger \hat{a}}\,X\, e^{-i\phi \hat{a}^\dagger \hat{a}}\,.
\ee
If $p(\phi)$ is the Gaussian distribution $p(\phi)\coloneqq  \frac{1}{\sqrt{2\pi\gamma}}  e^{-\frac{\phi^2}{2\gamma}}$, the general bosonic dephasing channel $\pazocal{D}^{(p)}$ exactly coincides with the bosonic dephasing channel $\pazocal{D}_\gamma$ analysed in this work. The action of $\pazocal{D}^{(p)}$ on operators of the form $\ketbra{n}{m}$ is given by
\bb
    \pazocal{D}^{(p)}(\ketbra{n}{m})=\tilde{p}(n-m)\ketbra{n}{m}\,,
\ee
where $\tilde{p}$ is the Fourier transform of the probability distribution $p$, i.e.
\bb
    \tilde{p}(k)\coloneqq \int_{-\infty}^{\infty}\mathrm{d}\phi\, p(\phi)\,e^{ik\phi}\,.
\ee
Let $\NN_\lambda^{(p)}$ be the composition between such a general bosonic dephasing channel $\pazocal{D}^{(p)}$ and the pure-loss channel, i.e.
\bb\label{def_arb_deph_loss}
    \NN_\lambda^{(p)} \coloneqq \pazocal{D}^{(p)}\circ \pazocal{E}_\lambda=\pazocal{E}_\lambda\circ \pazocal{D}^{(p)}\,.
\ee
We can apply the exact same method that we have introduced in the proof of Thereom~\ref{main-result} in order to analyse the anti-degradability of $\NN_\lambda^{(p)}$. The key observation is that in the proof of Theorem~\ref{main-result} we did not use the explicit expression of the channel $\pazocal{D}^{(p)}$ before stating \eqref{state_sym_ext}. This simple observation allows us to generalise our results to arbitrary bosonic dephasing channels, as stated in the following theorem.
\begin{theorem}[(Sufficient condition on the anti-degradability of the composition between a general bosonic dephasing channel and pure-loss channel)]Let $\lambda\in[0,1)$ and let $p(\cdot)$ be a probability distribution over $\mathbb{R}$. Let $A = (a_{mn})_{m,n \in \mathbb{N}}$ be the infinite matrix whose components are defined by
\bb
    a_{mn}\coloneqq\frac{\tilde{\phi}(n-m)}{\sum_{j=0}^{\min(n,m)} \sqrt{\,\pazocal{B}_j\!\left(n,\frac{2\lambda-1}{\lambda}\right)\,\pazocal{B}_j\!\left(m,\frac{2\lambda-1}{\lambda}\right)}},\quad\forall\,n,m\in\N\,.
\ee
The channel $\pazocal{N}_{\lambda}^{(p)}$ is anti-degradable as long as either $\lambda\in[0,\frac{1}{2}]$ or the infinite matrix $A$ is positive semi-definite.
\end{theorem}

\section{Coherence preservation of the bosonic loss-dephasing channel}\label{section_coherence_preserv}
In this section we use the notation introduced in Section~\ref{subsec_eb}.
\begin{theorem}\label{thm_two_capacities}
    Let $\gamma\ge0$ and $\lambda\in(0,1]$. For any energy constraint $N_s>0$, the energy-constrained two-way quantum and secret-key capacities of the bosonic loss-dephasing channel 
    are strictly positive, $ K(\pazocal{N}_{\lambda,\gamma},N_s)\ge Q_2(\pazocal{N}_{\lambda,\gamma},N_s)>0$. 
\end{theorem}
\begin{proof}
    We begin by assuming $N_s\in(0,1)$ and defining the two-mode state 
    \bb
    \ket{\Psi_{N_s}}_{AA'}\coloneqq \sqrt{1-N_s}\ket{00}_{AA'}+\sqrt{N_s}\ket{11}_{AA'}\,,
    \ee
    where the mean photon number of $A'$ system is equal to $N_s$. By exploiting Lemma~\ref{App_remark_commutation}, one can observe that the state
    \bb\label{state_2_qubit}
    \rho_{AB}\coloneqq \Id_{A}\otimes\pazocal{N}_{\lambda,\gamma}(\ketbra{\Psi_{N_s}}_{AA'})
    \ee
    is effectively a two-qubit state and its matrix with respect to the computational basis $\{\ket{00}, \ket{01}, \ket{10}, \ket{11}\}$ is given by: 
    \begin{align*}
         \rho_{AB}=&\left(\begin{matrix} 1-N_s\quad & 0& 0& \sqrt{(1-N_s)N_se^{-\gamma}\lambda}\\ 0 & 0\quad & 0\quad& 0\\
         0 & 0\quad & (1-\lambda)N_s\quad & 0\\ \sqrt{(1-N_s)N_se^{-\gamma}\lambda} & 0 & 0 &\lambda N_s\quad \end{matrix}\right) .
    \end{align*}
   If we perform partial transpose with respect to the system $B$, we find the matrix
        \begin{align*}
             (\rho_{AB})^{\intercal_B}=&\left(\begin{matrix} 1-N_s\quad & 0& 0& 0\\ 0 & 0\quad & \sqrt{(1-N_s)N_se^{-\gamma}\lambda}\quad& 0\\
         0 & \sqrt{(1-N_s)N_se^{-\gamma}\lambda}\quad & (1-\lambda)N_s\quad & 0\\ 0 & 0 & 0 &\lambda N_s\quad \end{matrix}\right) \,,
                 \end{align*}
                     whose eigenvalues are not all positive, i.e.~the state $\rho_{AB}$ is not PPT~\cite{2-qubit-distillation}. By exploiting the fact that any two-qubit state is distillable if and only if it is not PPT~\cite{2-qubit-distillation}, it follows that $ E_d\left(\Id_{A}\otimes\pazocal{N}_{\lambda,\gamma}(\ketbra{\Psi_{N_s}}_{AA'}) \right)>0\,$, where $E_d$ is the distillable entanglement. On the other hand, from Lemma~\ref{lemma_Ed_Q2_link} we have that 
                     \bb
                        K(\pazocal{N}_{\lambda,\gamma},N_s)\ge Q_2(\pazocal{N}_{\lambda,\gamma},N_s)\ge E_d\!\left(\Id_{A}\otimes\pazocal{N}_{\lambda,\gamma}(\ketbra{\Psi_{N_s}}_{AA'}) \right)\,.
                    \ee
                    This concludes the proof for $N_s\in(0,1)$. Since the energy-constrained capacities are monotonically non-decreasing in the energy constraint $N_s$, the proof follows for any $N_s>0$. 
\end{proof}
Since the state $\Id_{A} \otimes \pazocal{N}_{\lambda,\gamma}(\ketbra{\Psi_{N_s}}_{AA'})$ in~\eqref{state_2_qubit} is always entangled, it follows that the bosonic loss-dephasing channel $\pazocal{N}_{\lambda,\gamma}$ is never entanglement breaking~\cite{MARK}. We state this formally in the following theorem.

\begin{theorem}\label{thm_eb_property}
For all $\gamma\ge0$ and all $\lambda\in(0,1]$, the bosonic loss-dephasing channel $\pazocal{N}_{\lambda,\gamma}$ is not entanglement breaking.
\end{theorem}

In the following subsection we will find an explicit strictly positive lower bound on the two-way capacities of the bosonic loss-dephasing channel.

\subsection{Multi-rail multi-photon encoding}
Let $\Pi(N,k)$ be the set of partitions of $N$ objects into $k$ (possibly empty) parts. It is well known that $|\Pi(N,k)| = \binom{N+k-1}{k-1} = \binom{N+k-1}{N}$. Clearly, we can think of each $p\in \Pi(N,k)$ as a vector in $\N_+^k$, also denoted by $p$, with the constraint that $\sum_{\ell=1}^k p_\ell = N$. For each $p\in \Pi(N,k)$, define the associate $k$-mode Fock state as
\bb
\ket{\psi_p} \coloneqq \ket{p_1}\ldots \ket{p_k}\, .
\ee
Note that for $p,q\in \Pi(N,k)$, we have that $\braket{\psi_p}{\psi_q}=\delta_{p,q}$. Let us call
\bb
P_{N,k} \coloneqq \sum_{p\in \Pi(N,k)} \psi_p
\label{projector_N_k}
\ee
the projector onto the $k$-mode subspace of total photon number $N$. Note that the bosonic dephasing channel satisfies that
\bb
\pazocal{D}_\gamma^{\otimes k}\left(\ketbra{\psi_p}{\psi_q}\right) = e^{-\frac{\gamma}{2}\sum_\ell (p_\ell - q_\ell)^2} \ketbra{\psi_p}{\psi_q} = e^{-\frac{\gamma}{2}\|p-q\|^2} \ketbra{\psi_p}{\psi_q} = (K_{N,k,\gamma})_{pq} \ketbra{\psi_p}{\psi_q}\, ,
\ee
where $K_{N,k,\gamma}$ is the $\binom{N+k-1}{N}\times \binom{N+k-1}{N}$ matrix with entries
\bb
(K_{N,k,\gamma})_{pq} \coloneqq e^{-\frac{\gamma}{2}\|p-q\|^2} .
\ee
For any $\binom{N+k-1}{N}$-dimensional state $\sigma$, let us denote as $\overline{\sigma}$ the following isometrically equivalent state
\bb
    \overline{\sigma} \coloneqq \sum_{p,q\in \Pi(N,k)} \sigma_{pq} \ketbra{\psi_p}{\psi_q}\, .
    \label{rail_encoding}
\ee
The state $\overline{\sigma}$, which is termed as the \emph{rail encoding} of $\sigma$, is supported on the subspace of $k$ modes with total photon number equal to $N$. Now, let $\rho$ be a $\binom{N+k-1}{N}$-dimensional state and let us calculate the output of $\pazocal{D}_\gamma^{\otimes k}$ when the input is $\overline{\rho}$:
\bb
\pazocal{D}_\gamma^{\otimes k}\left( \overline{\rho} \right) = \sum_{p,q\in \Pi(N,k)} \rho_{pq} (K_{N,k,\gamma})_{pq} \ketbra{\psi_p}{\psi_q} = \overline{K_{N,k,\gamma} \circ \rho} = \overline{\Theta_{N,k,\gamma}(\rho)}\, ,
\ee
where the operation $\circ$ denotes the element-wise product between matrices and where we have introduced the following Hadamard channel:
\bb
\Theta_{N,k,\gamma}(X) \coloneqq K_{N,k,\gamma}\circ X\, ,
\ee
Since $\pazocal{D}_\gamma^{\otimes k}$ is a (completely) positive map, this in particular shows that $K_{N,k,\gamma} \geq 0$. (This latter statement can also be proved directly with techniques similar to that in the proof of~\cite[Lemma~15]{G-dilatable}.) In practice, the $k$-fold application of the bosonic dephasing channel on the $k$-mode $N$-photon code space behaves as a new Hadamard channel $\Theta_{N,k,\gamma}$ with associated matrix $K_{N,k,\gamma}$.

\subsubsection{Lower bound on the two-way capacity of the bosonic loss-dephasing channel}
 Since under the action of $\NN_{\lambda,\gamma}$ photons can only be lost and never added, and each photon has a probability $\lambda$ of being transmitted, the probability that an $N$-photon state will retain $N$ photons at the output of the channel is exactly $\lambda^N$. \emph{If that happens}, then the state in the code space is effectively left untouched by the loss and only dephased under the action of the Hadamard channel $\Theta_{N,k,\gamma}$.

More formally, from the Kraus representation
\bb
\pazocal{E}_\lambda (X) = \sum_{n=0}^\infty \frac{1}{n!}\, (1-\lambda)^n \lambda^{\frac{a^\dag a}{2}} a^n X (a^\dag)^n \lambda^{\frac{a^\dag a}{2}}
\ee
it is easy to deduce the handy identity
\bb
\pazocal{E}_\lambda^{\otimes k}\big(\overline{\rho}\big) = \lambda^N \overline{\rho} + \left(1-\lambda^N\right) \delta_{N,k,\lambda}\, ,
\label{action_loss_code_space}
\ee
valid for all $\binom{N+k-1}{N}$-dimensional states $\rho$, with the notation of~\eqref{rail_encoding}. 
Here, $\delta_{N,k,\lambda}$ is a suitable $k$-mode state supported on the subspace of total photon number at most $N-1$, and thus $\overline{\rho} \delta_{N,k,\lambda} = \delta_{N,k,\lambda} \overline{\rho} = 0$. In turn, the above identity implies that
\bb
\NN_{\lambda,\gamma}^{\otimes k}\big(\overline{\rho}\big) = \lambda^N\, \overline{K_{N,k,\gamma} \circ \rho} + \left(1-\lambda^N\right)\delta'_{N,k,\lambda,\gamma} = \lambda^N\, \overline{\Theta_{N,k,\gamma} (\rho)} + \left(1-\lambda^N\right)\delta'_{N,k,\lambda,\gamma} \, ,
\ee
where once again $\delta'_{N,k,\lambda,\gamma}$ is a suitable $k$-mode state supported on the subspace of total photon number at most $N-1$.

Therefore, we can use the channel $\NN_{\lambda,\gamma}^{\otimes k}$ to simulate $\Theta_{N,k,\gamma}$ \emph{probabilistically}, with probability $\lambda^N$. The simulation works as follows:
\begin{enumerate}[(i)]
\item The input state $\rho$ is encoded in the $k$-mode $N$-photon subspace according to the mapping $\rho \mapsto \overline{\rho}$.
\item The $k$-mode state $\overline{\rho}$ is sent across $\NN_{\lambda,\gamma}^{\otimes k}$, via $k$ uses of the bosonic loss-dephasing channel.
\item The total photon number is measured at the output. If $N$ photons are found then the simulation is successful, otherwise the protocol is aborted.
\end{enumerate}

A wealth of operational resource inequalities can be deduced from the above considerations. Here we limit ourselves to the observation that the two-way quantum capacity must satisfy
\bb
Q_2(\NN_{\lambda,\gamma}) &\textgeq{(i)} \frac{\lambda^N}{k}\, Q_2(\Theta_{N,k,\gamma})\,.
\ee
Consequently, it holds that
\bb
Q_2(\NN_{\lambda,\gamma})  &\ge \frac{\lambda^N}{k}\, Q_2(\Theta_{N,k,\gamma}) \\&\textgeq{(i)} \frac{\lambda^N}{k}\, I_{\text{coh}}\left(\Id\otimes\Theta_{N,k,\gamma}(\ketbra{\Psi}{\Psi})\right) \\
&\texteq{(iii)} \frac{\lambda^N}{k}\left[\log_2 \binom{N+k-1}{N} - S\left(\binom{N+k-1}{N}^{-1} K_{N,k,\gamma} \right) \right]\,. \\
\ee
Here, in (ii), we used the fact that the two-way quantum capacity of a channel can be lower bounded in terms of the coherent information~\cite{MARK,Sumeet_book} and we introduced the two-qu$d$it maximally entangled state $\ket{\Psi}$ of dimension $d=\binom{N+k-1}{N}$. In (iii), we used the definition of coherent information $I_{\text{coh}}(\rho_{AB})\coloneqq S(\rho_{B})-S(\rho_{AB})$, with $S(\cdot)$ being the von Neumann entropy, and the fact that
\bb
    \Id\otimes\Theta_{N,k,\gamma}(\ketbra{\Psi}{\Psi})&=\frac{1}{{\binom{N+k-1}{N}}}\sum_{p,q\in\Pi(N,k)}\ketbra{p}{q}\otimes\Theta_{N,k,\gamma}(\ketbra{p}{q})\\
    &=\frac{1}{{\binom{N+k-1}{N}}} \sum_{p,q\in\Pi(N,k)} (K_{N,k,\gamma})_{pq}\ketbra{p}{q}\otimes \ketbra{p}{q} \,,
\ee
which implies that the spectrum of $\Id\otimes\Theta_{N,k,\gamma}(\ketbra{\Psi}{\Psi})$ coincides with the spectrum of the matrix $\binom{N+k-1}{N}^{-1}K_{N,k,\gamma}$.
Consequently, we have that
\bb
Q_2(\NN_{\lambda,\gamma})  \ge \max_{N,k\in\N_+}\frac{\lambda^N}{k}\left[\log_2 \binom{N+k-1}{N} - S\left(\binom{N+k-1}{N}^{-1} K_{N,k,\gamma} \right) \right]\,. 
\ee
One can obtain a lower bound on the energy-constrained two-way quantum capacity $Q_2(\NN_{\lambda,\gamma},N_s)$ by restricting the optimisation to the values of $N$ and $k$ such that $\frac{N}{k}\le N_s$. Indeed, note that the rail-encoded state $\bar{\rho}$ satisfies the energy constraint as its mean photon number per mode is $\frac{N}{k}$. In formula, we have that
\bb
    Q_2(\NN_{\lambda,\gamma}, N_s)  \ge \max_{N,k\in\N_+:\, \frac{N}{k}\le N_s}\frac{\lambda^N}{k}\left[\log_2 \binom{N+k-1}{N} - S\left(\binom{N+k-1}{N}^{-1} K_{N,k,\gamma} \right) \right]\,. 
\ee
Note that $\log_2 \binom{N+k-1}{N} - S\left(\binom{N+k-1}{N}^{-1} K_{N,k,\gamma} \right)$ is always positive because $\binom{N+k-1}{N}^{-1} K_{N,k,\gamma}$ is a $\binom{N+k-1}{N}$-dimensional, non-maximally mixed, state and thus its von Neumann entropy is strictly smaller than by $\log_2 \binom{N+k-1}{N}$.
Consequently, we have the following theorem.
\begin{theorem}\label{thm_two_capacities_explicit}
    Let $\gamma\ge0$ and $\lambda\in(0,1]$. For any energy constraint $N_s>0$, the energy-constrained two-way quantum and secret-key capacities of the bosonic loss-dephasing channel 
    are lower bounded by
    \bb
        K(\pazocal{N}_{\lambda,\gamma},N_s)\ge Q_2(\pazocal{N}_{\lambda,\gamma},N_s)\ge \max_{\substack{N,k\in\N_+\\ \frac{N}{k}\le N_s}}\frac{\lambda^N}{k}\left[\log_2 \binom{N+k-1}{N} - S\left(\rho_{N,k,\gamma} \right) \right]>0\,.
    \ee
    Here, $S(\cdot)$ is the von Neumann entropy, $\rho_{N,k,\gamma}$ is a $\binom{N+k-1}{N}$-dimensional state defined by
    \bb
        \rho_{N,k,\gamma} &\coloneqq \binom{N+k-1}{N}^{-1}\sum_{p,q\in \Pi(N,k)}e^{-\frac{\gamma}{2}\|p-q\|_2^2}\ketbra{p}{q}\, ,
    \ee
    where $\Pi(N,k) \coloneqq\big\{p\in\N^k:\, \sumno_{i=1}^k p_i=N\big\}$ represents the set of partitions of a set of $N$ elements into $k$ parts, and the vectors $\{\ket{p}\}_{p\in\Pi(N,k)}$ are orthonormal.
    In particular,  
    \bb
        K(\pazocal{N}_{\lambda,\gamma},N_s)\ge Q_2(\pazocal{N}_{\lambda,\gamma},N_s)> \max_{N,k\in\N_+}\frac{\lambda^N}{k}\left[\log_2 \binom{N+k-1}{N} - S\left(\rho_{N,k,\gamma} \right) \right]>0\,.
    \ee

\end{theorem}

\section{Technical lemmas}

\subsection{Hadamard maps}\label{sec_hadamard_chann}
Given an infinite matrix $A = (a_{mn})_{m,n\in\mathbb{N}}, a_{mn}\in\mathbb{C}$, we can introduce a superoperator $H$, recognised as the \emph{Hadamard map}, whose action is defined as $H(\ketbra{m}{n})=a_{n,m}\ketbra{m}{n}$ for all $n,m\in\mathbb{N}$. We are interested in establishing requirements for an infinite matrix $A$ to ensure that the associated Hadamard map $H$ is a quantum channel. We begin with some preliminaries. Let $\ell^2(\mathbb{N})$ be the space of square-summable complex-valued sequences defined as
\bb\label{def_l2N}
    \ell^2(\mathbb{N})\coloneqq\left\{  x\coloneqq\{x_n\}_{n\in\mathbb{N}},x_n\in\mathbb{C}:\,\|x\|\coloneqq\sqrt{\sum_{n=0}^\infty|x_n|^2}<\infty \,\right\}\,.
\ee
An infinite matrix $A\coloneqq (a_{mn})_{m,n\in\mathbb{N}}, a_{mn}\in\mathbb{C}$ defines a linear operator on $\ell^2(\mathbb{N})$. The operator norm of $A$ is defined as follows:
\begin{align*}
    \| A \|_\infty \coloneqq \sup_{\substack{x\in\ell^2(\mathbb{N})\\  \|x\|=1}} \|Ax\|= \sup_{\substack{\{x_n\}_{n\in\mathbb{N}},x_n\in\mathbb{C}\\ \ \sum_{n=0}^\infty|x_n|^2=1}} \sqrt{ \sum_{m=0}^\infty \left| \sum_{n=0}^\infty a_{mn}x_n \right|^2}.
\end{align*}
$A$ is said to be bounded if $\| A \|_\infty<\infty$. The following lemma, referred to as \emph{Schur test}, gives a sufficent condition for an infinite matrix to be bounded (e.g.~\cite[Page 24, Problem 45]{Halmos_book}).
\begin{lemma}\label{shur_test}
 Let $A\coloneqq (a_{mn})_{m,n\in\mathbb{N}}, a_{mn}\in\mathbb{C}$, be an infinite matrix. Suppose that there exist $\{p_n\}_{n\in\mathbb{N}},p_n\in\mathbb{R}_{>0}$ and $\{q_m\}_{m\in\mathbb{N}},q_m\in\mathbb{R}_{>0}$, and $\beta > 0$, and $\gamma > 0$ such that
 \begin{align*}
    \sum_{m=0}^\infty |a_{mn}|{p_m} \leq \beta q_n \quad \text{and} \quad \sum_{n=0}^\infty |a_{mn}|{q_n} \leq \gamma p_m\,, \quad\forall\, m, n\in\mathbb{N}\,. 
 \end{align*}
Then the matrix $A$ satisfies $\|A\|_\infty \leq \beta \gamma$. In particular, $A$ is bounded.
\end{lemma}
By choosing $p_n=q_n=1$ and $\gamma=\beta=\sup_{n\in\N}\sum_{m=0}^\infty |a_{mn}|$, we obtain the following corollary:
\begin{cor}\label{shur_test_cons}
 Let $A = (a_{mn})_{m,n \in \mathbb{N}},a_{mn}\in \mathbb{C}$, be an infinite Hermitian matrix. If $\sup_{n\in\N}\sum_{m=0}^\infty |a_{mn}|$ is finite, then $A$ is bounded.
\end{cor}
\begin{lemma}\label{lemma_mat_infinita}
Let $A = (a_{mn})_{m,n \in \mathbb{N}},a_{mn}\in \mathbb{C}$, be a bounded Hermitian infinite matrix. Then $A$ is positive semi-definite as an operator on $\ell^2(\N)$ if and only if $A^{(d)}\coloneqq (a_{mn})_{m,n=0,1,\ldots,d-1} $ is positive semi-definite for all $d \in \mathbb{N}$, where $A^{(d)}$ is the $d \times d$ top left corner of $A$. 
\end{lemma}
\begin{proof}
Assume that $A^{(d)}$ is positive semi-definite for all $d \in \mathbb{N}$. Let us pick an arbitrary $x\in\ell^2(\mathbb{N})$. It is known that for any $\varepsilon>0$, there exists $d\in\N$ and $y^{(d)}\coloneqq(y^{(d)}_n)_{n\in\N},y^{(d)}_n\in\mathbb{C}$, with $y^{(d)}_n=0$ for all $n> d$, such that $ \|x-y^{(d)}\|<\varepsilon\,$. Note that
\begin{align*}
    x^\dagger A x&=\left(x-y^{(d)}\right)^\dagger A\, x   +  \left(y^{(d)}\right)^\dagger A\, (x-y^{(d)})+ \left(y^{(d)}\right)^\dagger A y^{(d)}
    \\&\textgeq{(i)} -\|x-y^{(d)}\|\,\|A\|_\infty\, \left(\|x\|+\|y^{(d)}\|\right) + \left(y^{(d)}\right)^\dagger A^{(d)} \,y^{(d)}
    \\&\textgeq{(ii)} -\varepsilon\,\|A\|_\infty\, \left(2\|x\|+\varepsilon\right) \,,
\end{align*}
where in (i) we applied Cauchy-Schwarz inequality twice and the definition of infinity norm as follows
\begin{align*}
    |(x-y^{(d)})^\dagger A\, x|&\le \|x-y^{(d)}\| \|A\, x \|\le \|x-y^{(d)}\| \|A\|_\infty \|x\|\,,\\
    |(y^{(d)})^\dagger A\, (x-y^{(d)})|&\le \|y^{(d)}\| \|A\, (x-y^{(d)})\|\le  \|y^{(d)}\| \|A\|_\infty \|x-y^{(d)}\|\,,
\end{align*}
and, in (ii), we exploited triangular inequality to derive
\bb
\|y^{(d)}\|\le \|y^{(d)}-x\| +\|x\|\le \varepsilon+\|x\|\,,
\ee
together with the fact that $A^{(d)}$ is positive semi-definite. Hence, since $\varepsilon>0$ is arbitrary, we conclude that $x^\dagger A x\ge0$, meaning that $A$ is positive semi-definite as an operator on $\ell^2(\N)$. 
\end{proof}
A square matrix is said to be diagonally dominant if
  \begin{align*}
      \sum_{m\ne n}|a_{mn}|\le |a_{nn}|,\quad\forall n.
  \end{align*} 
  In words, a square matrix is said to be diagonally dominant if for every row of the matrix, the absolute value of the diagonal entry in a row is larger than or equal to the sum of the  absolute values of all the other (non-diagonal) entries in that row. Note that for Hermitian matrices one can exchange row with column in this definition. We proceed to state a sufficient condition for a matrix $A$ to be positive semi-definite as an operator on $\ell^2(\N)$. While the following sufficient condition is usually stated for finite matrices, it can be generalised to infinite matrices as per Lemma~\ref{lemma_mat_infinita}.
\begin{lemma}\cite[Chapter 6]{HJ1}\label{lemma_diag_dom_finite}
For any $d\in\mathbb{N}$, let $A=(a_{mn})_{m,n=0,1,\ldots,d-1}, a_{mn}\in\mathbb{C}$, be a $d\times d$ Hermitian matrix with $a_{nn}\in\mathbb{R}_{\geq0}$ for all $n\in\{0,\ldots,d-1\}$. $A$ is positive semi-definite if it is diagonally dominant.
\end{lemma}

\begin{lemma}\label{lemma_diag_dom_infinite}
Let $A = (a_{mn})_{m,n \in \mathbb{N}},a_{mn}\in \mathbb{C}$ be an infinite Hermitian matrix with $a_{nn}\in\mathbb{R}_{\geq0},\forall n\in\mathbb{N}$. Assume that $\sup_{n\in\N} a_{nn}$ is finite and that $A$ is diagonally dominant. Then $A$ is bounded and positive semi-definite when seen as an operator on $\ell^2(\N)$.   
\end{lemma}

\begin{proof}
Since
\begin{align*}
\sup_{n\in\N}\sum_{m=0}^\infty |a_{mn}|\le 2\sup_{n\in\N} a_{n,n}<\infty\,,
\end{align*}
Corollary~\ref{shur_test_cons} implies that $A$ is bounded. The fact that $A$ is positive semi-definite as an operator on $\ell^2(\N)$ follows from Lemma~\ref{lemma_diag_dom_finite} together with Lemma~\ref{lemma_mat_infinita}.
\end{proof}
We now present a lemma from the literature that establishes necessary and sufficient conditions for an infinite matrix $A$ to ensure that the associated Hadamard map $H$ is a quantum channel. We then use it to derive an explicit sufficient condition for an infinite matrix $A$ to give rise to a cptp Hadamard map. 
\begin{lemma}\cite[Lemma S4]{exact-solution}
\label{lemma_hadamard_channel_sm}
Let $A\coloneqq (a_{mn})_{m,n\in\mathbb{N}}, a_{mn}\in\mathbb{C}$ be a bounded infinite matrix.
The following requirements establish the necessary and sufficient conditions for the associated Hadamard map $H$ to qualify as a quantum channel:
\begin{enumerate}[label=(\roman*)]
    \item  $a_{nn}=1,\,\forall n\in\mathbb{N}$;
    \item  $A$ is positive semi-definite as on operator on $\ell^2(\mathbb{N})$.
\end{enumerate}
\end{lemma}
As a consequence of Lemma~\ref{lemma_hadamard_channel_sm} and Lemma~\ref{lemma_diag_dom_infinite}, we obtain: 
\begin{lemma}\label{diag_dom_implies_hadamard}
    Let $A\coloneqq (a_{mn})_{m,n\in\mathbb{N}}$ be an infinite Hermitian matrix that is diagonally dominant with $a_{n,n}=1$ for all $n\in\N$. 
    In this case, its associated Hadamard map $H$ is a quantum channel.
\end{lemma}

\subsection{Miscellaneous Lemmas}

\begin{lemma}\label{lemmino}
    Let $\pazocal{H}_A$ and $\HH_{B}$ be two Hilbert spaces and let $\pazocal{N}:\pazocal{T}(\HH_A)\to\pazocal{T}(\HH_{B})$ be a quantum channel. For all normalised states $\ket{\psi_1},\ket{\psi_2}\in\HH_A$ and $\ket{\phi_1},\ket{\phi_2}\in\HH_B$ it holds that 
    \bb
        \left|\bra{\phi_1}\pazocal{N}(\ketbra{\psi_1}{\psi_2})\ket{\phi_2}\right|\le 1\,.
    \ee
\end{lemma}
\begin{proof}
    It holds that 
    \bb
        \left|\bra{\phi_1}\pazocal{N}(\ketbra{\psi_1}{\psi_2})\ket{\phi_2}\right|&\textleq{(i)} \|\pazocal{N}(\ketbra{\psi_1}{\psi_2})\|_\infty
        \\&\textleq{(ii)}  \|\pazocal{N}(\ketbra{\psi_1}{\psi_2})\|_1
        \\&\textleq{(iii)} \|\ketbra{\psi_1}{\psi_2}\|_1
        \\&=1\,.
    \ee
    Here, in (i), we exploited one of the definition of the operator norm in~\eqref{def2_operator_norm}.
    In (ii), we exploited that the trace norm is always an upper bound on the operator norm. Finally, in (iii), we leveraged the monotonicity of the trace norm under quantum channels~\cite{NC}.
\end{proof}

\begin{lemma}[\cite{Holevo-CJ,Holevo-CJ-arXiv}]\label{gen_choi_thm}
Let $\HH_A,\HH_{A'}$ be isomorphic Hilbert spaces, possibly infinite dimensional.
Let $\ket{\psi}_{A'A}$ be a pure state that satisfies $\Tr_{A'}\!\left[\ketbra{\psi}_{AA'}\right]>0$. The generalised Choi--Jamio\l{}kowski matrix defines an isomorphism between the set of quantum channels from $\HH_{A'}$ to $\HH_B$ and the set of bipartite states $\sigma_{AB}\in\pazocal{P}\left(\HH_{AB}\right)$ such that $\Tr_{B}\sigma_{AB}=\Tr_{A'}\!\left[\ketbra{\psi}_{AA'}\right]$. Specifically, for any quantum channel $\pazocal{N}_{A'\to B}: \pazocal{T}(\HH_{A'})\to\pazocal{T}(\HH_{B})$, it holds that
\bb\label{generalised_choi_eq_def}
\pazocal{N}_{A'\to B}\!\left(\ketbra{e_i}{e_j}\right) = \frac{1}{\sqrt{\lambda_i \lambda_j}}\Tr_{A}\!\left[\left(\ketbra{e_j}{e_i}_{A}\otimes\mathbb{1}_{B}\right)\,\sigma_{AB}\right],\quad\forall\,i,j\in\N\,,
\ee
where $(\ket{e_i})_{i\in\mathbb{N}}$ and $(\lambda_i)_{i\in\mathbb{N}}$ form a spectral decomposition of $\Tr_{A'}\!\left[\ketbra{\psi}_{AA'}\right]$, i.e.~$\Tr_{A'}\!\left[\ketbra{\psi}_{AA'}\right]=\sum_{i}\lambda_i \ketbra{e_i}_A\,$, and where the state $\sigma_{AB}\coloneqq\Id_{A}\otimes\pazocal{N}_{A'\to B}(\ketbra{\psi}_{AA'})$ is called the generalised Choi state of $\pazocal{N}$. Eq.~\ref{generalised_choi_eq_def} is enough to specify the channel $\pazocal{N}_{A'\to B}$ completely, as the linear span of the operators $\left(\ketbra{e_i}{e_j}\right)_{i,j\in\N}$ (i.e.~the set of finite-rank operators) is dense in $\pazocal{T}(\HH_{A'})$.
\end{lemma}


\begin{lemma}[\cite{Myhr2009}]\label{lemma_infinite_extendibility}
Let $\HH_A,\HH_{A'}, \HH_{B}$ be isomorphic Hilbert spaces, possibly infinite dimensional. Let $\pazocal{N}_{A'\to B}:\pazocal{T}(\HH_{A'})\to \pazocal{T}(\HH_B)$ be a quantum channel. Let $\ket{\psi}_{A'A}\in\HH_{A}\otimes\HH_{A'}$ be a pure state such that the reduced state $\Tr_{A'}[\ketbra{\psi}_{AA'}]$ is positive definite. Then, $\pazocal{N}_{A'\to B}$ is anti-degradable if and only if the state $\Id_{A}\otimes\pazocal{N}_{A'\to B}(\ketbra{\psi}_{AA'})$ is two-extendible on $B$, meaning that there exists a state $\rho_{AB_1B_2}\in\pazocal{P}\left(\HH_A\otimes\HH_{B_1}\otimes\HH_{B_2}\right)$ such that
\bb\label{eq_extend}
    \Tr_{B_2}\left[\rho_{AB_1B_2} \right]&= \Id_{A}\otimes\pazocal{N}_{A'\to B_1}(\ketbra{\psi}_{AA'})\,,\\
    \Tr_{B_1}\left[\rho_{AB_1B_2} \right]&= \Id_{A}\otimes\pazocal{N}_{A'\to B_2}(\ketbra{\psi}_{AA'}),
\ee
where $\HH_{B_1}$ and $\HH_{B_2}$ are Hilbert spaces that are isomorphic to $\HH_{B}$.
\end{lemma}
\begin{proof} 
Let $U_{A'E\to BE} $ be a Stinespring dilation of the channel $\pazocal{N}_{A'\to B}$. Further assume that $\pazocal{N}_{A'\to B}$ is anti-degradable. By definition, there exists a quantum channel $\pazocal{A}_{E\to B}$ such that $\pazocal{A}_{E\to B}\circ \pazocal{N}^{\text{c}}_{A'\to B}=\pazocal{N}_{A'\to B}$. Let us consider the tripartite state $\rho_{AB_1B_2}$, with $B_1,B_2$ being copies of $B$, defined as
\bb\label{link_choi_antideg}
    \rho_{AB_1B_2}= \Id_{A}\otimes\Id_{B_1} \otimes \pazocal{A}_{E\to B_2}\left(U_{A'E\to B_1E}\,\big(\ketbra{\psi}_{AA'}\otimes\ketbra{0}_E\big)\,U_{A'E\to B_1E}^\dagger\right)\,.
\ee
It holds that 
\begin{align*}
    \Tr_{B_2}\left[\rho_{AB_1B_2} \right] &= \Tr_{B_2}\left[ \Id_{A}\otimes\Id_{B_1} \otimes\pazocal{A}_{E\to B_2}\left(U_{A'E\to B_1E} \,\big(\ketbra{\psi}_{AA'}\otimes\ketbra{0}_E\big)\, U_{A'E\to B_1E}^\dagger\right)  \right]\\& = \Id_A\otimes\Tr_{E}\!\left[U_{A'E\to B_1E} \big(\ketbra{\psi}_{AA'}\otimes\ketbra{0}_E\big) U_{A'E\to B_1E}^\dagger \right]\\&=\Id_A\otimes \pazocal{N}_{A'\to B_1}\left( \ketbra{\psi}_{AA'} \right)\,,
\end{align*}
and 
\begin{align*}
        \Tr_{B_1}\left[\rho_{AB_1B_2} \right] &= \Tr_{B_1}\left[ \Id_{A}\otimes\Id_{B_1} \otimes\pazocal{A}_{E\to B_2}\left(U_{A'E\to B_1E} \,\big(\ketbra{\psi}_{AA'}\otimes\ketbra{0}_E\big)\, U_{A'E\to B_1E}^\dagger\right)  \right]\\& = \Id_{A} \otimes\pazocal{A}_{E\to B_2}\left(\Tr_{B_1}\left[U_{A'E\to B_1E} \big(\ketbra{\psi}_{AA'}\otimes\ketbra{0}_E\big) U_{A'E\to B_1E}^\dagger \right]\right)\\&=\Id_A\otimes \pazocal{A}_{E\to B_2}\circ \pazocal{N}^\text{c}_{A'\to B_2}\left( \ketbra{\psi}_{AA'} \right)\\&=\Id_A\otimes \pazocal{N}_{A'\to B_2}\left( \ketbra{\psi}_{AA'} \right)\,.
\end{align*}
Now, let us establish the converse. Assume that there exists $\rho_{AB_1B_2}$ which satisfies~\eqref{eq_extend}. Let $\ket{\Psi}_{AB_1B_2P}\in\HH_A\otimes\HH_{B_1}\otimes\HH_{B_2}\otimes \HH_P$ be a purification of $\rho_{AB_1B_2}$, with $\HH_{P}$ being the purifying Hilbert space. Note that both $\ket{\Psi}_{AB_1B_2P}$ and $U_{A'E\to B_1E} \,\big(\ket{\psi}_{AA'}\otimes\ket{0}_E\big)$ are purifications of $\Id_A\otimes \pazocal{N}_{A'\to B_1}\left( \ketbra{\psi}_{AA'} \right)$, with $\HH_{B_2}\otimes\HH_{P}$ and $\HH_{E}$ being their purifying Hilbert spaces, respectively. It follows that~\cite{NC} there exists an isometry $V_{E\to B_2P}:\HH_{E}\to \HH_{B_2}\otimes \HH_P$ such that $ V_{E\to B_2P}\,U_{A'E\to B_1E} \,\big(\ket{\psi}_{AA'}\otimes\ket{0}_E\big) = \ket{\Psi}_{AB_1B_2P}\,.$
Hence, the quantum channel $\pazocal{A}_{E\to B_2}:\pazocal{T}(\HH_{E})\to \pazocal{T}(\HH_{B_2})$, defined by $\pazocal{A}_{E\to B_2}(\cdot)=\Tr_P\left[ V_{E\to B_2P} (\cdot)V_{E\to B_2P}^\dagger \right]$,
satisfies that
\begin{align*}
      \Id_{A}\otimes\pazocal{A}_{E\to B_2}\circ \pazocal{N}^\text{c}_{A'\to B_2}(\ketbra{\psi}_{AA'})&=\Id_{A}\otimes\pazocal{A}_{E\to B_2}\left( \Tr_{B_1}\left[ U_{A'E\to B_1E}   (\ketbra{\psi}_{AA'}\otimes\ketbra{0}_E) U_{A'E\to B_1E} ^\dagger \right] \right)\\&=\Tr_{B_1P}\left[\ketbra{\Psi}_{AB_1B_2P}\right]\\&=\Tr_{B_1}\left[ \rho_{AB_1B_2} \right]\\&=\Id_A\otimes \pazocal{N}_{A'\to B_2}\left( \ketbra{\psi}_{AA'} \right)\,.
\end{align*}
Consequently, since the pure state $\ket{\psi}_{A'A}$ satisfies $\Tr_{A'}\!\left[\ketbra{\psi}_{AA'}\right]>0$,  Lemma~\ref{gen_choi_thm} implies that $\pazocal{A}_{E\to B_2}\circ \pazocal{N}^\text{c}_{A'\to B_2}= \pazocal{N}_{A'\to B_2}$, meaning that $\pazocal{N}_{A'\to B_2}$ is anti-degradable.
\end{proof}

\begin{lemma}\label{lemma_comp_antideg}
    Let $\pazocal{N},\pazocal{M}:\pazocal{T}(\HH_S)\to\pazocal{T}(\HH_S)$ be quantum channels. If either $\pazocal{M}$ or $\pazocal{N}$ is anti-degradable, then the composition $\pazocal{M}\circ\pazocal{N}$ is anti-degradable. Specifically, let $E_1$ and $E_2$ be the Stinespring environments of $\pazocal{N}$ and $\pazocal{M}$, respectively. If $\pazocal{N}$ is anti-degradable with anti-degrading map $\pazocal{A}_{E_1\to S}$, then $(\pazocal{M}\circ\pazocal{A}_{E_1\to S})\otimes \Tr_{E_2}$ is an anti-degrading map of $\pazocal{M}\circ\pazocal{N}$. Analogously, if $\pazocal{M}$ is anti-degradable with anti-degrading map $\pazocal{A}_{E_2\to S}$, then $\Tr_{E_1}\otimes\pazocal{A}_{E_2\to S}$ is an anti-degrading map of $\pazocal{M}\circ\pazocal{N}$. 
 \end{lemma}
\begin{proof}
Let $V^{S\to SE_1}$ and $W^{S\to SE_2}$ be Stinespring isometries associated with $\pazocal{N}$ and $\pazocal{M}$, respectively. By considering the following complementary channel of $\pazocal{M}\circ\pazocal{N}$,
\begin{align*}
    (\pazocal{M}\circ\pazocal{N})^\text{c}(\rho)= \Tr_{S}\!\left[W^{S\to SE_2}V^{S\to SE_1}\,\rho\, \left(V^{S\to SE_1}\right)^\dagger \left(W^{S\to SE_2}\right)^\dagger\right],\qquad\forall\,\rho\in\pazocal{T}(\HH_S)  \,,
\end{align*}
one can easily check that if $\pazocal{N}$ is anti-degradable with anti-degrading map $\pazocal{A}_{E_1\to S}$, then 
\bb
\left[(\pazocal{M}\circ\pazocal{A}_{E_1\to S})\otimes \Tr_{E_2}\right]\circ (\pazocal{M}\circ\pazocal{N})^\text{c}=\pazocal{M}\circ\pazocal{N}\,.
\ee
Analogously, one can easily verify that if $\pazocal{M}$ is anti-degradable with anti-degrading map $\pazocal{A}_{E_2\to S}$, then 
\bb
\left[ \Tr_{E_1}\otimes\pazocal{A}_{E_2\to S} \right]\circ (\pazocal{M}\circ\pazocal{N})^\text{c}=\pazocal{M}\circ\pazocal{N}\,.
\ee
\end{proof}

\end{document}